\documentclass{article}
\usepackage{amsmath}
\usepackage{graphicx}
\usepackage{enumerate}
\usepackage{natbib}
\usepackage{url} 
\RequirePackage{amsthm,amsmath}
\allowdisplaybreaks[4]
\usepackage{amssymb}

\usepackage{amsthm}
\usepackage{diagbox}
\usepackage{bbm}
\usepackage{bm}
\usepackage{float}
\usepackage{graphicx}
\usepackage[ruled,linesnumbered]{algorithm2e}
\usepackage{xr-hyper}
\usepackage{hyperref}[]
\hypersetup{
	colorlinks=true,
	linkcolor=blue,
	filecolor=magenta,      
	urlcolor=cyan,
	citecolor=blue,
}
\usepackage{cleveref}
\usepackage{subcaption}
\usepackage{xcolor}
\usepackage{multirow}
\usepackage{booktabs}
\usepackage{array}
\newcolumntype{Z}{>{\setbox0=\hbox\bgroup}c<{\egroup}@{\hspace*{-\tabcolsep}}}

\newtheorem{thm}{Theorem}[section]
\newtheorem{ass}{Assumption}[section]

\newtheorem{cor}{Corollary}[section]

\newtheorem{rem}{Remark}[section]

\newtheorem{lem}{Lemma}[section]

\usepackage{setspace}
\usepackage{diagbox}
\usepackage{tabu}
\usepackage{authblk}

\usepackage{xparse}
\usepackage{tikz}
\usetikzlibrary{positioning}
\usetikzlibrary{plotmarks}
\usetikzlibrary{matrix}
\usepackage{pgfplots}
\pgfplotsset{compat=1.7}
\usetikzlibrary{matrix,backgrounds, arrows.meta}
\pgfdeclarelayer{myback}
\pgfsetlayers{myback,background,main}
\usetikzlibrary{decorations.pathreplacing,angles,quotes}
\tikzset{mycolor/.style = {dashed,rounded corners,line width=1bp,color=#1}}%
\tikzset{myfillcolor/.style = {draw,fill=#1}}%
\tikzset{
	declare function={
		normcdf(\x,\m,\s)=1/(1 + exp(-0.07056*((\x-\m)/\s)^3 - 1.5976*(\x-\m)/\s));
	}
}
\usepackage[hmargin={0.8in, 0.8in}, vmargin={0.75in, 0.75in}]{geometry}
\makeatletter
\newcommand*{\addFileDependency}[1]{
  \typeout{(#1)}
  \@addtofilelist{#1}
  \IfFileExists{#1}{}{\typeout{No file #1.}}
}
\newcommand*\showfontsize{\f@size{} point}
\makeatother

\date{\vspace{-5ex}}

\newcommand{\blind}{1}

\begin{document}
\def\spacingset#1{\renewcommand{\baselinestretch}%
{#1}\small\normalsize} \spacingset{1}

\if1\blind
{
 \title{Two-Sample and Change-Point Inference for Non-Euclidean Valued Time Series}
		
    \author[1]{Feiyu Jiang}
    \author[2]{Changbo Zhu\footnote{Corresponding author. Email address: czhu4@nd.edu.}}
    \author[3]{Xiaofeng Shao}
        \affil[1]{Department of Statistics and Data Science, Fudan University}
    \affil[2]{Department of Applied and Computational Mathematics and Statistics, University of Notre Dame}

    \affil[3]{Department of Statistics, University of Illinois at Urbana Champaign}
		
	\date{}	
	\maketitle

} \fi

\if0\blind
{
	\title{Two-Sample and Change-Point Inference for Non-Euclidean Valued Time Series}	
	\author{}
	\date{}
	\maketitle
} \fi

\begin{abstract}
Data objects taking value in a general metric space  have become increasingly common in modern data analysis. In this paper, we study two important statistical inference problems, namely, two-sample testing and change-point detection, for such non-Euclidean data under temporal dependence. Typical examples of non-Euclidean valued time series include yearly mortality distributions,  time-varying networks, and covariance matrix time series. To accommodate unknown temporal dependence, we advance  the self-normalization (SN) technique \citep{shao2010} to the inference of non-Euclidean time series, which is substantially different from the existing SN-based inference for functional time series that reside in Hilbert space \citep{zhang2011}.
Theoretically, we propose new regularity conditions that could be easier to check than those in the recent literature, and derive the limiting distributions of the proposed test statistics under both null and local alternatives. For change-point detection problem, we also derive the consistency for the change-point location estimator, and combine our proposed change-point test with wild binary segmentation to perform multiple change-point estimation. Numerical simulations demonstrate the effectiveness and robustness of our proposed tests compared with existing  methods in the literature. Finally, we apply our tests  to two-sample inference in mortality data and change-point detection in cryptocurrency data.
\end{abstract}




\section{Introduction}

Statistical analysis of non-Euclidean data that reside in a metric space is gradually emerging as an important  branch of functional data analysis, motivated by increasing encounter of such data in many modern applications. Examples include 
the analysis of sequences of age-at-death distributions over calendar years \citep{mazzuco2015,shang2017},
covariance matrices in the analysis of diffusion tensors in medical imaging \citep{dryden2009non}, and  graph Laplacians of networks \citep{ginestet2017}.
One of the main challenges in dealing with such data is that the usual vector/Hilbert space operation, such as
projection and inner product may not be well defined and only the distance between two non-Euclidean
data objects is available. 

Despite the challenge, the list of papers that propose new statistical techniques to analyze 
non-Euclidean data has been growing. Building on Fr\'echet mean and variance \citep{frechet1948}, which are counterparts of mean and variance for metric space valued random object, \cite{dubey2019frechet} proposed a test for comparing $N(\geq 2)$ populations of metric space valued data. \cite{dubey2020frechet} developed a novel test to detect a change point in the Fr\'echet mean and/or variance in a sequence of independent non-Euclidean data. The classical linear and nonparametric regression has also been extended to  metric spaced valued data; see \cite{petersen2019}, \cite{tucker2022}, and \cite{zhang2022}, among others. 
So far, the majority of the literature on
non-Euclidean data has been limited to independent  data, and the only exceptions are \cite{zkp2022} and \cite{zhu2021,zhu2022}, which mainly focused on the autoregressive modeling of non-Euclidean valued time series. To the best of our knowledge, no  inferential tools are available for non-Euclidean valued time series in the literature.

In this paper, we address two important problems: two-sample testing and change-point detection, in the analysis of non-Euclidean  valued time series. These two problems are also well motivated by the data we analyzed in the paper, namely, the yearly age-at-death distributions for countries in Europe and  daily Pearson correlation matrices  for five cryptocurrencies.  For time series data, serial dependence is the rule rather than the exception. This motivates us  to develop new tests for non-Euclidean time series that is robust to temporal dependence. 
Note that the two testing problems have been addressed by \cite{dubey2019frechet} and \cite{dubey2020frechet}, respectively for independent non-Euclidean data, but as expected, their tests fail to control the size when there is temporal dependence in the series; see Section \ref{sec:simu} for simulation evidence. 

To accommodate unknown temporal dependence, we develop test statistics based on self-normalization \citep{shao2010,shaozhang2010}, which is a nascent inferential technique for time series data. It has been mainly developed for vector time series and has been extended to  functional time series in Hilbert space \citep{zhang2011,zhangshao2015}. The functional extension is however based on reducing the infinite dimensional functional data to finite dimension via functional principal component analysis, and then applying SN to the finite-dimensional vector time series. 
Such SN-based inference developed for time series in Hilbert space cannot be applied to non-Euclidean valued time series, since  the projection and inner product commonly used for data in Hilbert space  are not available for data objects that live in a general metric space. 
The SN-based extension to non-Euclidean valued time series is therefore fairly different from that in \cite{zhang2011} and \cite{zhangshao2015}, in terms of both methodology and theory. For independent non-Euclidean valued data, \cite{dubey2019frechet,dubey2020frechet} build on the empirical process theory \citep{vaart1996weak} by regulating the complexity of the analyzed metric space, which is in general abstract and  may not be easy to  verify. In our paper, we take a different approach that is inspired by the M-estimation theory  in \cite{pollard1985new} and \cite{hjort2011asymptotics} for Euclidean data, and extend  it to non-Euclidean setting.  { We assume that the metric distance between data and the estimator of the Fr\'echet mean admits certain decomposition, which includes a bias term, a leading stochastic term, and a remainder term.} Our technical assumptions are more intuitive and { could be }easier to check in practice. Furthermore,  we are able to obtain explicit asymptotic distributions of our test statistics under the local alternatives of rate $O(n^{-1/2})$, where $n$ is the sample size,  under our assumptions, whereas they seem difficult to derive under the entropy integral type conditions employed by \cite{dubey2019frechet,dubey2020frechet}.

The remainder of the paper is organized as follows. Section \ref{sec:pre} provides background of non-Euclidean metric space in which random objects of interest reside in, and some basic assumptions that will be used throughout the paper. Section \ref{sec:two} proposes SN-based two-sample tests for non-Euclidean time series. Section \ref{sec:cpt}  considers SN-based change-point tests. Numerical studies for the proposed tests are presented in Section \ref{sec:simu}, and Section \ref{sec:app} demonstrates the applicability of these tests through real data examples. Section \ref{sec:con} concludes.  Proofs of all results are relegated to Appendix \ref{sec:proof}. Appendix \ref{sec:example} summarizes the examples that satisfy assumptions in Section \ref{sec:pre}, and Appendix \ref{sec:simu_functional} provides simulation results for functional time series.


Some notations used throughout the paper are defined as follows. Let $\|\cdot\|$ denote the conventional Euclidean norm.  Let $D[0,1]$ denote the space of functions on $[0, 1]$ which are right continuous with left limits, endowed with the Skorokhod topology  \citep{Billingsley1968}. We use $\Rightarrow$ to denote weak convergence in $D[0,1]$ or more generally in $\mathbb{R}^m$-valued function space $D^m[0,1]$, where $m\in\mathbb{N}$; $\to_d$ to denote convergence in distribution; and $\to_p$ to denote convergence in probability. 
A sequence  of random variables $X_n$ is said to be $O_p(1)$ if it is bounded in
probability. For $x\in\mathbb{R}$, define $\lfloor x\rfloor$ as the largest integer that is smaller than or equal to $x$, and $\lceil x \rceil$ as the smallest integer that is greater than or equal to $x$.

\section{Preliminaries and Settings}\label{sec:pre}
In this paper, we consider a metric space $(\Omega,d)$ that is totally bounded, i.e. for any $\epsilon>0$, there exist a finite number of open $\epsilon$-balls whose union can cover $\Omega$. For a sequence of {\it stationary} random objects $\{Y_t\}_{t\in\mathbb{Z}}$ defined on $(\Omega,d)$, we follow \cite{frechet1948}, and define their Fr\'echet mean and variance  by
\begin{equation}\label{defn_frecht}
\mu=\arg\min_{\omega\in\Omega}\mathbb{E}d^2(Y_t,\omega),\quad V=\mathbb{E}d^2(Y_t,\mu),
\end{equation}
respectively. Fr\'echet mean extends the traditional mean in linear spaces to more general metric spaces by minimizing  expected squared metric distance between the random object $Y_t$ and the centroid akin to the conventional mean by minimizing the expected sum of residual squares. It is particularly useful for objects that lie in abstract spaces without explicit algebraic structure. Fr\'echet variance, defined by such expected  squared metric  distance, is then used for measuring the dispersion in data.

Given finite samples $\{Y_t\}_{t=1}^{n}$, we define their  Fr\'echet subsample mean and variance as 
\begin{align}\label{substat}
\begin{split}
    & \hat{\mu}_{[a,b]}=\arg\min_{\omega\in\Omega}\sum_{t=1+\lfloor na\rfloor}^{\lfloor nb\rfloor}d^2(Y_t,\omega),\\ &\hat{V}_{[a,b]}=\frac{1}{\lfloor nb\rfloor-\lfloor na\rfloor}\sum_{t=1+\lfloor na\rfloor}^{\lfloor nb\rfloor}d^2(Y_t,\hat{\mu}_{[a,b]}),
\end{split}
\end{align}
where $(a,b)\in\mathcal{I}_{\eta}$, $\mathcal{I}_{\eta}=\{(a,b): 0\leq a<b\leq 1, b-a\geq \eta \}$ for some trimming parameter $\eta\in(0,1)$.  The case  corresponding to $a=0$ and $b\geq \eta$ is further denoted as 
$$\hat{\mu}_{[0,b]}=\hat{\mu}_{b},\quad \hat{V}_{[0,b]}=\hat{V}_{b},$$
with special case of $b=1$ corresponding to  Fr\'echet sample mean and variance \citep{petersen2019}, respectively.

Note that both Fr\'echet  (subsample) mean and variance depend on the space $\Omega$ and metric distance $d$, which require further regulation for desired inferential purposes. In this paper, we do not impose independence assumptions, and our technical treatment differs substantially from those in the literature, c.f. \cite{petersen2019,dubey2019frechet,dubey2020frechet,dubey2020functional,dubey2021modeling}.

\begin{ass}\label{ass_diff}
	$\mu$ is unique, and for some $\delta>0$, there exists a constant $K>0$ such that,
	$$
	\inf _{d(\omega, \mu)<\delta}\left\{\mathbb{E}\left(d^{2}(Y_0, \omega)\right)-\mathbb{E}\left(d^{2}(Y_0, \mu)\right)-K d^{2}(\omega, \mu)\right\} \geq 0.
	$$
\end{ass}

\begin{ass}\label{ass_unique}
	For any $(a,b)\in\mathcal{I}_{\eta}$, $\hat{\mu}_{[a,b]}$ exists and is unique almost surely.
\end{ass}
\begin{ass}\label{ass_LLN}
	For any $\omega\in \Omega$, and $(a,b)\in\mathcal{I}_{\eta}$, as $n\to\infty$,
	$$
	\frac{1}{\lfloor nb\rfloor-\lfloor na\rfloor}\sum_{t=\lfloor na\rfloor+1}^{\lfloor nb\rfloor}[d^2(Y_t,\omega)-\mathbb{E}d^2(Y_t,\omega)]\to_p 0.
	$$
\end{ass}
\begin{ass}\label{ass_FCLTo}
	For some constant $\sigma>0$, 
	$$
	\frac{1}{\sqrt{n}}\sum_{t=1}^{\lfloor nr\rfloor}\left(d^2(Y_t,\mu)-V\right)\Rightarrow \sigma B(r),\quad r\in(0,1],
	$$ 
	where $B(\cdot)$ is a  standard  Brownian motion. 
\end{ass}
\begin{ass}\label{ass_expand}
	Let $B_{\delta}(\mu) \subset \Omega$ be a ball of radius $\delta$ centered at $\mu$.  For $\omega\in B_{\delta}(\mu)$, i.e. $d(\omega,\mu)\leq \delta$,  we assume the following expansion  
	$$
	d^2(Y_t,\omega)-d^2(Y_t,\mu)= K_dd^2(\omega,\mu)+ g(Y_t,\omega,\mu)+R(Y_t,\omega,\mu),\quad t\in\mathbb{Z},
	$$
	where $K_d\in(0,\infty)$ is a constant, and $g(Y_t,\omega,\mu)$ and $R(Y_t,\omega,\mu)$  satisfy that, as $n\to\infty$,
	$$ \sup_{(a,b)\in\mathcal{I}_{\eta}}\sup_{\omega\in B_{\delta}(\mu)}\left|\frac{ n^{-1/2}\sum_{t=\lfloor n a\rfloor+1}^{\lfloor n b\rfloor} g(Y_t,\omega,\mu)}{d(\omega,\mu)}\right|=O_p(1),
	$$
	and
	$$ \sup_{(a,b)\in\mathcal{I}_{\eta}}\sup_{\omega\in B_{\delta}(\mu)}\left|\frac{n^{-1/2} \sum_{t=\lfloor n a\rfloor+1}^{\lfloor n b\rfloor} R(Y_t,\omega,\mu)}{d(\omega,\mu)+n^{1/2}d^2(\omega,\mu)}\right|\to_p 0,
	$$
	respectively.
\end{ass}

Several remarks are given in order.  Assumptions \ref{ass_diff}-\ref{ass_LLN} are standard and  similar conditions can be found in \cite{dubey2019frechet,dubey2020frechet}  and \cite{petersen2019}.  Assumptions \ref{ass_diff} and \ref{ass_unique} are adapted from Assumption (A1) in \cite{dubey2020frechet}, and are required for identification purpose.  In particular, { Assumption \ref{ass_diff} requires that the expected squared metric distance $\mathbb{E}d^2(Y_t,\omega)$ can be well separated from the Fr\'echet variance, and the separation is quadratic in terms of the distance $d(\omega,\mu).$ }Assumption \ref{ass_unique} is useful for obtaining the uniform convergence of the subsample estimate of  Fr\'echet mean, i.e.,  $\hat{\mu}_{[a,b]}$, which is a key ingredient  in forming the self-normalizer in SN-based inference. 
Assumption \ref{ass_LLN} is a pointwise weak law of large numbers,  c.f. Assumption (A2) in \cite{dubey2020frechet}.  Assumption \ref{ass_FCLTo}  requires the invariance principle to hold to regularize the partial sum that appears in Fr\'echet subsample variances.   Note that 
$d^2(Y_t,\omega)$ takes value in $\mathbb{R}$ for any fixed $\omega\in\Omega$, thus  both Assumption \ref{ass_LLN} and \ref{ass_FCLTo}  could be implied by  high-level weak temporal dependence conditions (e.g., strong mixing) in conventional Euclidean space, see \cite{shao2010, shao2015} for discussions.

 \Cref{ass_expand} distinguishes our theoretical analysis from the existing literature.
Its idea is inspired by  \cite{pollard1985new} and \cite{hjort2011asymptotics} for M-estimators. { In the conventional Euclidean space, i.e. $(\Omega,d)=(\mathbb{R}^m,\|\cdot\|)$ for $m\geq 1$, it is easy to see that the expansion in \Cref{ass_expand} holds with $K_d=1$, $
       g(Y_t,\omega,\mu)=2(\mu-\omega)^{\top}(Y_t-\mu)$ and $R(Y_t,\omega,\mu)\equiv 0.$ In more general cases,  Assumption \ref{ass_expand} can be interpreted as the expansion of $d^2(Y_t,\omega)$ around the target value $d^2(Y_t,\mu)$. In particular, $K_dd^2(\omega,\mu)$ can be viewed as the bias term,
  $g(Y_t,\omega,\mu)$ works as the asymptotic leading term  that is proportional to the distance $d(\omega,\mu)$ while $R(Y_t,\omega,\mu)$ is the asymptotically negligible remainder term. More specifically, after suitable normalization, it reads as,
  \begin{flalign*}
    &n^{-1/2} \sum_{t=\lfloor na\rfloor+1}^{\lfloor nb\rfloor} [d^2(Y_t,\omega)-d^2(Y_t,\mu)] \\= &\underbrace{ n^{1/2}(b-a)K_dd^2(\omega,\mu)}_{\text{bias  term}} + \underbrace{d(\omega,\mu)\frac{n^{-1/2}\sum_{t=\lfloor na\rfloor+1}^{\lfloor nb\rfloor}g(Y_t,\omega,\mu)}{d(\omega,\mu)}}_{\text{stochastic  term}}\\&+\underbrace{n^{-1/2}\sum_{t=\lfloor na\rfloor+1}^{\lfloor nb\rfloor} R(Y_t,\omega,\mu)}_{\text{remainder term}}.
\end{flalign*}
And the verification of this assumption can be done by analyzing each term. } 
In comparison, existing literature, e.g. \cite{petersen2019},   \cite{dubey2019frechet,dubey2020frechet,dubey2021modeling},  impose assumptions on the complexity of $(\Omega,d)$. These assumptions typically involve the behaviors of entropy  integral and covering numbers rooted in the  empirical process theory \citep{vaart1996weak}, which are abstract and difficult to check in practice, see  Propositions 1 and 2 in \cite{petersen2019}.  Assumption \ref{ass_expand}, on the contrary, regulates directly on the metric $d$ and  could be easily checked  for the examples below. Moreover, Assumption \ref{ass_expand} is useful for deriving local powers of tests to be developed in this paper, see Section \ref{sec:theory_two} and \ref{sec:theory_cpt} for more details. 
Examples that can satisfy Assumptions \ref{ass_diff}-\ref{ass_expand} include: \begin{itemize}
	\item $L_2$ metric $d_L$ for $\Omega$ being the set of square integrable functions on $[0,1]$;
	\item 2-Wasserstein metric $d_W$ for $\Omega$ being the set of univariate probability distributions on $\mathbb{R}$;
	\item Frobenius metric $d_F$ for $\Omega$  being the set of square matrices, including the special cases of covariance matrices and graph Laplacians;
	\item log-Euclidean metric $d_E$ for $\Omega$  being the set of covariance matrices.
\end{itemize}
We refer to Appendix \ref{sec:example}  for more details of these examples and  verifications of above assumptions for them.

\section{Two-Sample Testing}\label{sec:two}
This section considers two-sample testing  in metric space under temporal dependence. 
For two sequences of temporally dependent random objects $\{Y_t^{(1)},Y_t^{(2)}\}_{t\in\mathbb{Z}}$ on $(\Omega,d)$, we denote $Y_t^{(i)}\sim P^{(i)}$, where $P^{(i)}$ is the underlying marginal distribution of $Y_t^{(i)}$ with Fr\'echet mean and variance $\mu^{(i)}$ and $V^{(i)}$, $i=1,2$. Given finite sample observations $\{Y_t^{(1)}\}_{t=1}^{n_1}$ and $\{Y_t^{(2)}\}_{t=1}^{n_2}$, we are interested in the following two-sample testing problem, 
$$
\mathbb{H}_0: P^{(1)}=P^{(2)},~~\mbox{  v.s. }\mathbb{H}_a: P^{(1)}\neq P^{(2)}.
$$
Let $n=n_1+n_2$, we assume two samples are balanced, i.e.  $n_1/n\to\gamma_1$ and $n_2/n\to\gamma_2$ with $\gamma_1,\gamma_2\in(0,1)$ and $\gamma_1+\gamma_2=1$ as $\min(n_1,n_2)\to\infty.$ For $r\in(0,1]$, we define their recursive Fr\'echet sample mean and variance by 
$$\hat{\mu}^{(i)}_{r}=\arg\min_{\omega\in\Omega}\sum_{t=1}^{\lfloor rn_i\rfloor }d^2(Y_t^{(i)},\omega),\quad \hat{V}^{(i)}_{r}=\frac{1}{\lfloor rn_i\rfloor}\sum_{t=1}^{\lfloor rn_i\rfloor}d^2(Y_t^{(i)},\hat{\mu}^{(i)}_{r}),\quad i=1,2.$$

A natural candidate test of $\mathbb{H}_0$ is to compare their Fr\'echet sample mean and variance by contrasting $(\hat{\mu}^{(1)}_{1},\hat{V}^{(1)}_{1})$ and $(\hat{\mu}^{(2)}_{1},\hat{V}^{(2)}_{1})$. For the mean part,  it is tempting to use  $d(\hat{\mu}^{(1)}_{1},\hat{\mu}^{(2)}_{1})$ as the testing statistic. However, this is a non-trivial task as the limiting behavior of $d(\hat{\mu}^{(1)}_{1},\hat{\mu}^{(2)}_{1})$ depends heavily on the structure of the metric space, which may not admit  conventional algebraic operations. 
Fortunately, both $\hat{V}^{(1)}_{1}$ and $\hat{V}^{(2)}_{1}$ take value in $\mathbb{R}$, and it is thus intuitive  to compare their difference. In fact, \cite{dubey2019frechet} propose the test statistic of the form
$$
U_n= \frac{n_1n_2}{n \hat{\sigma}_1^2\hat{\sigma}_2^2}(\hat{V}^{(1)}_{1}-\hat{V}^{(2)}_{1})^2,
$$
where $\hat{\sigma}_i^2$ is a consistent estimator of $\lim_{n_i\to\infty}\mathrm{Var}\{\sqrt{n}(\hat{V}^{(i)}_{1}-V^{(i)})\}$, $i=1,2.$ 

However,  $U_n$ requires both  within-group and  between-group independence, which is too stringent to be realistic for  applications in this paper. When either of such  independence is violated, the test may fail to control size, see Section \ref{sec:simu} for numerical evidence. Furthermore, taking into account the temporal dependence   requires replacing the variance by long-run variance, whose consistent estimation  usually involves laborious tuning  such as choices of kernels and bandwidths  \citep{newey1987simple,andrews1991heteroskedasticity}.  To this end, we invoke self-normalization technique to bypass the foregoing issues. 

The core principle of self-normalization for the time series inference is to  use an inconsistent long-run variance estimator that is a function of recursive estimates to
yield an asymptotically pivotal statistic. The SN procedure does not involve any tuning parameter or involves less number of tuning parameters compared to traditional counterparts. See \cite{shao2015} for a comprehensive review of recent developments
for low dimensional time series. For recent extension to inference for high-dimensional time series, we refer to 
\cite{wang2020} and \cite{wzvs2022}.

\subsection{Test Statistics}
Define the recursive subsample test statistic based on Fr\'echet variance as
$$
T_n(r)=r(\hat{V}^{(1)}_{r}-\hat{V}^{(2)}_{r}),~r\in [\eta,1],
$$
and then construct the SN based test statistic as 
\begin{equation}\label{Dn1}
	D_{n,1}=\frac{n\left[T_n(1)\right]^2}{\sum_{k=\lfloor n\eta\rfloor}^n \left[T_n(\frac{k}{n})-\frac{k}{n}T_n(1)\right]^2},
\end{equation}
where $\eta\in(0,1)$ is a  trimming parameter  for controlling the estimation effect of $T_n(r)$ when $r$ is close to 0, which is important for deriving the uniform convergence of $\{\sqrt{n}T_n(r), r\in[\eta,1]\}$,  see \cite{zhou2013inference} and \cite{jiang2022modelling} for similar technical treatments.

The testing statistic \eqref{Dn1} is composed of the numerator $n[T_n(1)]^2$, which captures the difference in Fr\'echet variances, and the denominator  $\sum_{k=\lfloor n\eta\rfloor}^n \big[T_n(\frac{k}{n})-$ $\frac{k}{n}T_n(1)\big]^2$, which   is called self-normalizer and mimics the behavior of the numerator with suitable centering and trimming. For each $r\in[\eta,1]$, $T_n(r)$ is expected to be a consistent estimator for $r(V^{(1)}-V^{(2)}).$ Therefore,  under $\mathbb{H}_a$, $T_n(1)$ is large when there is significant difference in Fr\'echet variance, whereas the key element $T_n(r)-rT_n(1)$ in  self-normalizer remains to be small. This suggests that we should reject $\mathbb{H}_0$ for large values of $D_{n,1}$.  

Note that  \eqref{Dn1} only targets at difference in Fr\'echet variances. To detect the  difference in Fr\'echet means, we can use  contaminated Fr\'echet variance \citep{dubey2020frechet}. Let
$$
\hat{V}^{C,(1)}_{r}=\frac{1}{\lfloor rn_1\rfloor }\sum_{t=1}^{\lfloor rn_1\rfloor}d^2(Y_t^{(1)},\hat{\mu}^{(2)}_{r}),\quad \text{and}\quad\hat{V}^{C,(2)}_{r}=\frac{1}{\lfloor rn_2\rfloor }\sum_{t=1}^{\lfloor rn_2\rfloor}d^2(Y_t^{(2)},\hat{\mu}^{(1)}_{r}),
$$
and  $$
T_n^C(r)=r(\hat{V}^{C,(1)}_{r}+\hat{V}^{C,(2)}_{r}-\hat{V}^{(1)}_{r}-\hat{V}^{(2)}_{r}).
$$
The contaminated Fr\'echet  sample  variances $\hat{V}^{C,(1)}_{r}$ and $\hat{V}^{C,(2)}_{r}$ switch the role of $\hat{\mu}_r^{(1)}$ and $\hat{\mu}_r^{(2)}$ in $\hat{V}^{(1)}_{r}$ and $\hat{V}^{(2)}_{r}$, respectively, and could be viewed as proxies for measuring  Fr\'echet mean differences. 

Intuitively, it is expected that $\hat{V}^{C,(i)}_{r}\approx \mathbb{E}d^2(Y_t^{(i)},\mu^{(3-i)}),$ and $\hat{V}^{(i)}_{r}\approx \mathbb{E}d^2(Y_t^{(i)},$ $\mu^{(i)}), i=1,2$.  Under $\mathbb{H}_0$, both $\hat{\mu}_r^{(1)}$ and $\hat{\mu}_r^{(2)}$ are consistent estimators for ${\mu}^{(1)}={\mu}^{(2)}$, thus $\hat{V}^{C,(i)}_{r}\approx\hat{V}^{(i)}_{r}, i=1,2$, which indicates a small value for $T_n^C(r)$.  On the contrary, when $d({\mu}^{(1)},{\mu}^{(2)})>0$,   $\hat{V}^{C,(i)}_{r}$ could be  much larger than $\hat{V}^{(i)}_{r}$ as $\mathbb{E}d^2(Y_t^{(i)},\mu^{(3-i)})>\mathbb{E}d^2(Y_t^{(i)},\mu^{(i)})=\arg\min_{\omega\in\Omega}\mathbb{E}d^2(Y_t^{(i)},\omega)$, $i=1,2$, resulting in large value of $T_n^C(r)$.

The power-augmented test statistic is thus defined by     
\begin{equation}\label{D2}
	D_{n,2}=\frac{n\left\{\left[T_n(1)\right]^2+\left[T_n^C(1)\right]^2\right\}}{ \sum_{k=\lfloor n\eta\rfloor}^n  \left\{\left[T_n(\frac{k}{n})-\frac{k}{n}T_n(1)\right]^2+  \left[T_n^C(\frac{k}{n})-\frac{k}{n}T_n^C(1)\right]^2\right\}},
\end{equation}
where the additional term $\sum_{k=\lfloor n\eta\rfloor}^n  \left[T_n^C(\frac{k}{n})-\frac{k}{n}T_n^C(1)\right]^2$ that appears in the self-normalizer is used to stabilize finite sample performances.   

{
\begin{rem}
    Our proposed tests could be adapted to comparison of $N$-sample populations \citep{dubey2019frechet}, where $N\geq 2$. An natural way of extension would be   aggregating all the  pairwise differences in Fr\'echet variance and contaminated variance. Specifically, let the $N$ groups of random data objects be $\{Y_t^{(i)}\}_{t=1}^{n_i}$, $i=1,\cdots,N$.   The null hypothesis is given as  $$\mathbb{H}_0: P^{(1)}=\cdots=P^{(N)},$$ for some $N\geq 2$. 
    
    Let $\hat{\mu}^{(i)}_{r}$ and $\hat{V}^{(i)}_{r}$, $r\in[\eta,1]$ be the   Fr\'echet subsample mean and variance, respectively, for the $i$th group, $i=1,\cdots, N$.  
For $1\leq i\neq j\leq N$, define the pairwise contaminated Fr\'echet  subsample  variance as 
$$
\hat{V}^{C,(i,j)}_{r}=\frac{1}{\lfloor rn_i\rfloor }\sum_{t=1}^{\lfloor rn_i\rfloor}d^2(Y_t^{(i)},\hat{\mu}^{(j)}_{r}),~r\in [\eta,1],
$$
and define the recursive statistics  
    $$
    T_n^{i,j}(r)=r(\hat{V}^{(i)}_{r}-\hat{V}^{(j)}_{r}),\quad T_n^{C,i,j}(r)=r(\hat{V}^{C,(i,j)}_{r}+\hat{V}^{C,(j,i)}_{r}-\hat{V}^{(i)}_{r}-\hat{V}^{(j)}_{r}),~r\in [\eta,1].
    $$

In the same spirit of the test statistics $D_{n,1}$ and $D_{n,2}$, for $n=\sum_{i=1}^N n_i,$ we may construct their counterparts for the $N$-sample testing problem as 
$$
D^{(N)}_{n,1}=\frac{n\sum_{i<j}\left[T_n^{i,j}(1)\right]^2}{\sum_{k=\lfloor n\eta\rfloor}^n \sum_{i<j}\left[T_n^{i,j}(\frac{k}{n})-\frac{k}{n}T_n^{i,j}(1)\right]^2},
$$
and 
$$
D^{(N)}_{n,2}=\frac{n\sum_{i<j}\left\{\left[T_n^{i,j}(1)\right]^2+\left[T_n^{C,i,j}(1)\right]^2\right\}}{\sum_{k=\lfloor n\eta\rfloor}^n \sum_{i<j}\left\{\left[T_n^{i,j}(\frac{k}{n})-\frac{k}{n}T_n^{i,j}(1)\right]^2+\left[T_n^{C,i,j}(\frac{k}{n})-\frac{k}{n}T_n^{C,i,j}(1)\right]^2\right\}}.
$$
\end{rem}
Compared with classical $N$-sample testing problem in Euclidean spaces, e.g. analysis of variance (ANOVA), the above modification does not require Gaussianity, equal variance, or serial independence. Therefore, they could be  work for  broader classes of distributions.  We leave out the details for the sake of space.
}

\subsection{Asymptotic Theory}\label{sec:theory_two}
Before we present asymptotic results of the proposed tests, we need a slightly stronger assumption than Assumption \ref{ass_FCLTo} to regulate the joint behavior of partial sums for both samples. 
\begin{ass}\label{ass_FCLT}
	For some $\sigma_1>0$ and $\sigma_2>0$, we have 
	$$
	\frac{1}{\sqrt{n}}\sum_{t=1}^{\lfloor nr\rfloor} 
	\left(\begin{matrix}
		d^2(Y_t^{(1)},\mu^{(1)})-V^{(1)}\\
		d^2(Y_t^{(2)},\mu^{(2)})-V^{(2)}
	\end{matrix}\right)\Rightarrow  \left(\begin{matrix}
		\sigma_1B^{(1)}(r)\\
		\sigma_2B^{(2)}(r)
	\end{matrix}\right),
	$$
	where $B^{(1)}(\cdot)$ and $B^{(2)}(\cdot)$ are two  standard Brownian motions with unknown correlation parameter $\rho\in (-1,1)$, and $\sigma_1,\sigma_2\neq 0$ are unknown parameters characterizing the long-run variance.
\end{ass}

\begin{thm}\label{thmtwo_0}
	Suppose Assumptions \ref{ass_diff}-\ref{ass_expand} (with \ref{ass_FCLTo} replaced by  \ref{ass_FCLT}) hold for both $\{Y_t^{(1)}\}_{t=1}^{n_1}$ and $\{Y_t^{(2)}\}_{t=1}^{n_2}$.   Then as $n\to\infty$,  under $\mathbb{H}_0$, for $i=1,2$,
	$$
	D_{n,i}\to_d \frac{\xi^2_{\gamma_1,\gamma_2}(1;\sigma_1,\sigma_2)}{\int_{\eta}^{1}\left[\xi_{\gamma_1,\gamma_2}(r;\sigma_1,\sigma_2)-r\xi_{\gamma_1,\gamma_2}(1;\sigma_1,\sigma_2)\right]^2dr}:=\mathcal{D}_{\eta},\quad 
	$$
	where \begin{equation}\label{xi}
	    \xi_{\gamma_1,\gamma_2}(r;\sigma_1,\sigma_2)=\gamma_1^{-1}\sigma_1B^{(1)}(\gamma_1r)-\gamma_2^{-1}\sigma_2B^{(2)}(\gamma_2r).
	\end{equation}
\end{thm}
Theorem \ref{thmtwo_0} obtains the same limiting null distribution for  Fr\'echet variance based test $D_{n,1}$ and its power-augmented version $D_{n,2}$. Although $D_{n,2}$ contains contaminated variance $T_n^C(1)$, its contribution is asymptotically vanishing as $n\to\infty$. This is  an immediate consequence of the fact that  $$\sup_{r\in[\eta,1]}|\sqrt{n}T_n^C(r)|\to_p0,$$
see proof of Theorem \ref{thmtwo_0} in Appendix \ref{sec:proof}. Similar phenomenon has been documented in \cite{dubey2019frechet} under different assumptions.

We next consider the power behavior under the Pitman local alternative,
$$
\mathbb{H}_{an}:\quad V^{(1)}-V^{(2)}=n^{-\kappa_V}\Delta_V,\quad  \mbox{ and        }\quad d^2(\mu^{(1)},\mu^{(2)})=n^{-\kappa_M}\Delta_M,$$  with  $\Delta_V\in\mathbb{R}$, $\Delta_M\in(0,\infty)$, and $\kappa_V,\kappa_M\in (0,\infty).$ 
\begin{thm}\label{thmtwo_alter}
	Suppose Assumptions \ref{ass_diff}-\ref{ass_expand} (with \ref{ass_FCLTo} replaced by  \ref{ass_FCLT}) hold for both $\{Y_t^{(1)}\}_{t=1}^{n_1}$ and $\{Y_t^{(2)}\}_{t=1}^{n_2}$.  As $n\to\infty$,  under $\mathbb{H}_{an}$, 
	\begin{itemize}
		\item if $\max\{\kappa_V,\kappa_M\}\in(0,1/2)$, then for $i=1,2$,
		$D_{n,i}\to_p\infty$;  
		\item if $\min\{\kappa_V,\kappa_M\}\in(1/2,\infty)$, then for $i=1,2$,
		$D_{n,i}\to_d\mathcal{D}_{\eta}$; 
		\item if $\kappa_V=1/2$ and $\kappa_M\in(1/2,\infty)$, then for $i=1,2$,
		$$D_{n,i}\to_d \frac{\left(\xi_{\gamma_1,\gamma_2}(1;\sigma_1,\sigma_2)+\Delta_V\right)^2}{\int_{\eta}^{1}\left(\xi_{\gamma_1,\gamma_2}(r;\sigma_1,\sigma_2)-r\xi_{\gamma_1,\gamma_2}(1;\sigma_1,\sigma_2)\right)^2dr};$$
		\item if  $\kappa_V\in (1/2,\infty)$ and $\kappa_M=1/2$, then $D_{n,1}\to_d\mathcal{D}_{\eta}$, and  $$ D_{n,2}\to_d \frac{\left(\xi_{\gamma_1,\gamma_2}(1;\sigma_1,\sigma_2)\right)^2+4K_d^2\Delta_M^2}{\int_{\eta}^{1}\left(\xi_{\gamma_1,\gamma_2}(r;\sigma_1,\sigma_2)-r\xi_{\gamma_1,\gamma_2}(1;\sigma_1,\sigma_2)\right)^2dr};$$
		\item if $\kappa_V=\kappa_M=1/2$, then  
		\begin{flalign*}
			&D_{n,1}\to_d \frac{\left(\xi_{\gamma_1,\gamma_2}(1;\sigma_1,\sigma_2)+\Delta_V\right)^2}{\int_{\eta}^{1}\left(\xi_{\gamma_1,\gamma_2}(r;\sigma_1,\sigma_2)-r\xi_{\gamma_1,\gamma_2}(1;\sigma_1,\sigma_2)\right)^2dr},\\ &D_{n,2}\to_d \frac{\left(\xi_{\gamma_1,\gamma_2}(1;\sigma_1,\sigma_2)+\Delta_V\right)^2+4K_d^2\Delta_M^2}{\int_{\eta}^{1}\left(\xi_{\gamma_1,\gamma_2}(r;\sigma_1,\sigma_2)-r\xi_{\gamma_1,\gamma_2}(1;\sigma_1,\sigma_2)\right)^2dr};
		\end{flalign*}
		
	\end{itemize}
	where $K_d$ is defined in Assumption \ref{ass_expand}.
\end{thm}
Theorem \ref{thmtwo_alter} presents the asymptotic  behaviors for both test statistics under local alternatives in various regimes. In particular,  $D_{n,1}$ can detect differences in Fr\'echet variance at local rate $n^{-1/2}$, but possesses trivial power against   Fr\'echet mean difference regardless of the regime of $\kappa_M$. In comparison,  $D_{n,2}$ is powerful for differences in both Fr\'echet variance  and Fr\'echet mean at local rate $n^{-1/2}$, which validates our claim  that $D_{n,2}$ indeed augments power.

Our results merit additional remarks when compared with \cite{dubey2019frechet}. In \cite{dubey2019frechet}, they only  obtain
the consistency of their test under either  $n^{1/2}|V^{(1)}-V^{(2)}|\to\infty$  or $n^{1/2}d^2(\mu^{(1)},\mu^{(2)})\to\infty$, while Theorem \ref{thmtwo_alter} explicitly characterizes the asymptotic distributions of our test statistics under local alternatives of order $O(n^{-1/2})$, which depend on $\kappa_V$ and $\kappa_M$. Such theoretical improvement relies crucially on our newly developed proof techniques based on Assumption  \ref{ass_expand}, and it seems difficult to derive such limiting distributions under empirical-process-based assumptions in \cite{dubey2019frechet}.  { However, we do admit that self-normalization could result in moderate power loss compared with $t$-type test statistics, see \citep{shaozhang2010} for evidence in Euclidean space.  }

{
Note that the limiting distributions derived in \Cref{thmtwo_0} and \Cref{thmtwo_alter} contain a key quantity $\xi_{\gamma_1,\gamma_2}(r;\sigma_1,\sigma_2)$ defined in \eqref{xi}, which depends on nuisance parameters $\sigma_1,\sigma_2$ and $\rho$. This may hinder the practical use of the tests.  The following corollary, however, justifies the wide applicability  of our tests. 
\begin{cor}\label{cor_xi}
     Under Assumption \ref{ass_FCLT}, if either $\gamma_1=\gamma_2=1/2$ or $\rho=0$, then for any constants $C_a,C_b\in\mathbb{R}$, 
     $$
     \frac{\left(\xi_{\gamma_1,\gamma_2}(1;\sigma_1,\sigma_2)+C_a\right)^2+C_b^2}{\int_{\eta}^{1}\left(\xi_{\gamma_1,\gamma_2}(r;\sigma_1,\sigma_2)-r\xi_{\gamma_1,\gamma_2}(1;\sigma_1,\sigma_2)\right)^2dr}=_d \frac{\left(B(1)+C_a/C_{\xi}\right)^2+ (C_b/C_{\xi})^2}{\int_{\eta}^{1}\left(B(r)-rB(1)\right)^2dr},
     $$
     where 
$$
C_{\xi}=\begin{cases}
 \sqrt{2\sigma_1^2+2\sigma_2^2-4\rho\sigma_1\sigma_2}, \quad & \text{if } \gamma_1=\gamma_2, \\
\sqrt{{\sigma_1^2/\gamma_1}+{\sigma_2^2/\gamma_2}}, \quad & \text{if }\rho=0.
\end{cases} 
$$
\end{cor}
Therefore, by choosing $C_a=C_b=0$ in \Cref{cor_xi}, we obtain the pivotal limiting distribution 
$$
\mathcal{D}_{\eta}=_d\frac{B^2(1)}{\int_{\eta}^{1}\left(B(r)-rB(1)\right)^2dr}.
$$
The asymptotic distributions in \Cref{thmtwo_alter} can be similarly derived by letting either $C_a=\Delta_V$ or $C_b=2K_d\Delta_M$.
}

Therefore, when  either two samples are of the same length $(\gamma_1=\gamma_2)$  or  two samples are asymptotically independent $(\rho=0)$, the limiting distribution $\mathcal{D}_{\eta}$ is pivotal.   In practice, we reject $\mathbb{H}_0$ if $D_{n,i}>Q_{\mathcal{D}_{\eta}}(1-\alpha)$ where $Q_{\mathcal{D}_{\eta}}(1-\alpha)$ denotes the $1-\alpha$ quantile of (the pivotal) 
${D}_{\eta}$.  

In Table \ref{tab_two_critic}, we tabulate commonly used critical values under various choices of $\eta$ by simulating 50,000 i.i.d. $\mathcal{N}$(0,1) random variables 10,000 times and approximating a standard Brownian motion by standardized partial sum of i.i.d. $\mathcal{N}$(0,1) random variables.

\begin{table}[H]
	\centering
	\caption{Simulated critical values $Q_{\mathcal{D}_{\eta}}(1-\alpha)$}
	\label{tab_two_critic}
	\begin{tabular}{ccccc}
		\hline
		\diagbox{$\alpha$}{$\eta$} & 0.02   & 0.05   & 0.1    & 0.15   \\ \hline
		10\%                                         & 28.51  & 28.88  & 30.02  & 31.87   \\
		5\%                                          & 46.10  & 46.72  & 48.80  & 51.87   \\
		1\%                                          & 101.58 & 103.70 & 108.93 & 116.72  \\
		0.5\%                                        & 131.55 & 134.00 & 142.34 & 151.93\\
		\hline
	\end{tabular}
\end{table}

\section{Change-Point Test}\label{sec:cpt}
Inspired by the two-sample tests developed in Section \ref{sec:two}, this section considers the change-point detection problem for a sequence of random objects $\{Y_t\}_{t=1}^n$, i.e.
$$
\mathbb{H}_{0}: Y_{1}, Y_{2}, \ldots, Y_{n} \sim P^{(1)}
$$
against the single change-point alternative,
$$
\mathbb{H}_{a}: \text { there exists } 0<\tau<1 \text { such that }Y_t=\left\{\begin{array}{l}
	Y_t^{(1)}\sim P^{(1)},   1\leq t\leq \lfloor n\tau\rfloor \\
	Y_t^{(2)}\sim P^{(2)},   \lfloor n\tau\rfloor +1\leq t \leq n.
\end{array}\right.
$$
The single change-point testing problem can be roughly viewed as  two-sample testing without knowing where 
the two samples split, and they share certain similarities in terms of statistical methods and theory. 

Recall the  Fr\'echet subsample mean $\hat{\mu}_{[a,b]}$  and variance  $\hat{V}_{[a, b]}$ in \eqref{substat}, we further define the pooled contaminated variance separated by $r\in(a,b)$ as
$$
\hat{V}_{[r ; a, b]}^{C}=\frac{1}{\lfloor n r\rfloor-\lfloor n a\rfloor} \sum_{i=\lfloor n a\rfloor+1}^{\lfloor n r\rfloor} d^{2}\left(Y_{i}, \hat{\mu}_{[r, b]}\right)+\frac{1}{\lfloor n b\rfloor-\lfloor n r\rfloor} \sum_{i=\lfloor n r\rfloor+1}^{\lfloor n b\rfloor} d^{2}\left(Y_{i}, \hat{\mu}_{[a, r]}\right).
$$
Define the subsample test  statistics
$$
T_{n}(r ; a, b)=\frac{(r-a)(b-r)}{b-a}\left(\hat{V}_{[a, r]}-\hat{V}_{[r, b]}\right),
$$
and 
$$
T_{n}^C(r ; a, b)=\frac{(r-a)(b-r)}{b-a}\left(\hat{V}_{[r ; a, b]}^{C}-\hat{V}_{[a, r]}-\hat{V}_{[r, b]}\right).
$$
Note  that $T_{n}(r ; a, b)$ and $T_{n}^C(r ; a, b)$ are natural extensions of $T_n(r)$ and $T_n^C(r)$ from two-sample testing problem to change-point detection problem by viewing  $\{Y_t\}_{t=\lfloor na\rfloor+1}^{\lfloor nr\rfloor}$ and $\{Y_t\}_{t=\lfloor nr\rfloor+1}^{\lfloor nb\rfloor}$ as two separated samples.  
Intuitively, the contrast statistics $T_{n}(r ; a, b)$ and $T_{n}^C(r ; a, b)$ are expected to attain their maxima (in absolute value) when $r$ is set at or close to the true change-point location $\tau$. 

\subsection{Test Statistics}
For some trimming parameters $\eta_1$ and $\eta_2$ such that $\eta_1>2\eta_2$, and $\eta_1\in(0,1/2)$, in the same spirit of $D_{n,1}$ and  $D_{n,2}$, and with a bit abuse of notation, we define the testing statistics $$
SN_i= \max _{\lfloor n\eta_1\rfloor\leq k\leq n-\lfloor n\eta_1\rfloor}D_{n,i}(k),\quad i=1,2,
$$
where \begin{flalign*}
	D_{n,1}(k)=& \frac{n\left[T_{n}\left(\frac{k}{n} ; 0,1\right)\right]^2}{ \sum_{l=\lfloor n\eta_2\rfloor}^{k-\lfloor n\eta_2\rfloor} \left[T_{n}\left(\frac{l}{n} ; 0, \frac{k}{n}\right)\right]^2+ \sum_{l=k+\lfloor n\eta_2\rfloor}^{n-\lfloor n\eta_2\rfloor} \left[T_{n}\left(\frac{l}{n} ; \frac{k}{n}, 1\right)\right]^2},\\
	D_{n,2}(k)=& \frac{n\left\{\left[T_{n}\left(\frac{k}{n} ; 0,1\right)\right]^2+\left[T_{n}^C\left(\frac{k}{n} ; 0,1\right)\right]^2\right\}}{ L_n(k)+R_n(k)},
\end{flalign*}
with 
\begin{flalign*}
	L_n(k)=&\sum_{l=\lfloor n\eta_2\rfloor}^{k-\lfloor n\eta_2\rfloor} \left\{\left[T_{n}\left(\frac{l}{n} ; 0, \frac{k}{n}\right)\right]^2+\left[T^C_{n}\left(\frac{l}{n} ; 0, \frac{k}{n}\right)\right]^2 \right\},\\
	R_n(k)=&\sum_{l=k+\lfloor n\eta_2\rfloor}^{n-\lfloor n\eta_2\rfloor} \left\{\left[T_{n}\left(\frac{l}{n} ; \frac{k}{n}, 1\right)\right]^2+ \left[T^C_{n}\left(\frac{l}{n} ; \frac{k}{n}, 1\right)\right]^2\right\}.
\end{flalign*}
The trimming parameter $\eta_1$ plays a similar role as $\eta$ in two-sample testing problem for stabilizing the estimation effect for relatively small sample sizes, while the additional trimming $\eta_2$ is introduced to ensure that  the subsample estimates in the self-normalizers are constructed with the subsample size proportional to  $n$. Furthermore, we note that  the self-normalizers here are modified to accommodate for the unknown change-point location, see \cite{shaozhang2010}, \cite{Zhang2018} for more discussion.



\subsection{Asymptotic Theory}\label{sec:theory_cpt}
\begin{thm}\label{thm_cpt_H0}
	Suppose Assumptions \ref{ass_diff}-\ref{ass_expand} hold. Then, under $\mathbb{H}_0$,  we have for $i=1,2$,
	$$
	SN_i=\max _{\lfloor n\eta_1\rfloor\leq k\leq n-\lfloor n\eta_1\rfloor}D_{n,i}(k) \Rightarrow \sup _{r\in[\eta_1,1-\eta_1]} \frac{[B(r)-rB(1)]^2}{V(r,\eta)}:=\mathcal{S}_\eta,
	$$
	where $V(r,\eta)=\int_{\eta_2}^{r-\eta_2} [B(u)-u/rB(r)]^2du+\int_{r+\eta_2}^{1-\eta_2} [B(1)-B(u)-(1-u)/(1-r)\{B(1)-B(r)\}]^2du$.
\end{thm}
Similar to Theorem \ref{thmtwo_0}, Theorem \ref{thm_cpt_H0} states that both change-point test statistics  have the same pivotal limiting null distribution $\mathcal{S}_{\eta}$.  
The test is thus rejected when   $SN_{i}>Q_{\mathcal{S}_{\eta}}(1-\alpha)$,  $i=1,2$, where $Q_{\mathcal{S}_{\eta}}(1-\alpha)$ denotes the $1-\alpha$ quantile of $\mathcal{S}_{\eta}$.
In Table \ref{tab_cpt_critic}, we tabulate commonly used critical values under various choices of $(\eta_1,\eta_2)$ by simulations.
\begin{table}[H]
	\centering
	\caption{Simulated critical values $Q_{\mathcal{S}_{\eta}}(1-\alpha)$}
	\label{tab_cpt_critic}
	\begin{tabular}{cccc}
		\hline
		\diagbox{$\alpha$}{$(\eta_1,\eta_2)$} & (0.02,0.05)   & (0.04,0.1)   & (0.05,0.15)    \\ \hline
		10\%                                         & 30.29  & 32.09  & 33.36   \\
		5\%                                          & 41.31  & 44.36  & 46.50   \\
		1\%                 &                         72.66 & 79.24 & 82.13  \\
		0.5\%                  &                       91.31 & 96.90 & 101.48\\
		\hline
	\end{tabular}
\end{table}

Recall in Theorem \ref{thmtwo_alter}, we have obtained the local power of two-sample tests $D_{n,1}$ and $D_{n,2}$ at rate $n^{-1/2}$. To this end, consider the local alternative 
$$\mathbb{H}_{an}: V^{(1)}-V^{(2)}=n^{-1/2}\Delta_V,\quad \mbox{and}\quad d^2(\mu^{(1)},\mu^{(2)})=n^{-1/2}\Delta_M,$$ where  $\Delta_V\in\mathbb{R}$ and $\Delta_M\in(0,\infty)$. The following theorem states the asymptotic power behaviors of $SN_1$ and $SN_2$.
\begin{thm}\label{thm_cpt_Ha}
	Suppose Assumptions \ref{ass_diff}-\ref{ass_expand} (with \ref{ass_FCLTo} replaced by \ref{ass_FCLT}) hold. If $\Delta_V\neq 0$ and $\Delta_M\neq 0$ are fixed, then under $\mathbb{H}_{an}$, if $\tau\in(\eta_1,1-\eta_1)$, then as $n\to\infty$, we have 
	\begin{align*}
		&\lim_{|\Delta_V|\to\infty}\lim_{n\to\infty} \left\{\max _{\lfloor n\eta_1\rfloor\leq k\leq n-\lfloor n\eta_1\rfloor}D_{n,1}(k)\right\}\to_p\infty,\\ &\lim_{\max\{|\Delta_V|,\Delta_M\}\to\infty}\lim_{n\to\infty} \left\{\max _{\lfloor n\eta_1\rfloor\leq k\leq n-\lfloor n\eta_1\rfloor}D_{n,2}(k)\right\}\to_p\infty.
	\end{align*}
\end{thm}
{ We note that \Cref{thm_cpt_Ha} deals with the alternative involving two different sequences before and after the change-point, while \Cref{thm_cpt_H0} only involves one stationary sequence. Therefore, we need to replace  \Cref{ass_FCLTo} by \Cref{ass_FCLT}. }

{  \Cref{thm_cpt_Ha} demonstrates  that our tests are capable of detecting local alternatives at rate  $n^{-1/2}$.} In addition, it is seen from Theorem \ref{thm_cpt_Ha} that $SN_1$ is consistent under the local alternative of Fr\'echet variance change  as $|\Delta_V|\to\infty$, while $SN_2$ is  consistent not only under  $|\Delta_V|\to\infty$ but also under the local alternative of Fr\'echet mean change as $\Delta_M\to\infty$.  {Hence $SN_2$ is expected to  capture a wider class of alternatives than $SN_1$, and 
these results are consistent with findings for two-sample problems in Theorem \ref{thmtwo_alter}.}

When $\mathbb{H}_0$ is rejected, it is natural to estimate the change-point location by  
\begin{equation}\label{cpt_est}
	\hat{\tau}_i=n^{-1}\hat{k}_i, \quad \hat{k}_i=\arg\max_{\lfloor n\eta_1\rfloor\leq k\leq n-\lfloor n\eta_1\rfloor}D_{n,i}(k), \quad\mbox{for $i=1,2$.}   
\end{equation}
We will show that the estimators are consistent under the fixed alternative, i.e. $\mathbb{H}_a: V^{(1)}-V^{(2)}=\Delta_V.$ Before that, we need to regulate the behaviour of  Fr\'echet mean and variance under $\mathbb{H}_a$. 

Let 
\begin{flalign*}
	\mu(\alpha)=&\arg\min_{\omega\in\Omega}\left\{\alpha \mathbb{E}(d^2(Y_t^{(1)},\omega))+(1-\alpha)\mathbb{E}(d^2(Y_t^{(2)},\omega))\right\},\\
	V(\alpha)=&\alpha \mathbb{E}(d^2(Y_t^{(1)},\mu(\alpha)))+(1-\alpha)\mathbb{E}(d^2(Y_t^{(2)},\mu(\alpha))),
\end{flalign*}
be the  limiting Fr\'echet mean and variance  of two mixture distributions indexed by $\alpha\in[0,1]$.
\begin{ass}\label{ass_mix}
	$\mu(\alpha)$ is unique for all $\alpha\in[0,1]$, and 
	$$
	|V^{(2)}-V(\alpha)|\geq \varphi(\alpha), \quad |V^{(1)}-V(\alpha)|\geq \varphi(1-\alpha),
	$$
	such that $\varphi(\alpha)\geq 0$ is a continuous, strictly increasing function of $\alpha\in[0,1]$ satisfying $\varphi(0)=0$ and $\varphi(1)\leq |\Delta_V|$.  
\end{ass}
The uniqueness of  Fr\'echet mean and variance for mixture distribution  is also imposed in \cite{dubey2020frechet}, see Assumption (A2) therein. Furthermore, Assumption \ref{ass_mix} imposes a bi-Lipschitz type condition on $V(\alpha)$, and is used to distinguish the Fr\'echet variance $V(\alpha)$ under mixture distribution from  $V^{(1)}$ and $V^{(2)}$. 

\begin{thm}\label{thm_cpt_con}
	Suppose Assumptions \ref{ass_diff}-\ref{ass_expand} (with \ref{ass_FCLTo} replaced by \ref{ass_FCLT}), and Assumption \ref{ass_mix} hold.  Under $\mathbb{H}_a$, for $i=1,2$,
	we have $\hat{\tau}_i\to_p\tau$, where $\hat{\tau}_i$ is defined in \eqref{cpt_est}.
\end{thm}
Theorem \ref{thm_cpt_con} obtains the consistency of $\hat{\tau}_i$, $i=1,2$ when Fr\'echet variance changes. 
We note that it is very challenging to derive the consistency result when $\mathbb{H}_a$ is caused by Fr\'echet mean change alone, which is partly due to the lack of  explicit algebraic structure on $(\Omega,d)$ that we can exploit and the use of self-normalization. We leave this problem for future investigation.

\subsection{Wild Binary Segmentation} \label{wbs}
To detect multiple change-points and identify the their locations  given the time series $\{Y_t\}_{t=1}^n$, we can combine our change-point test with the so-called wild binary segmentation (WBS) \citep{wbs}. The testing procedure in conjunction with WBS can be described as follows. 

Let $I_M = \{ (s_m, e_m)   \}_{m=1,2, \dots, M}$, where $s_m, e_m$ are drawn uniformly from $\{ 0, 1/n, 1/(n-1), \dots, 1/2, 1 \}$ such that $\lceil n e_m \rceil - \lfloor n s_m \rfloor \geq 20$. Then we simulate $J$ i.i.d samples, each sample is of size $n$, from multivariate Gaussian distribution with mean 0 and identity covariance matrix, i.e., for $j=1, 2, \dots, J$,  $\{ Z^{j}_i \}_{i=1}^n \overset{i.i.d.}{\sim} \mathcal{N}(0,1)$. For the $j$th sample $\{ Z^{j}_i \}_{i=1}^n$, let $\widetilde{D}(k; s_m,e_m; \{Z_i^j\}_{i=1}^n)$ be the statistic $D_{\lfloor n e_m \rfloor - \lceil n s_m \rceil +1, 2}(k)$ that is computed based on sample $ \{ Z_{\lceil n s_m \rceil}^j, Z_{ \lceil n s_m \rceil + 1}^j, \dots,$ $Z_{\lfloor n e_m \rfloor}^j \} $ and
\begin{align*}
	\xi_j = \max_{1 \leq m \leq M} \max_{\lfloor \tilde{n}_m \eta_1 \rfloor \leq k \leq \tilde{n}_m - \lfloor \tilde{n}_m \eta_1 \rfloor } \widetilde{D}(k; s_m,e_m;  \{Z_i^j\}_{i=1}^n),
\end{align*}
where $ \tilde{n}_m = \lceil n e_m \rceil - \lfloor n s_m \rfloor +1 $. Setting $\xi$ as the $95\%$ quantile of $ \xi_1, \xi_2, \dots, \xi_J $, we can apply our test in combination with WBS algorithm to the data sequence
$\{Y_1, Y_2, \dots Y_n\}$ by running Algorithm \ref{alg:WBS} as $\text{WBS}(0, 1, \xi)$.  The main rational behind this algorithm is that we exploit the asymptotic pivotality of our SN test statistic, and the limiting null distribution of our test statistic applied to random objects is identical to that applied to i.i.d $\mathcal{N}$(0,1) random variables.  
Thus this threshold is expected to well approximate the
95\% quantile of the finite sample distribution of the maximum SN test statistic on the $M$ random intervals under the null.

\begin{algorithm}
	\SetKwProg{Fn}{Function}{:}{}
	\Fn{WBS(s, e, $\xi$)}{
		\uIf {$\lceil n s \rceil - \lfloor n e \rfloor < 20$ }{
			STOP;
		}\Else{
			$\mathcal{M}_{s,e} \leftarrow$ set of those $1 \leq m \leq M$ for which $s \leq s_m, e_m \leq e$; \\ 
			$m_0 \leftarrow \arg \max_{m \in \mathcal{M}_{s,e}} \max_{\lfloor \tilde{n}_m \eta_1 \rfloor \leq k \leq \tilde{n}_m - \lfloor \tilde{n}_m \eta_1 \rfloor } \widetilde{D}(k;s_m,e_m; \{Y_i\}_{i=1}^n)$, where $\tilde{n}_m = \lceil n e_m \rceil - \lfloor n s_m \rfloor +1$; \\
			$k_0 \leftarrow \max_{\lfloor \tilde{n}_0 \eta_1 \rfloor \leq k \leq \tilde{n}_0 - \lfloor \tilde{n}_0 \eta_1 \rfloor } \widetilde{D}(k;s_{m_0},e_{m_0}; \{Y_i\}_{i=1}^n)$, where $\tilde{n}_0 = \lceil n e_{m_0} \rceil - \lfloor n s_{m_0} \rfloor +1$; \\
			\uIf {$ \widetilde{D}(k_0;s_{m_0},e_{m_0}; \{Y_i\}_{i=1}^n) > \xi$}{ add $k_0$ to the set of estimated change points; \\
				WBS($s,k_0/n,\xi$); \\
				WBS($k_0/n,e,\xi$);
			}\Else{
				STOP
			}
		}
	}
	\caption{WBS}
	\label{alg:WBS}
\end{algorithm}

\section{Simulation}\label{sec:simu}
In this section, we examine the size and power performance of our proposed tests in two-sample testing (Section \ref{sec:simu_two}), change-point detection (Section \ref{sec:simu_CP}) problems, and provide simulation results of  WBS based change-point estimation (Section \ref{sec:simu_WBS}).  We refer to Appendix \ref{sec:simu_functional} with additional simulation results regarding comparison with FPCA approach for two-sample tests in functional time series.

The time series random objects considered in this section include (i). univariate Gaussian probability distributions equipped with 2-Wasserstein metric $d_W$;  (ii). graph Laplacians of weighted graphs equipped with Frobenius metric $d_F$;   (iii). covariance matrices \citep{dryden2009non} equipped with log-Euclidean metric $d_E$. Numerical experiments are conducted according to the following data generating processes (DGPs):
\begin{itemize}
	\item[(i)]  Gaussian univariate probability distribution: we consider 
	\begin{align*}
	    &Y_t^{(1)}=\mathcal{N}(\arctan (U_{t,1}),[\arctan(U_{t,1}^2)+1]^2), \\
	    & Y_t^{(2)}=\mathcal{N}(\arctan (U_{t,2})+\delta_1, \delta_2^2[\arctan(U_{t,2}^2)+1]^2).
	\end{align*}
	\item[(ii)] graph Laplacians:  each graph has $N$ nodes ($N=10$ for two-sample test and $N=5$ for change-point test)  that are categorized into two communities with $0.4N$ and $0.6N$ nodes respectively, and the edge weight for the first  community, the second community and between community are set as $0.4+\arctan(U_{t,1}^2)$, $0.2+\arctan(U_{t,1}^{'2}), 0.1$ for the first sample $Y_t^{(1)}$, and $\delta_2[0.4+\arctan(U_{t,2}^{2})]$, $\delta_2[0.2+\arctan(U_{t,2}^{'2})], 0.1+\delta_1$ for the   second sample $Y_t^{(2)}$, respectively;
	\item[(iii)] covariance matrix: $Y_t^{(i)}=(2I_3+Z_{t,i})(2I_3+Z_{t,i})^{\top}$, $i=1,2$, such that all the entries of    $Z_{t,1}$ (resp. $Z_{t,2}$) are independent copies of  $\arctan(U_{t,1} )$  (resp. $\delta_1+\delta_2\arctan(U_{t,2})$).
\end{itemize}
For  DGP (i)-(iii),  $(U_{t,1},U_{t,2})^{\top}$ (with independent copies $(U'_{t,1},U'_{t,2})^{\top}$)  are generated according to the following VAR(1) process,

\begin{equation}\label{eq_U}
	\left(
	\begin{matrix}
		U_{t,1}\\
		U_{t,2}
	\end{matrix}
	\right)=\rho\left(
	\begin{matrix}
		U_{t-1,1}\\
		U_{t-1,2}
	\end{matrix}
	\right)+\epsilon_t,  \quad \epsilon_t\overset{i.i.d.}{\sim}\mathcal{N}\left(0,\left(\begin{matrix}
		1 & a\\
		a&1
	\end{matrix}\right)\right);
\end{equation}
where $a\in\{0,0.5\}$  measures the cross-dependence, and $\rho \in \{-0.4,0,0.4,0.7\}$ measures the temporal dependence within each sample (or each segment in change-point testing). For size evaluation in change-point tests, only $\{Y_t^{(1)}\}$ is used.

Furthermore, $\delta_1\in[0,0.3]$ and $\delta_2\in[0.7,1]$ are used to characterize the change in the underlying distributions. In particular, $\delta_1$  can only capture the location shift, while $\delta_2$ measures the scale change, and the case $(\delta_1,\delta_2)=(0,1)$ corresponds to $\mathbb{H}_0$.
For DGP (i) and (ii), i.e. Gaussian distribution with 2-Wasserstein metric $d_W$ and graph Laplacians with Euclidean metric $d_F$,  the location  parameter $\delta_1$ directly shifts Fr\'echet mean while keeping  Fr\'echet variance constant; and the scale  parameter $\delta_2$   works on  Fr\'echet variance only while holding the Fr\'echet mean fixed. For DGP (iii), i.e. covariance matrices, the log-Euclidean metric $d_E$ operates nonlinearly, and thus changes in either $\delta_1$ or $\delta_2$ will be reflected on  changes in both  Fr\'echet mean  and variance.

The comparisons of our proposed methods with \cite{dubey2019frechet}  for two-sample testing and \cite{dubey2020frechet} for change-point testing are also reported, which are generally referred to as $\mathrm{DM}$.

\subsection{Two-Sample Test}\label{sec:simu_two}
For the two-sample testing problems, we set the sample size as $n_1=n_2\in\{50,100,200,400\}$, and trimming parameter as $\eta=0.15$. 
Table \ref{tab_two} presents the sizes of our tests and DM  test for three DGPs  based on 1000 Monte Carlo replications at nominal significance level $\alpha=5\%$.

In all three {subtables}, we see that: (a) both $D_1$ and $D_2$ can deliver reasonable size under all settings; (b)  $\mathrm{DM}$  suffers from severe size distortion when dependence magnitude 
among data is strong; (c) when two samples are  dependent, i.e. $a=0.5$, $\mathrm{DM}$ is a bit undersized even when data is  temporally independent. These findings suggest that our SN-based tests provide more accurate size relative to $\mathrm{DM}$ when either within-group temporal dependence or between-group  dependence is exhibited.

\begin{table}[H]
	\centering
	\caption{Size Performance ($\times 100\%$) at $\alpha=$5\%  for all three DGPs.}
	\label{tab_two}
	\resizebox{\textwidth}{!}{
		\begin{tabular}{cccccccccccccccccccccc}
			\hline
			&       & \multicolumn{6}{c}{Gaussian Distribution based on $d_W$}                                      &  & \multicolumn{6}{c}{Graph Laplacian based on $d_F$}                                            &  & \multicolumn{6}{c}{Covariance Matrix based on $d_E$}                                          \\ \cline{1-8} \cline{10-15} \cline{17-22} 
			$\rho$               & $n_i$ & \multicolumn{2}{c}{$D_1$} & \multicolumn{2}{c}{$D_2$} & \multicolumn{2}{c}{DM} &  & \multicolumn{2}{c}{$D_1$} & \multicolumn{2}{c}{$D_2$} & \multicolumn{2}{c}{DM} &  & \multicolumn{2}{c}{$D_1$} & \multicolumn{2}{c}{$D_2$} & \multicolumn{2}{c}{DM} \\ \cline{1-8} \cline{10-15} \cline{17-22} 
			a                    &       & 0           & 0.5         & 0           & 0.5         & 0          & 0.5       &  & 0            & 0.5        & 0           & 0.5         & 0          & 0.5       &  & 0           & 0.5         & 0           & 0.5         & 0          & 0.5       \\ \hline
			-0.4                 & 50    & 6.1         & 7.1         & 6.2         & 7.1         & 10.1       & 7.0       &  & 5.6          & 7.2        & 4.4         & 5.5         & 9.4        & 6.7       &  & 6.4         & 6.0         & 6.1         & 6.4         & 11.4       & 5.5       \\
			& 100   & 4.6         & 5.2         & 4.6         & 5.2         & 7.4        & 6.7       &  & 5.8          & 5.3        & 5.1         & 4.8         & 8.0        & 5.4       &  & 5.8         & 6.0         & 5.7         & 6.0         & 9.2        & 6.0       \\
			& 200   & 5.0         & 5.1         & 5.1         & 5.2         & 8.9        & 6.4       &  & 5.7          & 4.8        & 5.3         & 4.1         & 6.7        & 4.7       &  & 5.7         & 5.7         & 5.8         & 5.8         & 8.3        & 5.8       \\
			& 400   & 4.1         & 5.1         & 4.2         & 5.1         & 8.2        & 6.0       &  & 4.4          & 4.2        & 4.2         & 4.2         & 6.5        & 5.3       &  & 5.0         & 5.8         & 4.9         & 5.9         & 9.4        & 5.5       \\ \hline
			\multirow{4}{*}{0}   & 50    & 4.5         & 5.5         & 4.8         & 4.9         & 5.0        & 4.2       &  & 4.6          & 5.9        & 3.9         & 4.9         & 5.3        & 5.6       &  & 5.9         & 5.8         & 5.3         & 5.7         & 5.9        & 4.0       \\
			& 100   & 3.9         & 4.9         & 3.8         & 4.8         & 4.4        & 3.2       &  & 5.8          & 4.6        & 4.7         & 4.4         & 4.8        & 3.3       &  & 5.0         & 5.8         & 4.8         & 5.6         & 5.0        & 3.6       \\
			& 200   & 5.9         & 6.0         & 6.0         & 5.9         & 5.9        & 2.4       &  & 5.1          & 6.1        & 5.0         & 5.8         & 4.8        & 2.8       &  & 6.2         & 4.9         & 6.2         & 4.6         & 5.0        & 2.4       \\
			& 400   & 5.5         & 4.8         & 5.3         & 4.8         & 4.5        & 2.8       &  & 4.6          & 4.0        & 4.6         & 3.7         & 4.8        & 3.5       &  & 5.7         & 4.9         & 5.7         & 4.8         & 4.8        & 2.7       \\ \hline
			\multirow{4}{*}{0.4} & 50    & 5.0         & 5.2         & 5.1         & 4.4         & 9.7        & 6.8       &  & 4.6          & 4.6        & 4.3         & 3.9         & 7.6        & 6.2       &  & 7.0         & 6.3         & 7.0         & 5.3         & 12.8       & 7.1       \\
			& 100   & 6.5         & 4.7         & 5.7         & 5.0         & 9.8        & 5.1       &  & 5.8          & 5.8        & 5.2         & 5.1         & 8.2        & 5.8       &  & 5.9         & 6.4         & 5.8         & 6.3         & 9.4        & 6.5       \\
			& 200   & 4.8         & 4.4         & 4.7         & 4.2         & 10.7       & 4.8       &  & 6.2          & 5.0        & 5.3         & 4.7         & 6.2        & 5.7       &  & 6.5         & 5.8         & 6.3         & 5.2         & 10.0       & 6.7       \\
			& 400   & 5.3         & 5.6         & 4.8         & 5.5         & 9.3        & 6.0       &  & 6.2          & 4.9        & 5.7         & 4.7         & 8.5        & 5.3       &  & 5.8         & 4.6         & 5.6         & 4.1         & 10.2       & 6.6       \\ \hline
			\multirow{4}{*}{0.7} & 50    & 6.4         & 8.1         & 7.1         & 7.1         & 30.1       & 21.1      &  & 4.8          & 5.9        & 6.1         & 6.2         & 12.1       & 9.7       &  & 6.3         & 7.8         & 8.3         & 7.4         & 33.3       & 20.9      \\
			& 100   & 7.8         & 6.2         & 7.1         & 5.1         & 27.9       & 18.4      &  & 5.0          & 4.6        & 5.3         & 5.3         & 12.0       & 8.2       &  & 6.5         & 7.5         & 5.8         & 6.7         & 26.7       & 18.5      \\
			& 200   & 6.3         & 4.2         & 5.3         & 3.9         & 23.9       & 17.1      &  & 4.9          & 5.4        & 5.3         & 4.6         & 10.0       & 7.0       &  & 5.8         & 6.2         & 5.5         & 5.3         & 24.5       & 21.7      \\
			& 400   & 4.6         & 5.7         & 3.7         & 4.8         & 23.6       & 18.7      &  & 4.9          & 5.4        & 4.2         & 4.8         & 10.3       & 7.3       &  & 4.5         & 5.5         & 4.3         & 4.5         & 24.2       & 19.3      \\ \hline
	\end{tabular}
}
\end{table}
In Figure \ref{Fig:power_two}, we further compare size-adjusted power of our SN-based tests  and DM  test, in view of   the size-distortion of DM. {That is, the critical values  are set as the empirical 95\% quantiles of the test statistics obtained in the size evaluation, so that all curves start from the nominal level at 5\%.  } For all settings, we note that  $D_2$ is more powerful  than (or equal to) $D_1$. In particular, $D_1$ has trivial power in DGP (i) and (ii) when only  Fr\'echet mean difference is present. In addition, $D_2$ is more powerful in detecting  Fr\'echet mean differences than DM for DGP (i) and (ii), and beats DM in DGP (i) for detecting Fr\'echet variance differences, although it is slightly worse than DM in detecting Fr\'echet variance differences for  DGP (ii) and (iii). Due to robust size and power performance, we thus recommend $D_2$ for practical purposes.  

\begin{figure}[H]
	\centering 
	\begin{subfigure}{0.32\textwidth}
		\centering
		\includegraphics[width=1\textwidth]{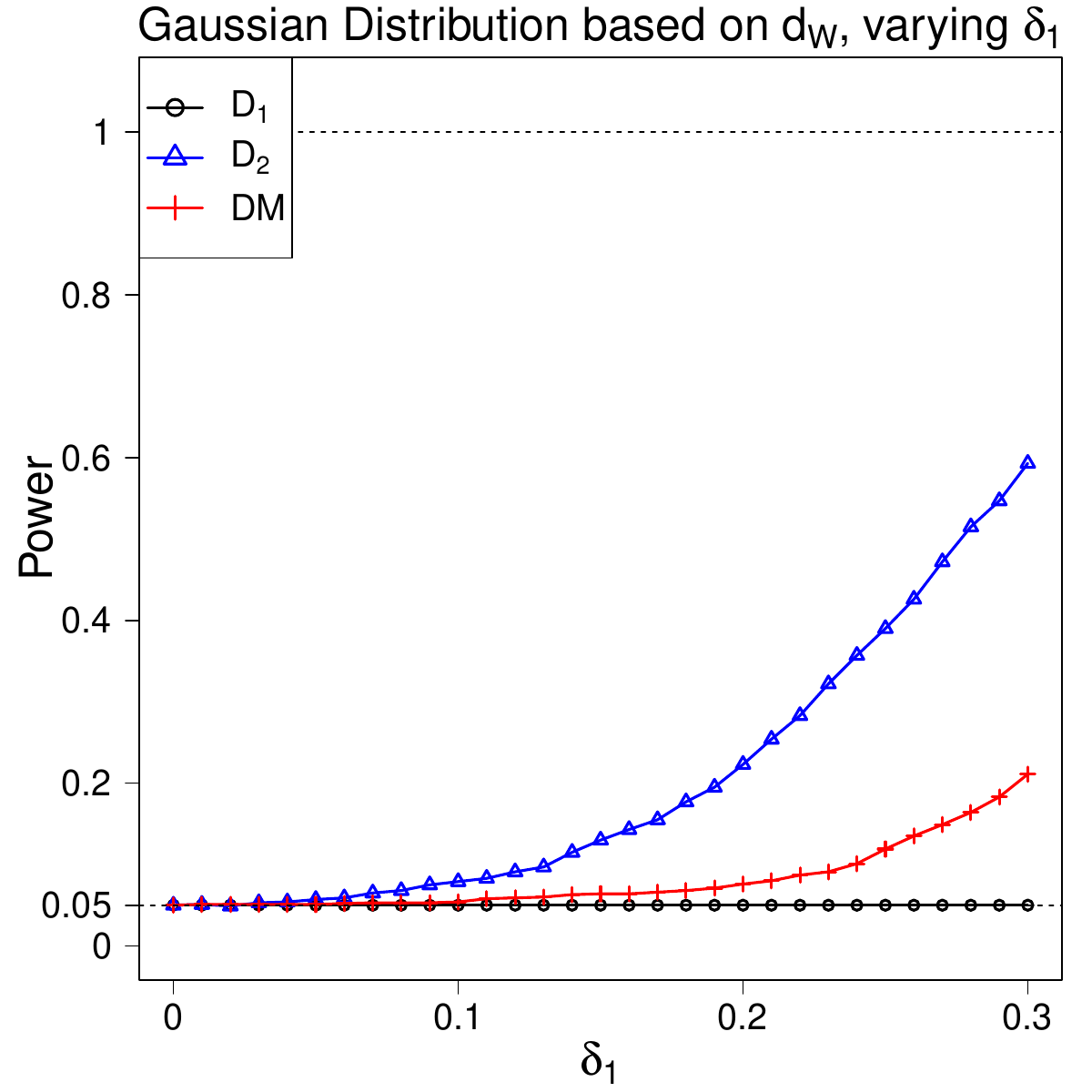}
	\end{subfigure}
	\begin{subfigure}{0.32\textwidth}
		\centering
		\includegraphics[width=1\textwidth]{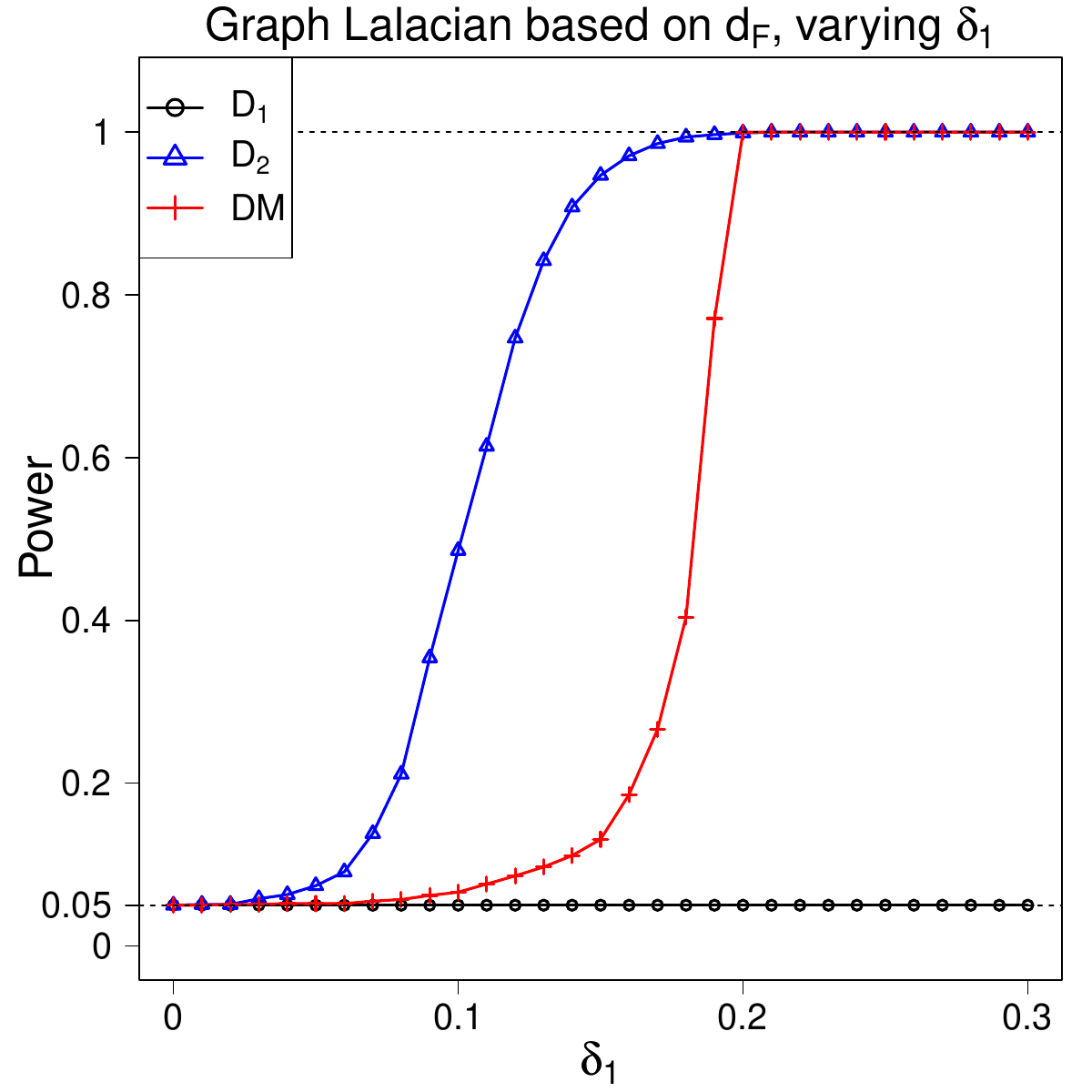}
	\end{subfigure}
	\begin{subfigure}{0.32\textwidth}
		\centering
		\includegraphics[width=1\textwidth]{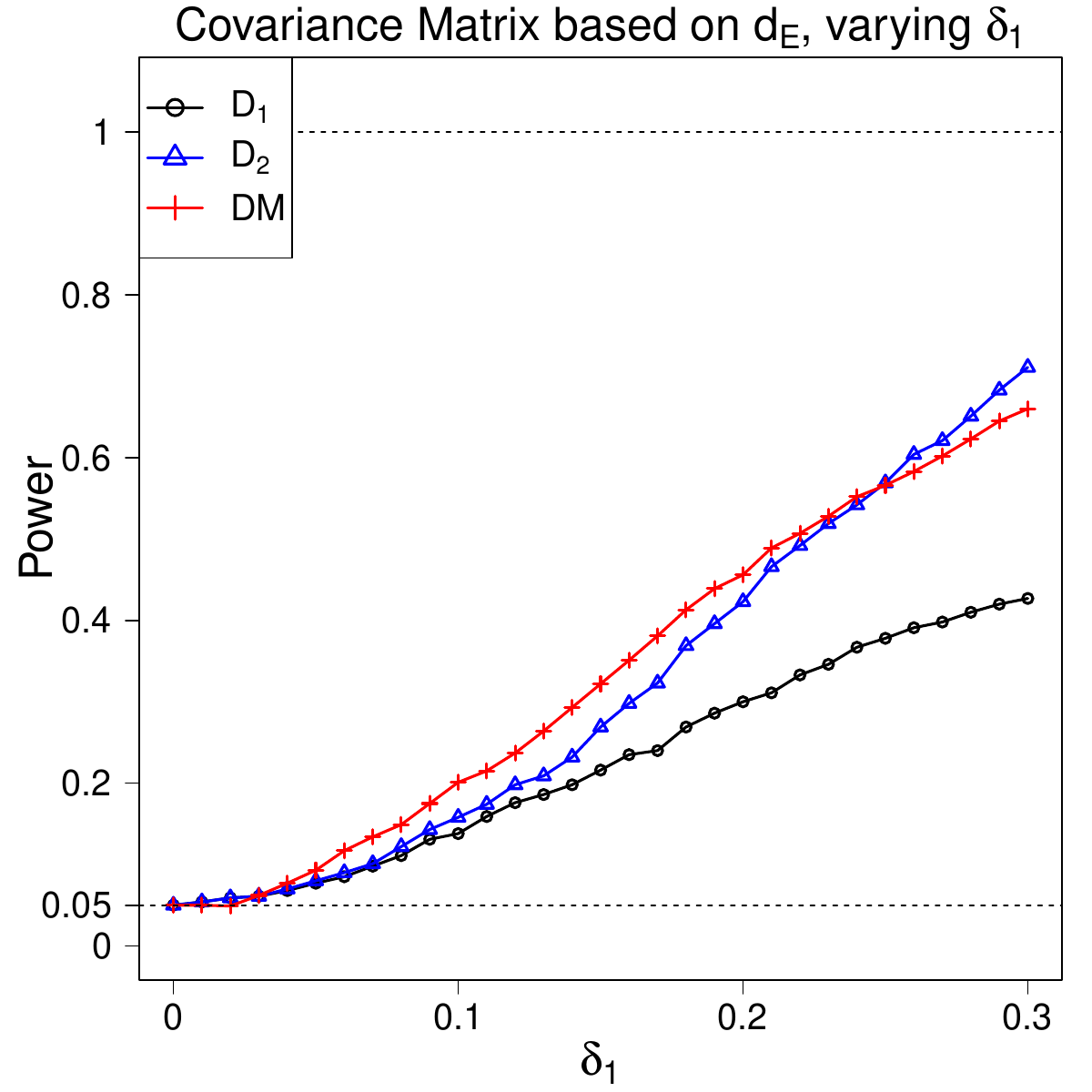}
	\end{subfigure}
	\centering 
	\begin{subfigure}{0.32\textwidth}
		\centering
		\includegraphics[width=1\textwidth]{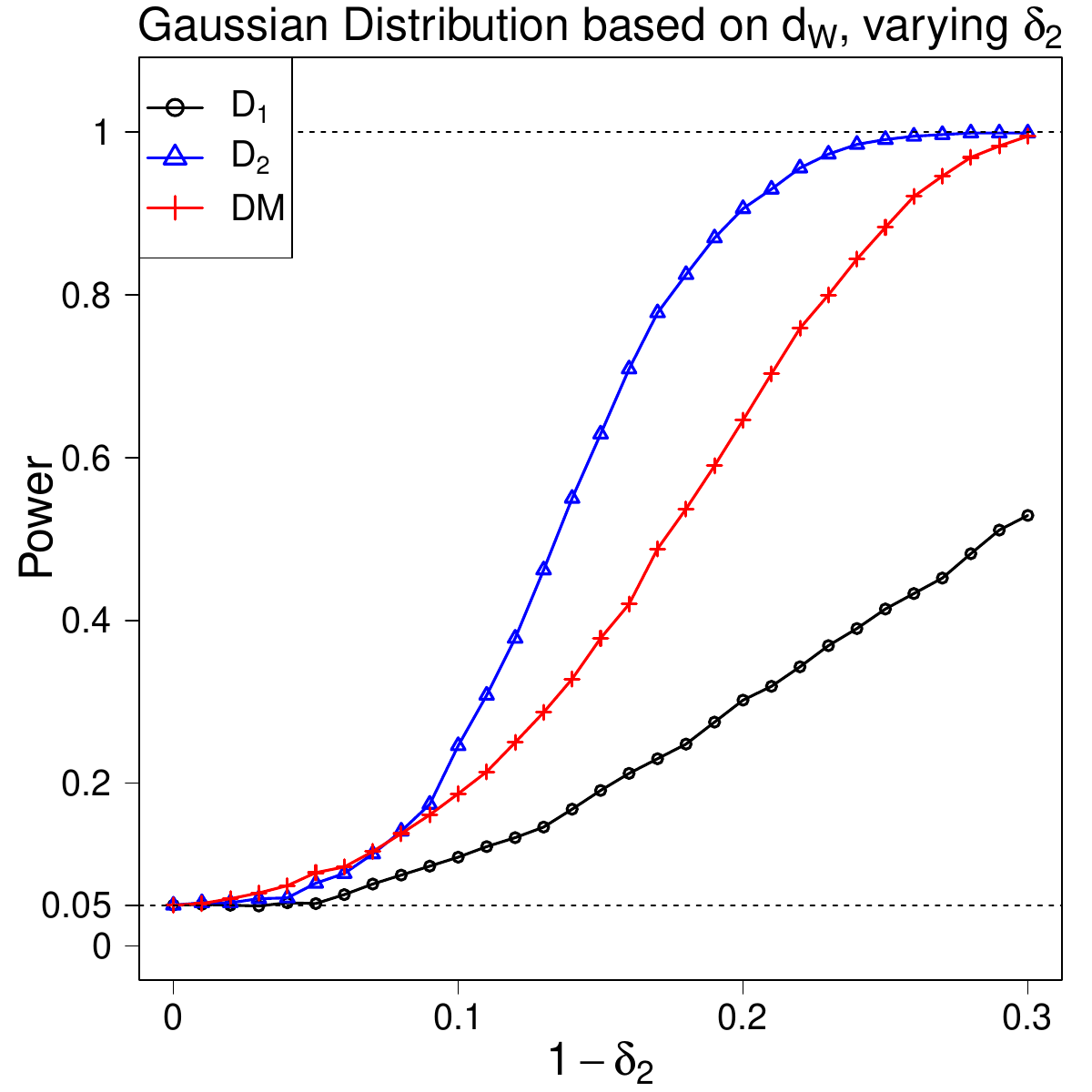}
	\end{subfigure}
	\begin{subfigure}{0.32\textwidth}
		\centering
		\includegraphics[width=1\textwidth]{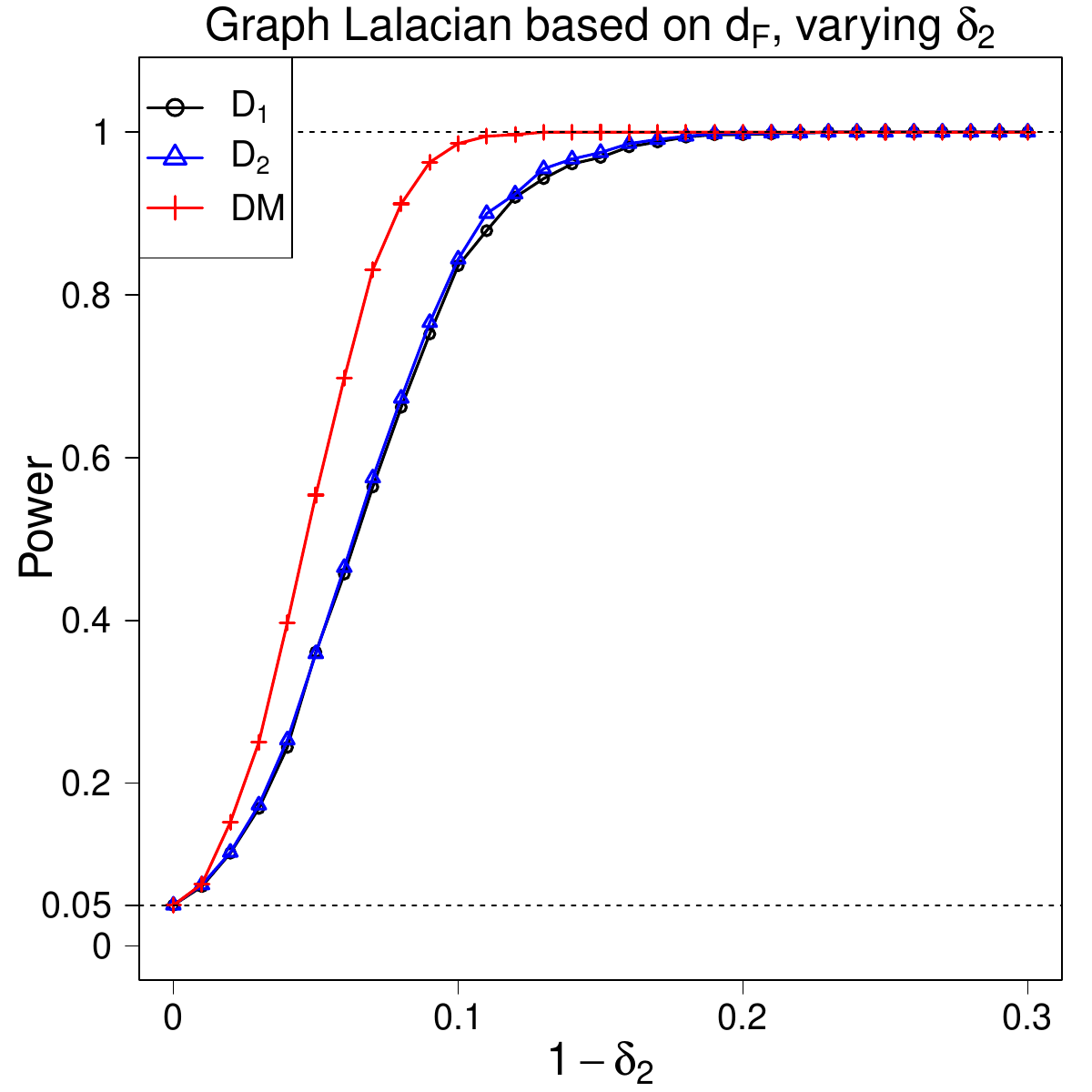}
	\end{subfigure}
	\begin{subfigure}{0.32\textwidth}
		\centering
		\includegraphics[width=1\textwidth]{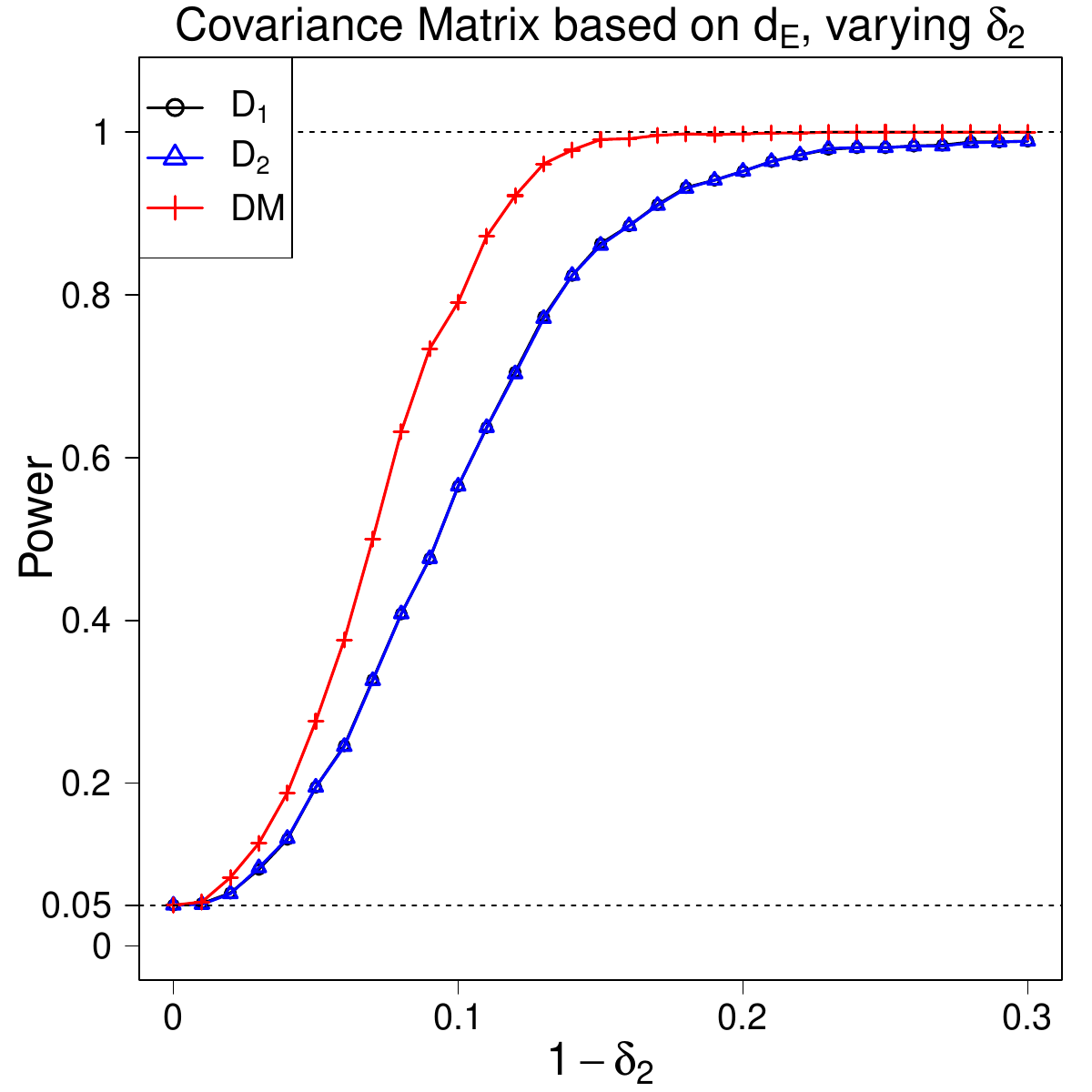}
	\end{subfigure}
	\caption{Size-Adjusted Power ($\times 100\%$) at $\alpha=5\%$, Two-Sample Test for all three DGPs, $n_i=400$ and $\rho=0.4$}
	\label{Fig:power_two}
\end{figure}

\subsection{Change-Point Test}
\label{sec:simu_CP}

For the change-point testing problems, we set the sample size $n\in \{200, 400,800\}$, and trimming
parameter as $(\eta_1,\eta_2)=(0.15,0.05)$. Table \ref{tab_cpt_size} outlines the size performance of our tests and DM test for three DGPs 
based on 1000 Monte Carlo replications at nominal significance level $\alpha=5\%$.  DM tests based asymptotic critical value and bootstraps (with 500 replications) are  denoted as $\mathrm{DM}^a$ and $\mathrm{DM}^b$, respectively. 

From Table \ref{tab_cpt_size}, we find that $SN_1$ always exhibits accurate size while $SN_2$ is a bit conservative. As a comparison, {the tests
based on $\mathrm{DM}^a$ and $\mathrm{DM}^b$  suffer}  from severe distortion when strong temporal  dependence is present, although $\mathrm{DM}^b$ is slightly better than $\mathrm{DM}^a$ in DGP (i) and (ii).

\begin{table}[H]
	\centering
	\caption{Size Performance ($\times 100\%$) at $\alpha=$5\%}
	\label{tab_cpt_size}
	\resizebox{\textwidth}{!}{
		\begin{tabular}{ccccccccccccccccccccccc}
			\hline
			&     &  & \multicolumn{6}{c}{Gaussian Distribution based on $d_W$}                      &  & \multicolumn{6}{c}{Graph Laplacian based on $d_F$}                        &  & \multicolumn{6}{c}{Covariance Matrix based on $d_E$}                       \\ \cline{1-2} \cline{4-9} \cline{11-16} \cline{18-23} 
			$\rho$                & $n$ &  & $SN_{1}$    &     & $SN_{2}$    &    & $\mathrm{DM}^a$ & $\mathrm{DM}^b$ &  & $SN_{1}$   &    & $SN_{2}$   &   & $\mathrm{DM}^a$ & $\mathrm{DM}^b$ &  & $SN_{1}$   &    & $SN_{2}$   &    &  $\mathrm{DM}^a$ & $\mathrm{DM}^b$ \\ \cline{1-2} \cline{4-9} \cline{11-16} \cline{18-23} 
			\multirow{3}{*}{-0.4} & 200 &  & 5.5        &     & 3.9        &    & 15.0              & 1.9              &  & 5.8       &    & 3.4       &   & 28.9              & 3.4              &  & 6.0       &    & 4.6       &    & 3.5               & 5.0              \\
			& 400 &  & 5.5        &     & 4.5        &    & 13.1              & 6.7              &  & 4.6       &    & 1.8       &   & 18.9              & 4.5              &  & 5.6       &    & 5.1       &    & 3.6               & 5.8              \\
			& 800 &  & 4.4        &     & 3.9        &    & 13.5              & 10.3             &  & 5.1       &    & 3.4       &   & 12.3              & 7.2              &  & 4.8       &    & 4.4       &    & 3.7               & 5.3              \\ \cline{1-2} \cline{4-9} \cline{11-16} \cline{18-23} 
			\multirow{3}{*}{0}    & 200 &  & 4.9        &     & 2.2        &    & 9.2               & 1.0              &  & 5.1       &    & 1.6       &   & 13.1              & 1.0              &  & 6.2       &    & 3.5       &    & 3.1               & 4.1              \\
			& 400 &  & 5.4        &     & 2.3        &    & 5.7               & 2.3              &  & 5.6       &    & 1.7       &   & 9.4               & 2.3              &  & 5.3       &    & 3.7       &    & 4.2               & 5.8              \\
			& 800 &  & 5.4        &     & 3.6        &    & 6.1               & 4.6              &  & 5.0       &    & 2.6       &   & 6.5               & 3.6              &  & 4.7       &    & 3.9       &    & 3.4               & 5.7              \\ \cline{1-2} \cline{4-9} \cline{11-16} \cline{18-23} 
			\multirow{3}{*}{0.4}  & 200 &  & 4.3        &     & 3.9        &    & 44.7              & 17.1             &  & 5.9       &    & 2.4       &   & 27.7              & 3.6              &  & 5.4       &    & 1.1       &    & 7.9               & 10.1             \\
			& 400 &  & 4.6        &     & 2.2        &    & 29.8              & 17.1             &  & 6.3       &    & 1.9       &   & 17.9              & 4.0              &  & 5.2       &    & 1.7       &    & 6.1               & 9.2              \\
			& 800 &  & 6.6        &     & 2.1        &    & 20.4              & 17.3             &  & 5.5       &    & 2.1       &   & 13.4              & 6.0              &  & 5.3       &    & 3.7       &    & 6.8               & 8.3              \\ \cline{1-2} \cline{4-9} \cline{11-16} \cline{18-23} 
			\multirow{3}{*}{0.7}  & 200 &  & 5.9        &     & 10.6       &    & 91.4              & 66.0             &  & 5.4       &    & 5.8       &   & 68.9              & 20.2             &  & 7.9       &    & 0.3       &    & 29.5              & 35.3             \\
			& 400 &  & 4.1        &     & 5.8        &    & 84.5              & 69.7             &  & 5.3       &    & 4.0       &   & 53.5              & 22.5             &  & 5.9       &    & 0.7       &    & 22.7              & 28.9             \\
			& 800 &  & 5.6        &     & 3.9        &    & 77.8              & 70.4             &  & 4.4       &    & 2.3       &   & 40.8              & 25.8             &  & 5.5       &    & 1.4       &    & 20.4              & 26.1             \\ \hline
		\end{tabular}
	}
\end{table}

In Figure \ref{Fig:power_cpt}, we plot the size-adjusted power of our tests and DM test based on bootstrap calibration. Here, the size-adjusted power of DM$^b$ is implemented following \cite{dominguez2000size}.  Similar to the findings in change-point tests, we find that $SN_1$ has trivial power in DGP (i) and (ii) when  there is only Fr\'echet mean change and is worst among all three tests.  Furthermore, $SN_2$ is slightly less powerful compared to DM and the power loss is moderate. Considering its better size control, $SN_2$ is preferred.

\begin{figure}[H]
	\centering 
	\begin{subfigure}{0.32\textwidth}
		\centering
		\includegraphics[width=1\textwidth]{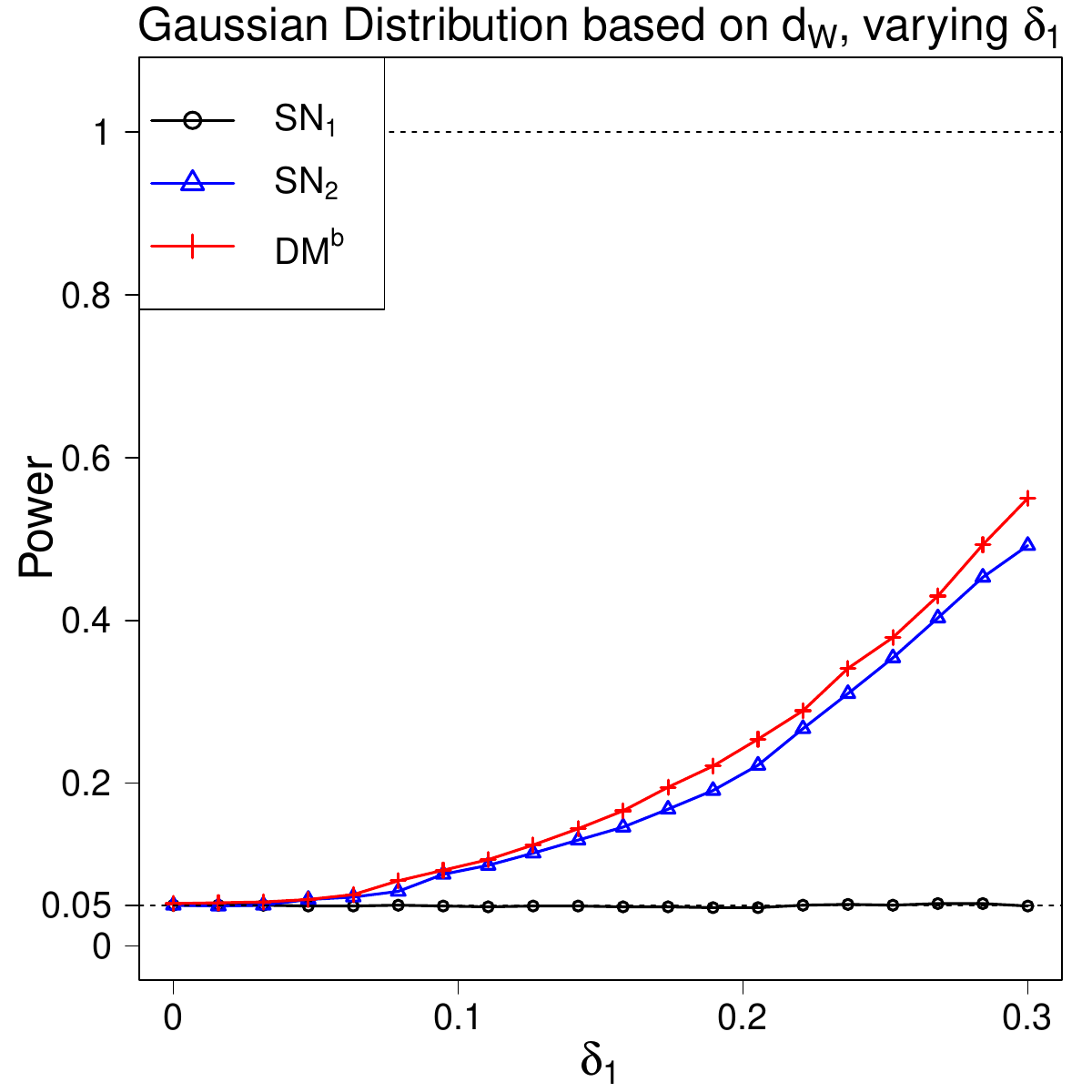}
	\end{subfigure}
	\begin{subfigure}{0.32\textwidth}
		\centering
		\includegraphics[width=1\textwidth]{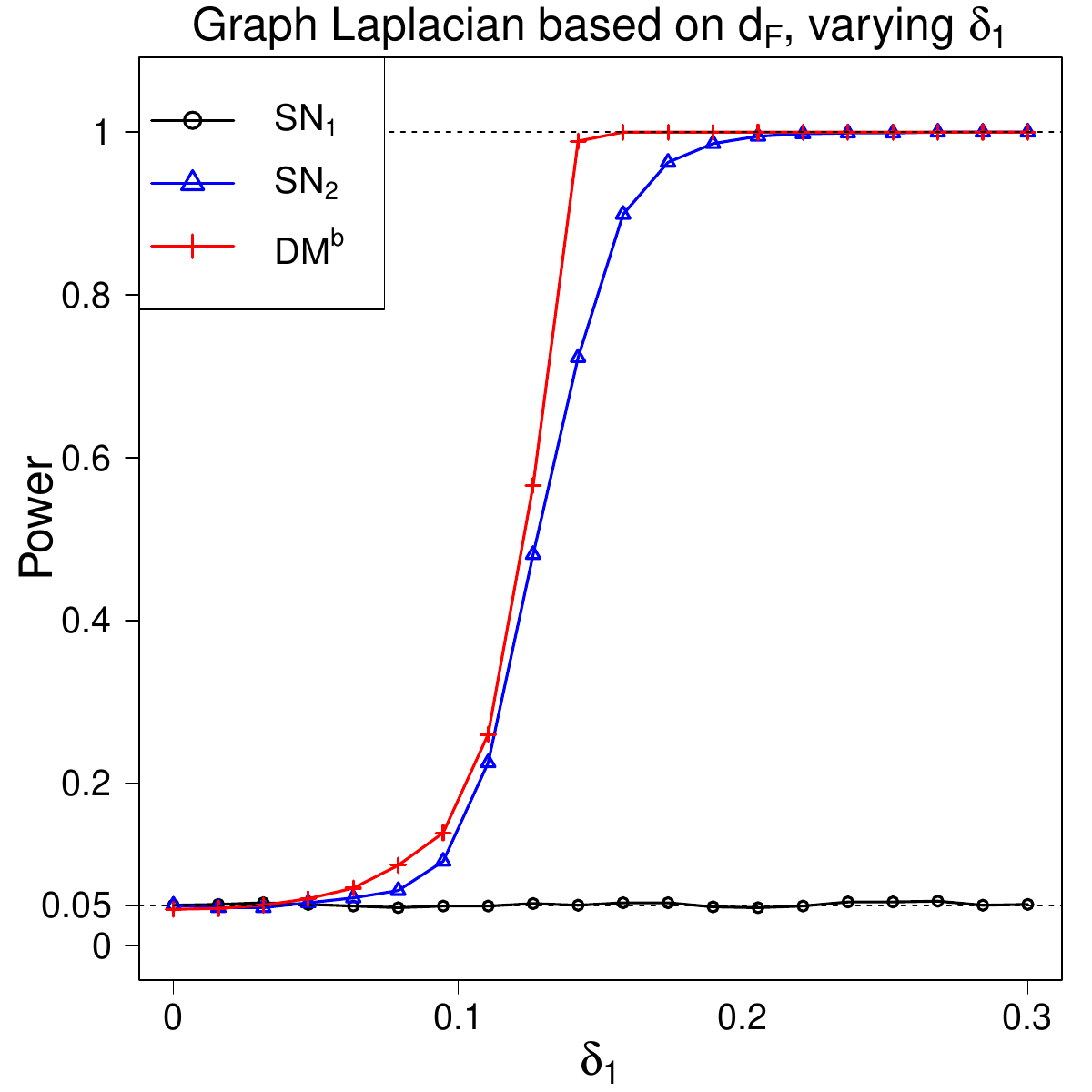}
	\end{subfigure}
	\begin{subfigure}{0.32\textwidth}
		\centering
		\includegraphics[width=1\textwidth]{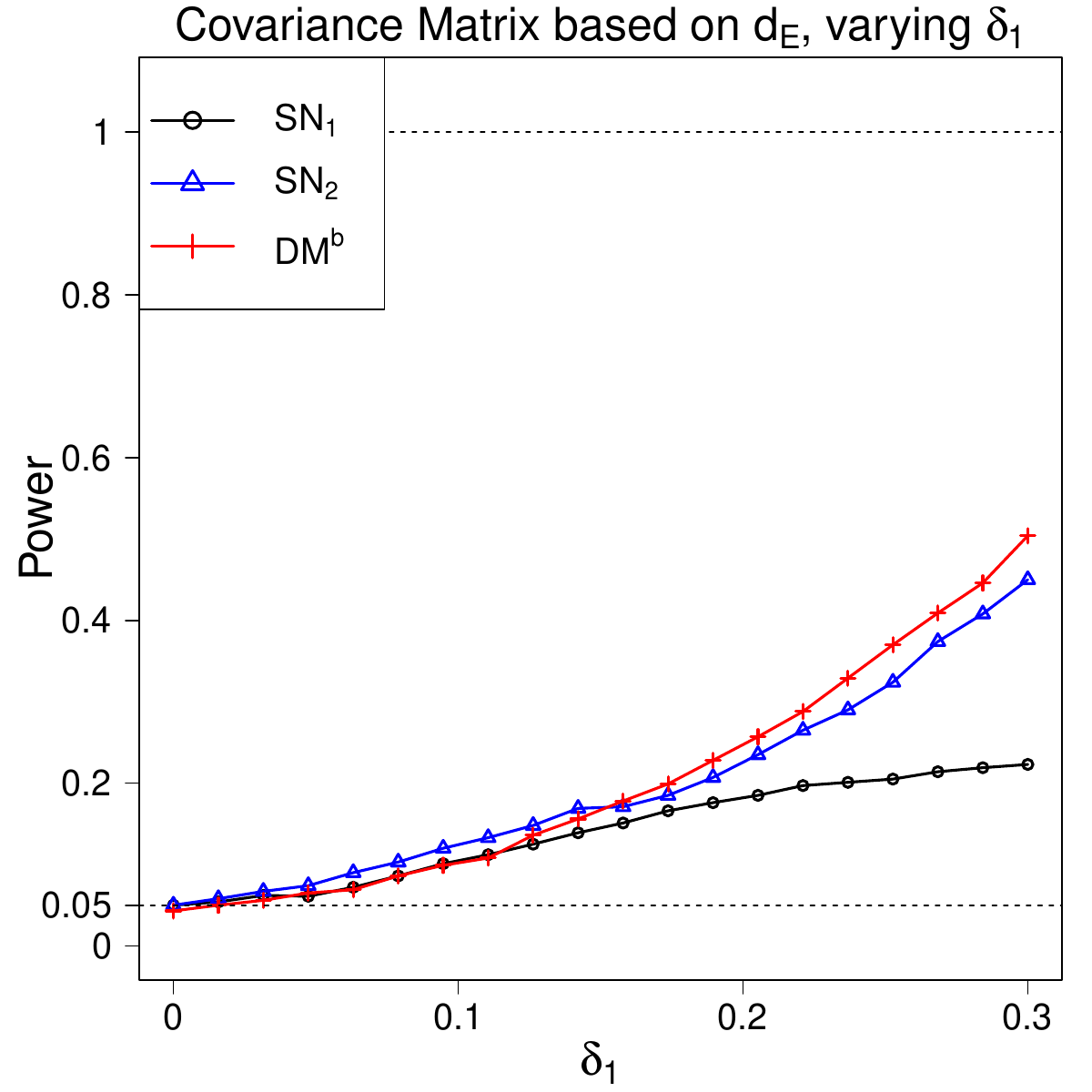}
	\end{subfigure}
	\centering 
	\begin{subfigure}{0.32\textwidth}
		\centering
		\includegraphics[width=1\textwidth]{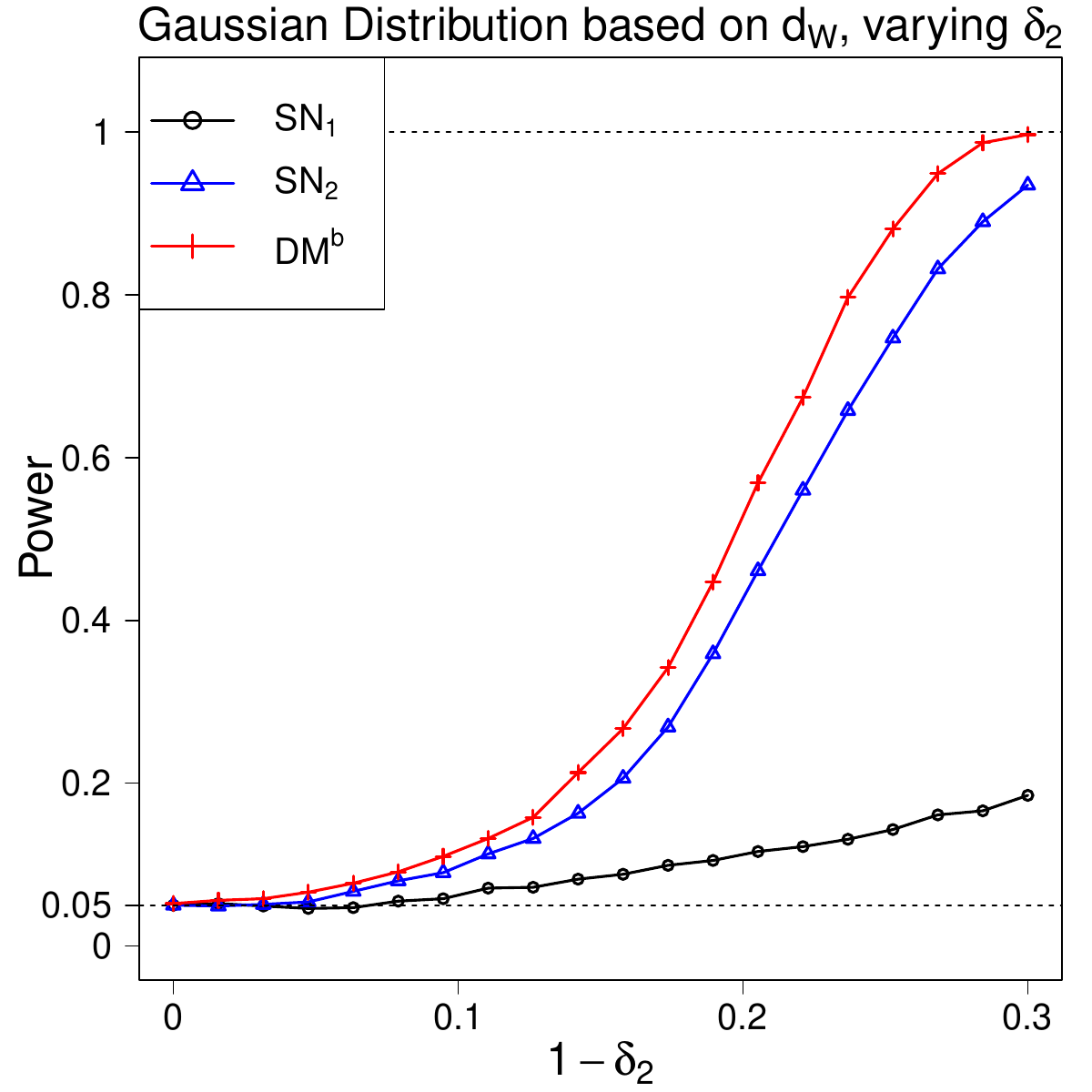}
	\end{subfigure}
	\begin{subfigure}{0.32\textwidth}
		\centering
		\includegraphics[width=1\textwidth]{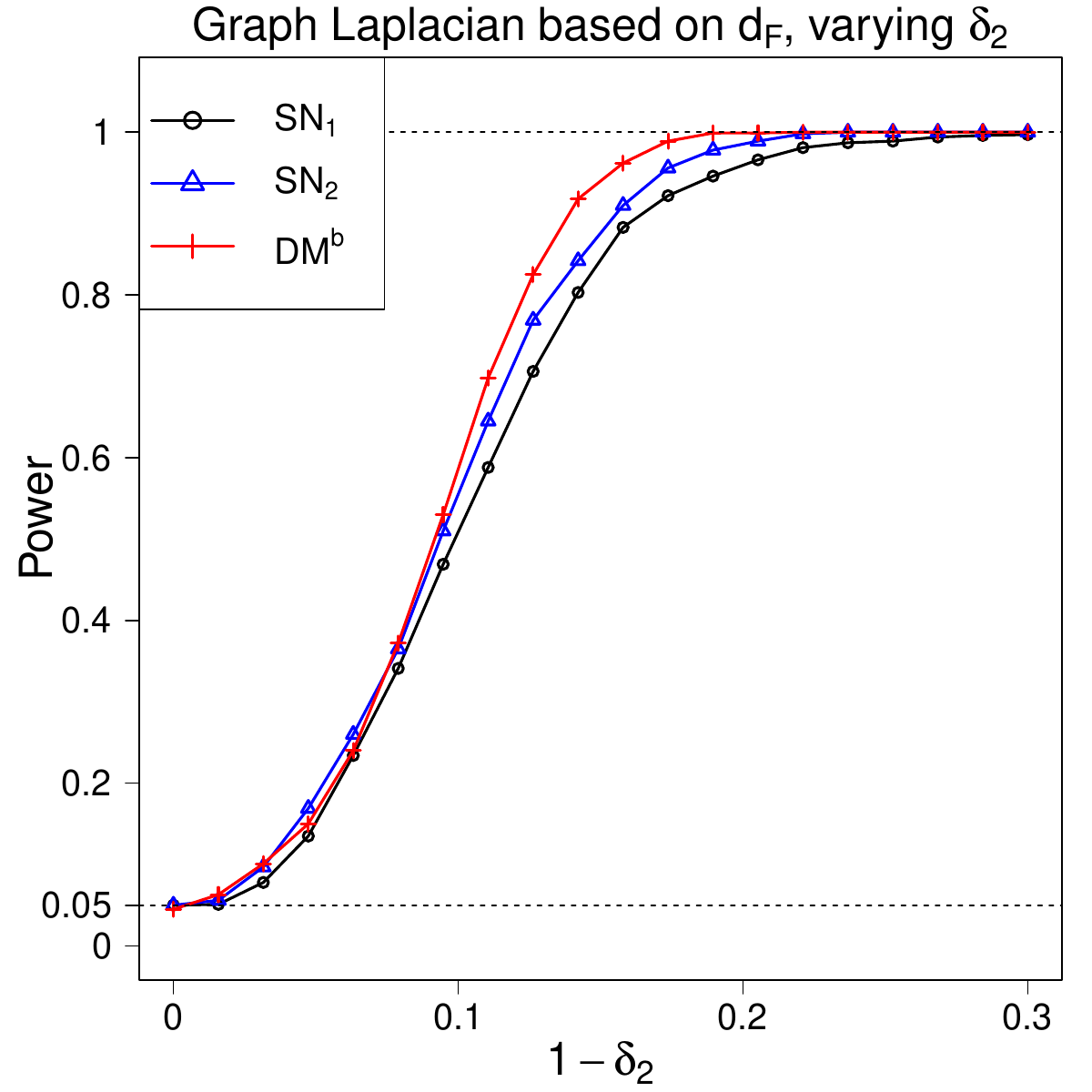}
	\end{subfigure}
	\begin{subfigure}{0.32\textwidth}
		\centering
		\includegraphics[width=1\textwidth]{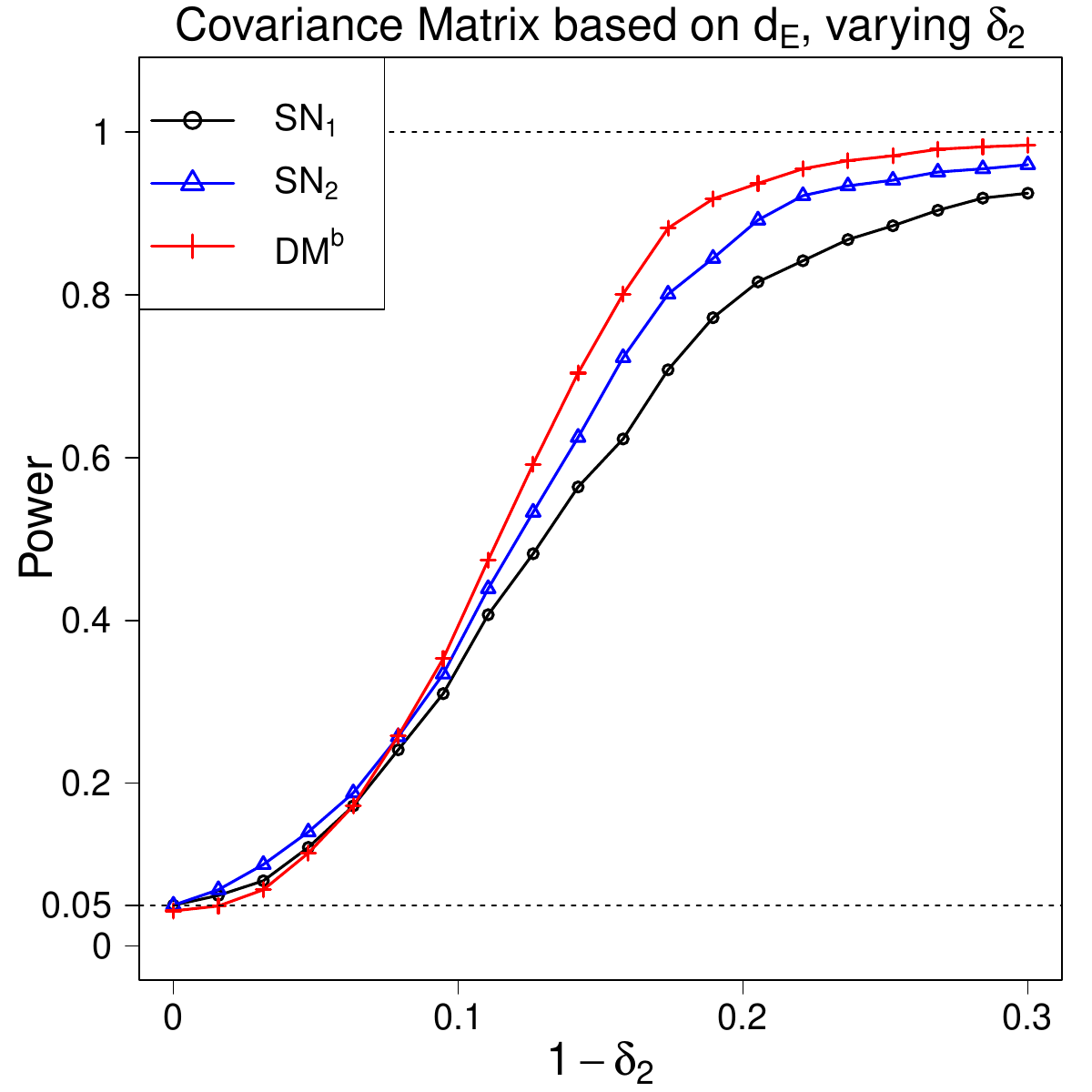}
	\end{subfigure}
	\caption{Size-Adjusted Power ($\times 100\%$) at $\alpha=5\%$, Change-Point Test for all three DGPs,  $n=400$, $\tau=0.5$, and $\rho=0.4$}
	\label{Fig:power_cpt}
\end{figure}

We further provide numerical evidence for the estimation accuracy by considering the alternative hypothesis of $\delta_1=1-\delta_2=0.3$ with true change-point location at $\tau=0.5$ for DGP (i)-(iii) in the main context. When varying sample size $n\in\{400,800,1600\}$, we find that for all DGPs, the  histograms 
of $\hat{\tau}$ (based on SN$_2$) plotted in  Figure \ref{Fig:cpt} get more concentrated around the truth $\tau=0.5$, when sample size increases, which is consistent with our theoretical consistency of $\hat{\tau}$.

\begin{figure}[H]
	\centering 
	\begin{subfigure}{0.32\textwidth}
		\centering
		\includegraphics[width=1\textwidth]{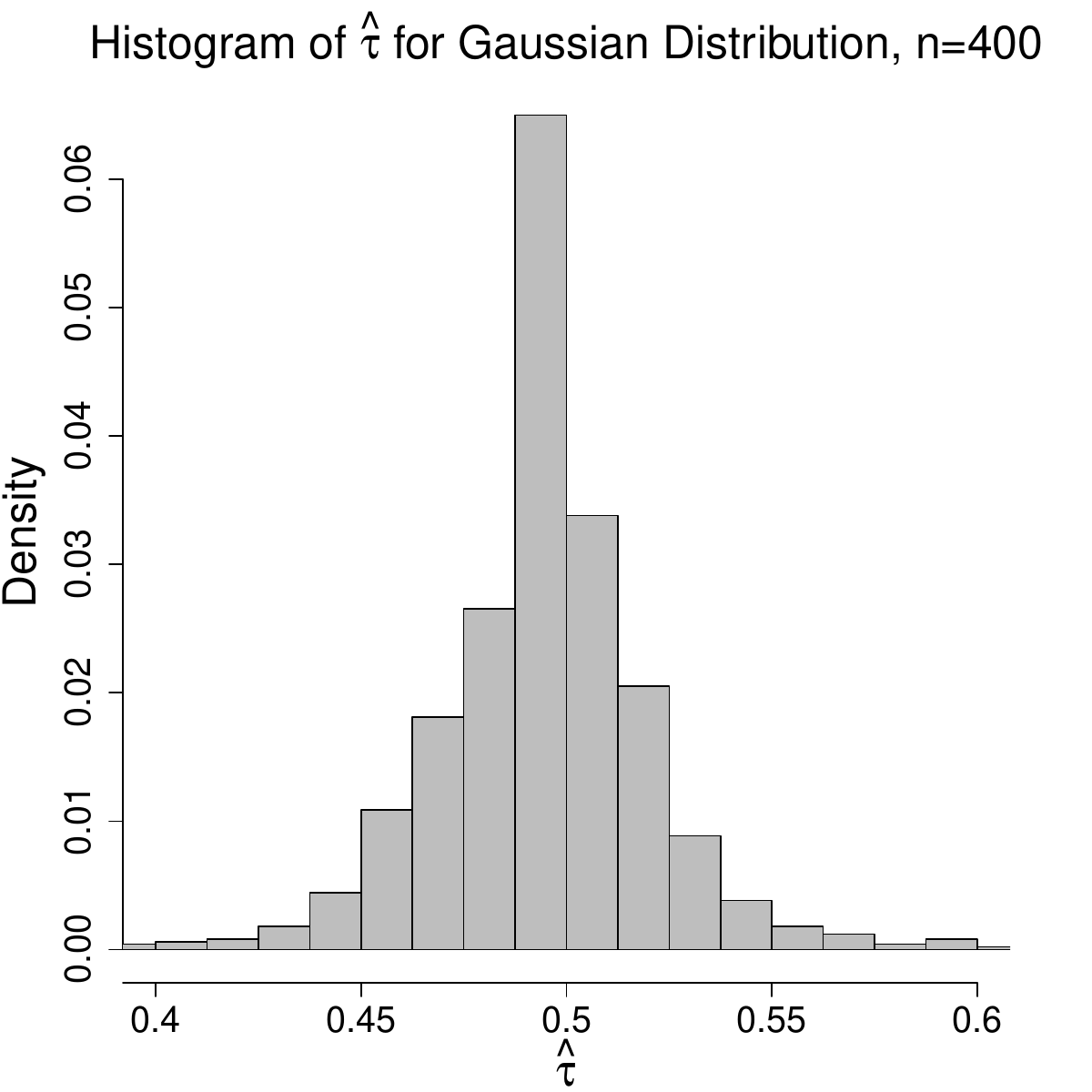}
	\end{subfigure}
	\begin{subfigure}{0.32\textwidth}
		\centering
		\includegraphics[width=1\textwidth]{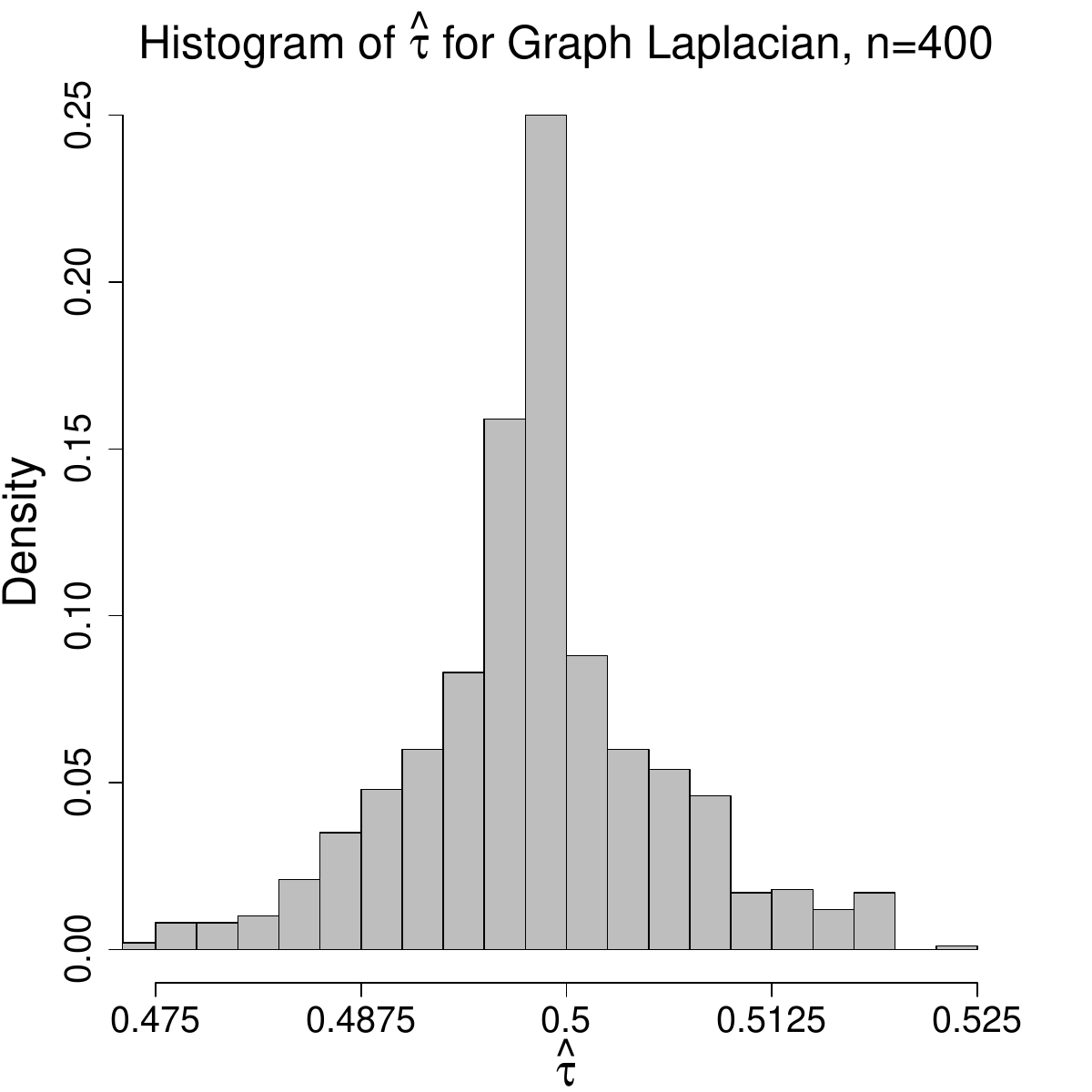}
	\end{subfigure}
	\begin{subfigure}{0.32\textwidth}
		\centering
		\includegraphics[width=1\textwidth]{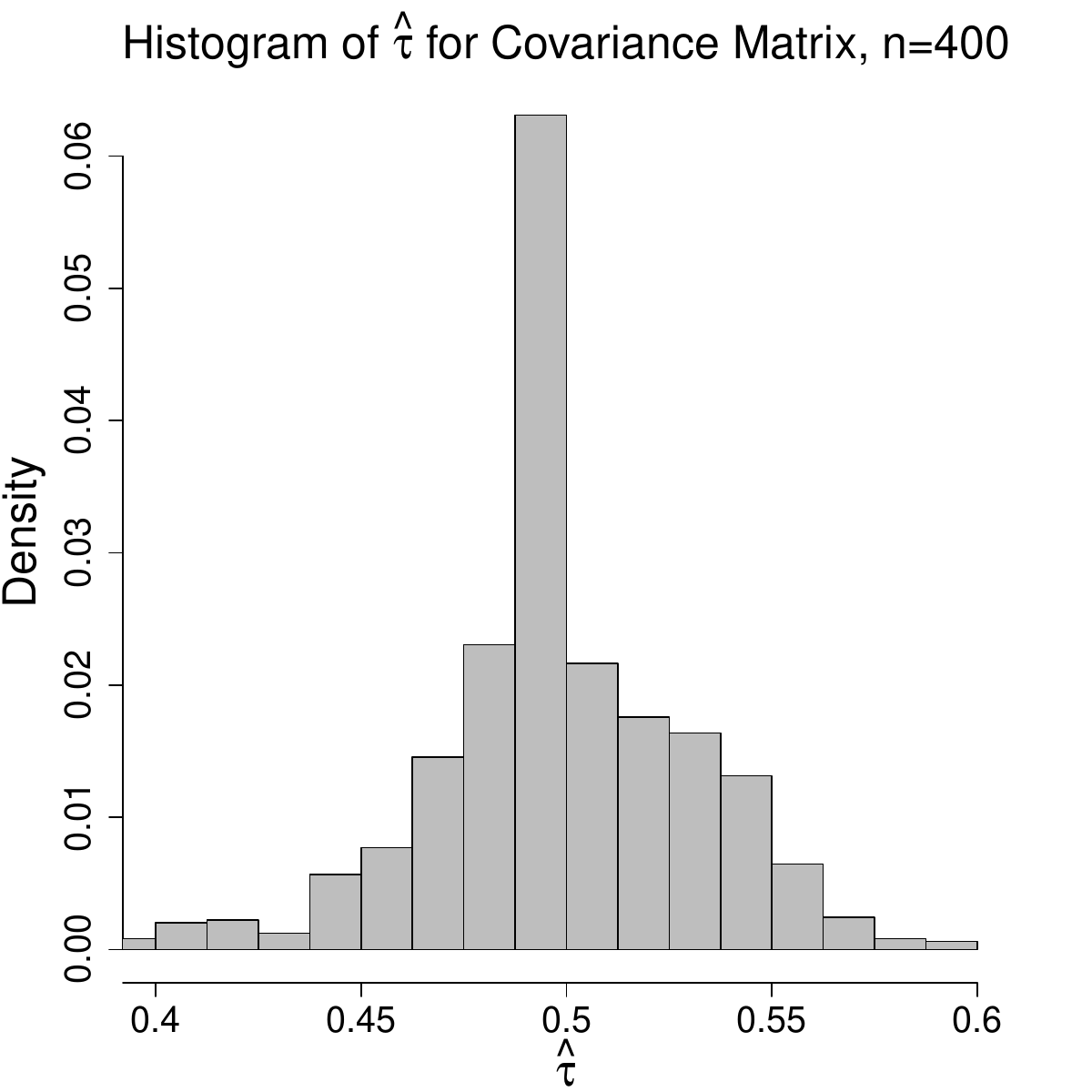}
	\end{subfigure}
	\begin{subfigure}{0.32\textwidth}
		\centering
		\includegraphics[width=1\textwidth]{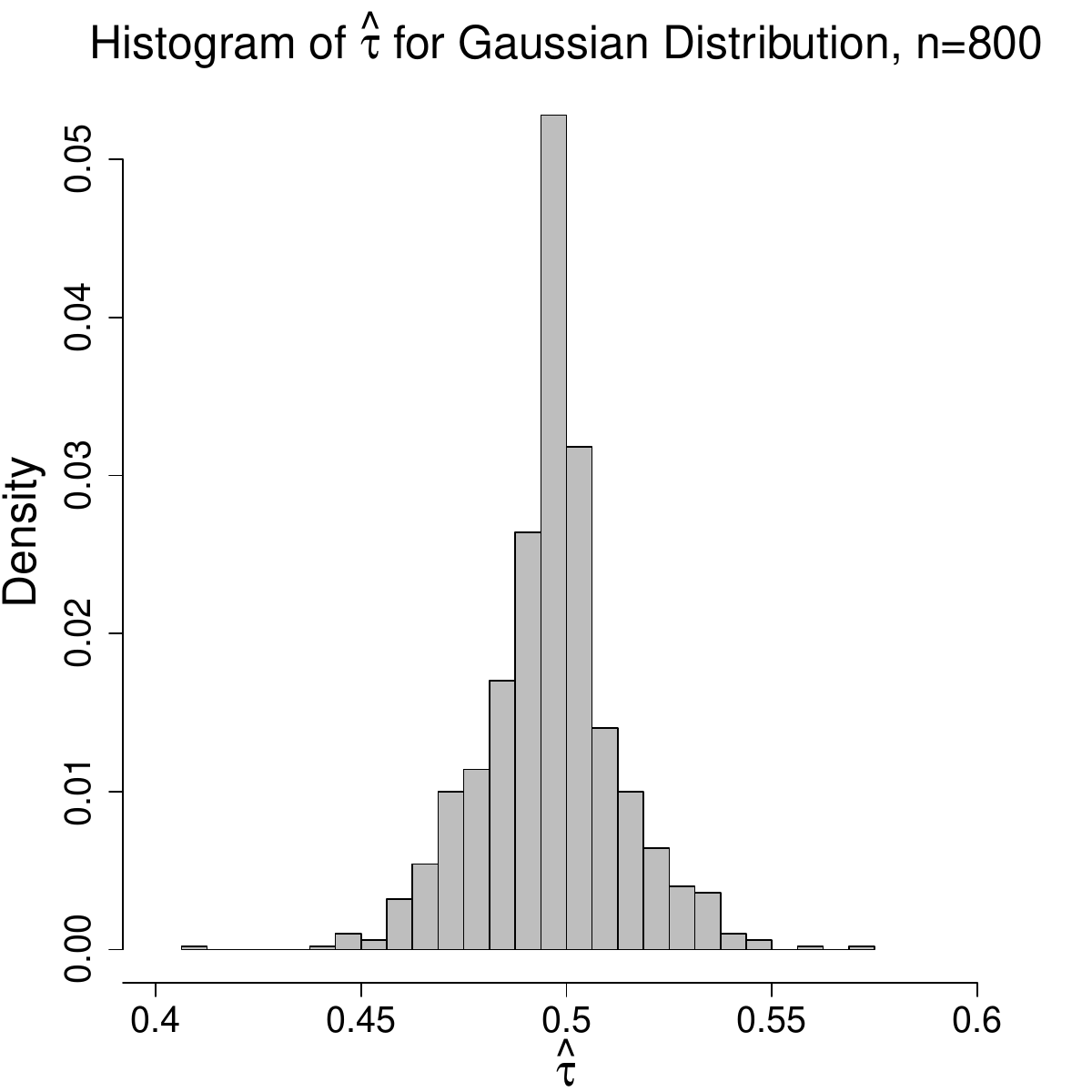}
	\end{subfigure}
	\begin{subfigure}{0.32\textwidth}
		\centering
		\includegraphics[width=1\textwidth]{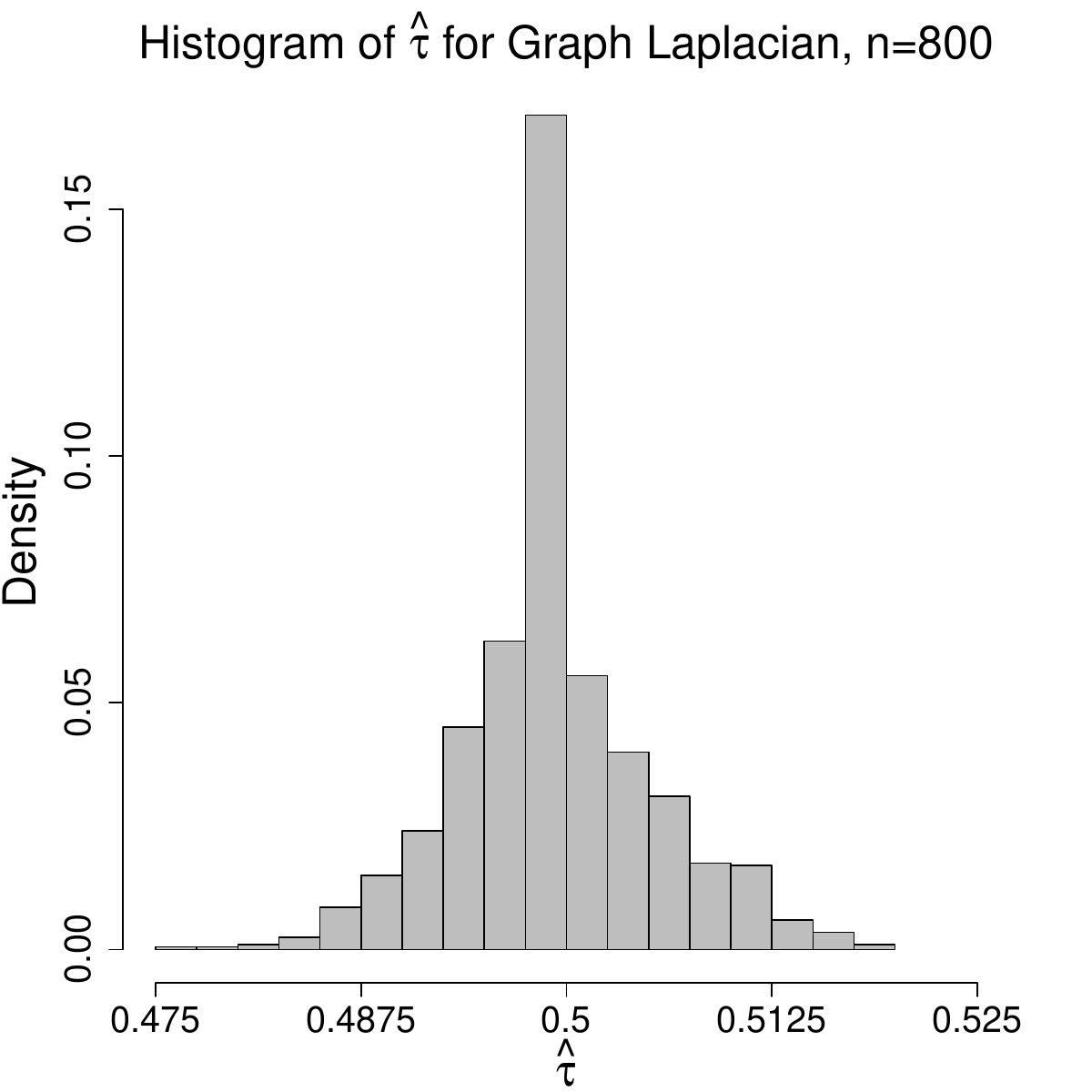}
	\end{subfigure}
	\begin{subfigure}{0.32\textwidth}
		\centering
		\includegraphics[width=1\textwidth]{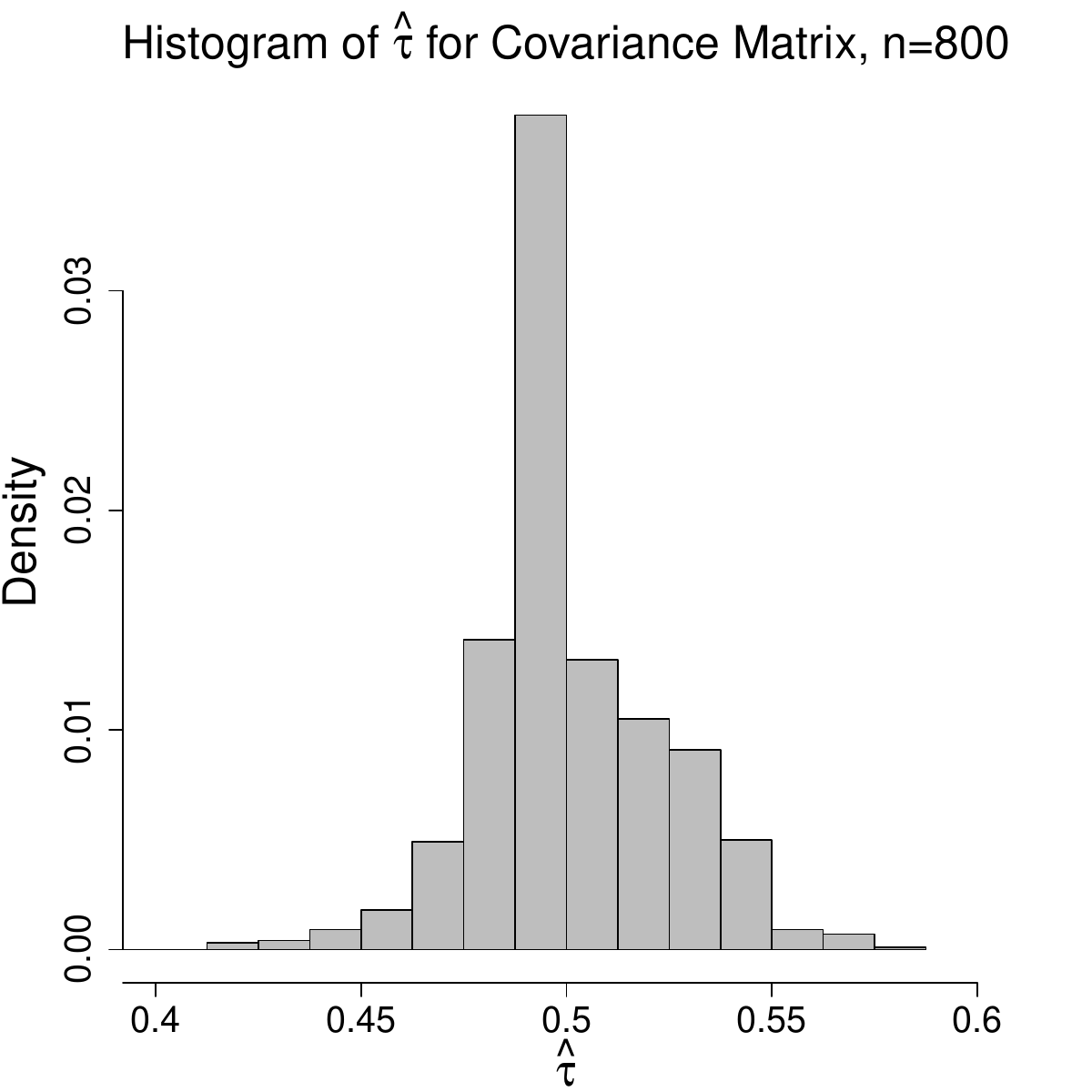}
	\end{subfigure}
	\begin{subfigure}{0.32\textwidth}
		\centering
		\includegraphics[width=1\textwidth]{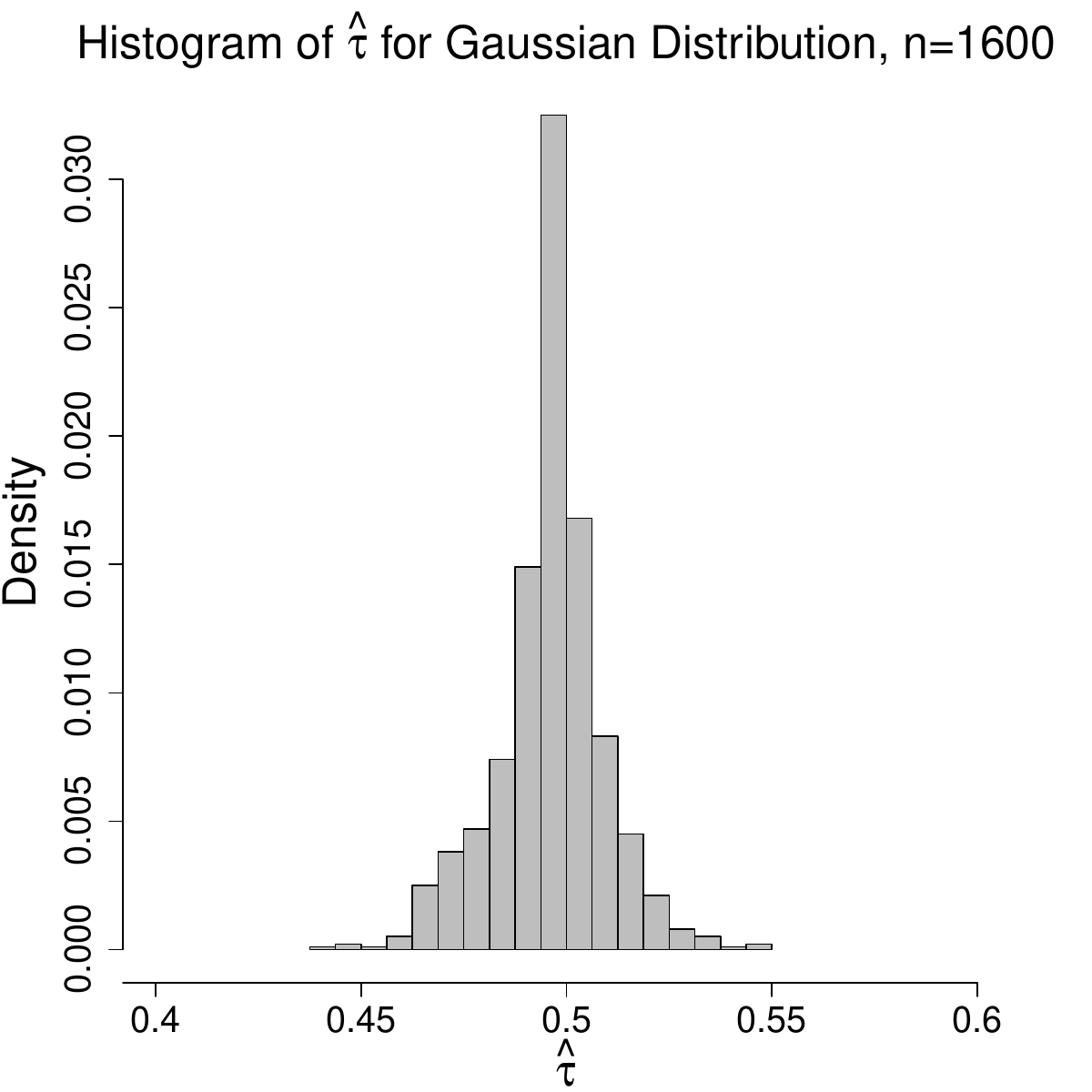}
	\end{subfigure}
	\begin{subfigure}{0.32\textwidth}
		\centering
		\includegraphics[width=1\textwidth]{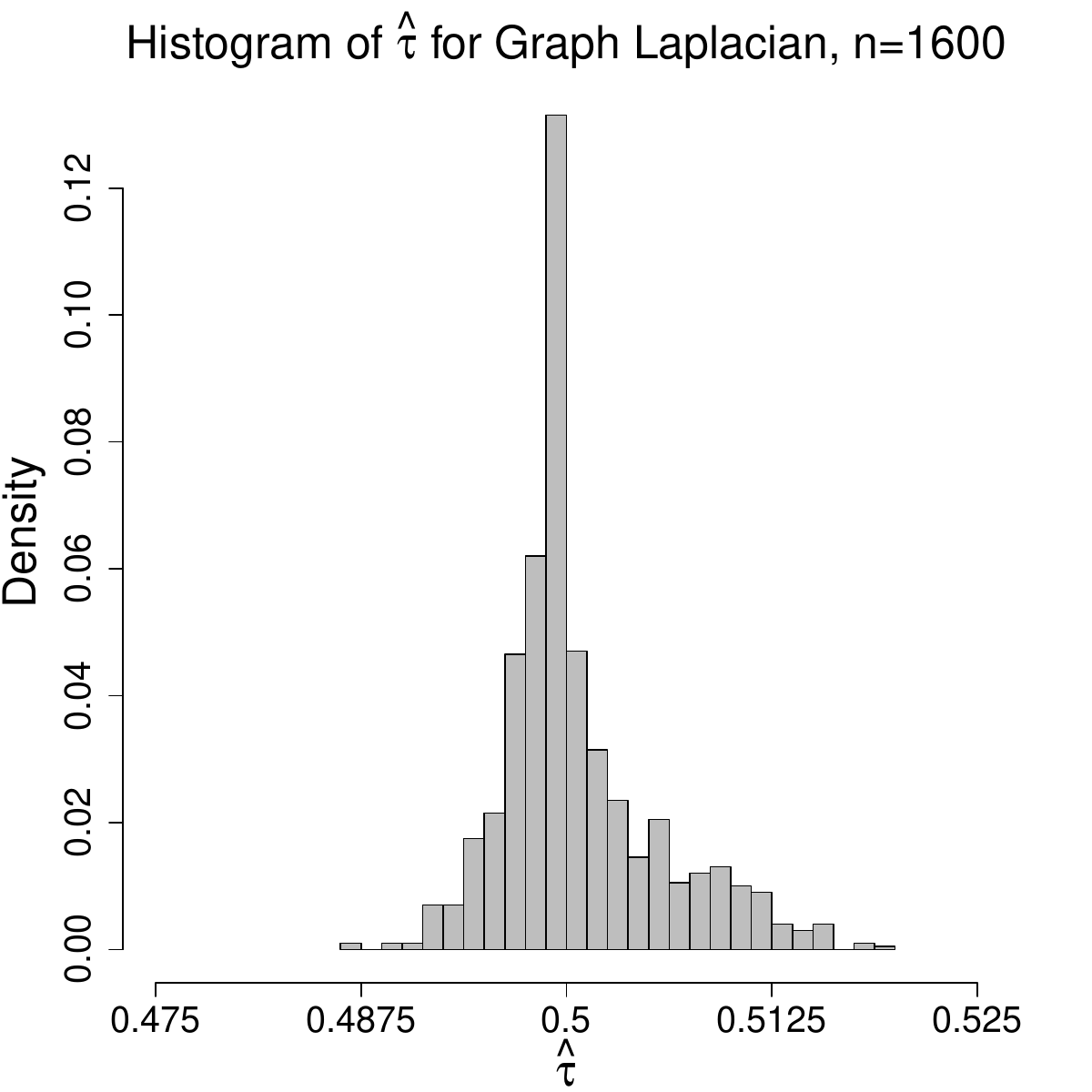}
	\end{subfigure}
	\begin{subfigure}{0.32\textwidth}
		\centering
		\includegraphics[width=1\textwidth]{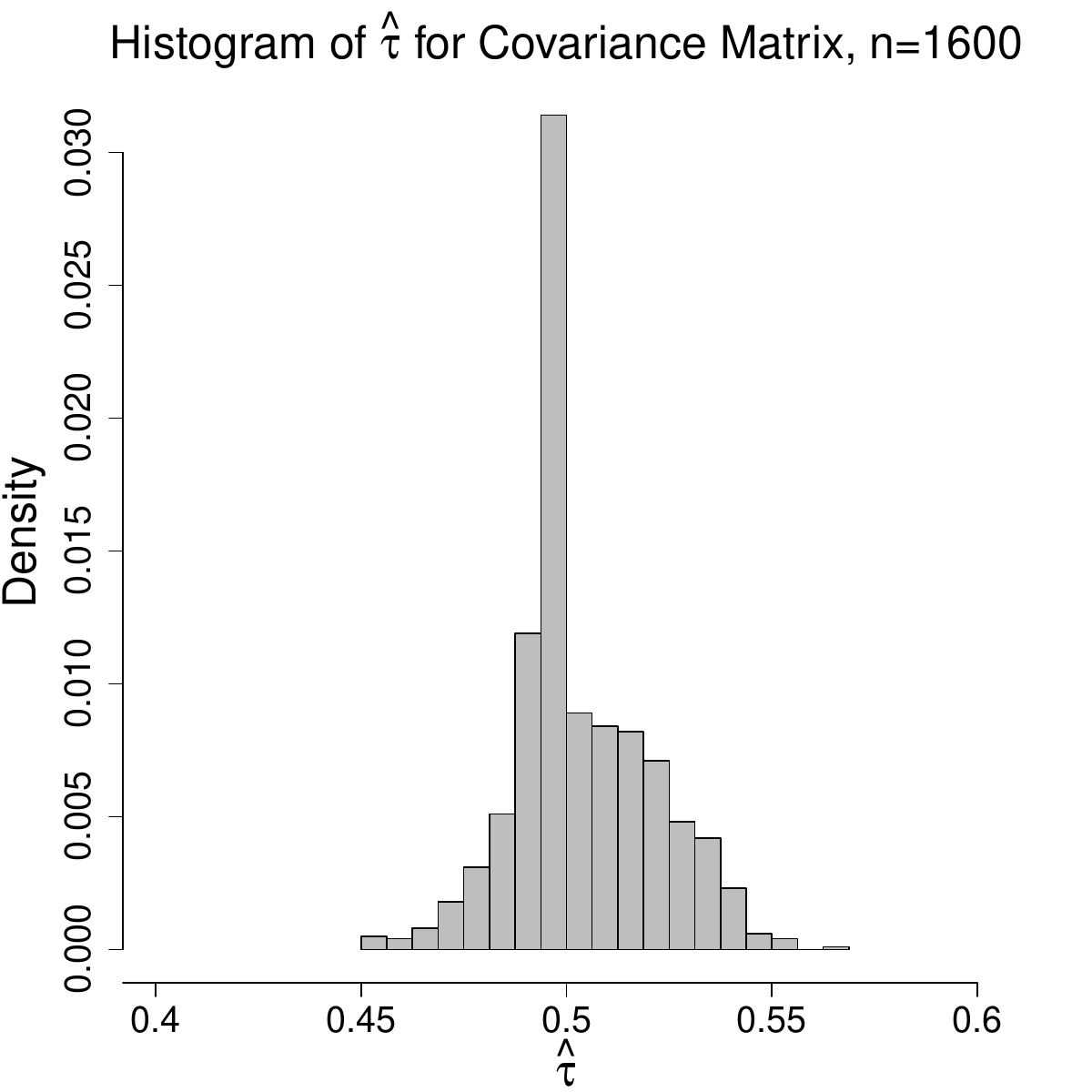}
	\end{subfigure}
	\caption{Histogram of Estimated Change-Point Locations for all three DGPs with $\delta_1=0.3, \delta_2=0.7$, $\tau=0.5$,  $\rho=0.4$. Upper: $n=400$. Middle: $n=800$. Lower: $n=1600$.}
	\label{Fig:cpt}
\end{figure}

\subsection{Multiple Change Point Detection}\label{sec:simu_WBS}
For simulations of multiple change point estimation, we consider non-Euclidean time series of length $n=500$ generated from the following two models. These models are the same as before, but reformulated for better presentation purpose. 
\begin{itemize}
	\item Gaussian univariate probability distribution: $Y_t=\mathcal{N}(\arctan (U_{t})+\delta_{t,1}, $ $ \delta_{t, 2}^2[\arctan(U_{t}^2)+1]^2)$.
	\item covariance matrix: $Y_t=(2I_3+Z_{t})(2I_3+Z_{t})^{\top}$ with $Z_t= \delta_{t,1}+\delta_{t,2}\arctan(U_{t}).$
\end{itemize}
Here, $U_{t}$ are generated according to the AR(1) process $U_t=\rho U_{t-1}+\epsilon_t$, $\epsilon_t\overset{i.i.d.}{\sim}\mathcal{N}(0,1).$ There are 3 change points at $t=110, 250$ and $370$. The changes point locations are reflected in the definitions of $\{ \delta_{t,1} \}$ and $\{\delta_{t,2}\}$, where
\begin{align*}
	& \delta_{t,1} = a_1 \mathbb{I}_{\{n \leq 110 \} }  + a_2 \mathbb{I}_{\{ 110 < n \leq 250 \} } + a_3 \mathbb{I}_{\{ 250 < n \leq 370 \} } + a_4 \mathbb{I}_{\{ 370 < n \leq 500 \} }, \\
	& \delta_{t,2} = b_1 \mathbb{I}_{\{ n \leq 110 \}}  + b_2 \mathbb{I}_{\{ 110 < n \leq 250 \}} + b_3 \mathbb{I}_{ \{ 250 < n \leq 370 \} } + b_4 \mathbb{I}_{ \{ 370 < n \leq 500 \} }.
\end{align*}
For each model, we consider 3 cases that are differentiated by the magnitudes of $a_i, b_i, i=1,2,3,4$. For the data generating model of Gaussian distributions, we  set
\begin{itemize}
	\item \mbox{[Case 1]} $(a_1, a_2, a_3, a_4) = (0, 0.7, 0, 0.8), (b_1, b_2, b_3, b_4) = (1, 1.5, 0.7, 1.4)$;
	\item \mbox{[Case 2]} $(a_1, a_2, a_3, a_4) = (0, 0.2, 0, 0.3), (b_1, b_2, b_3, b_4) = (0.5, 1.5, 0.4, 1.4)$;
	\item \mbox{[Case 3]} $(a_1, a_2, a_3, a_4) = (0, 0.5, 1.5, 3.3), (b_1, b_2, b_3, b_4) = (0.2, 1.5, 3.8, 6.5)$.
\end{itemize}
As for the data generating model of covariance matrices, we set \begin{itemize}
	\item \mbox{[Case 1]}  $(a_1, a_2, a_3, a_4) = (0, 1.2, 0, 1.3), (b_1, b_2, b_3, b_4) = (0.8, 1.5, 0.7, 1.6)$;
	\item \mbox{[Case 2]} $(a_1, a_2, a_3, a_4) = (0, 1, 0, 1)$, $(b_1, b_2, b_3, b_4) = (0.5, 2, 0.4, 1.9)$;
	\item \mbox{[Case 3]} $(a_1, a_2, a_3, a_4) = (0, 2, 3.9, 5.7)$, $(b_1, b_2, b_3, b_4) = (0.2, 0.7, 1.3, 2)$.
\end{itemize}  
Cases 1 and 2 correspond to non-monotone changes and Case 3 considers the monotone change. 
Here, our method described in Section \ref{wbs} is denoted as WBS-$SN_2$ (that is, a combination of WBS and our $SN_2$ test statistic). 
The method DM in conjunction with binary segmentation, referred as BS-DM, is proposed in \cite{dubey2019frechet} and included in this simulation for comparison purpose. In addition, our statistic $SN_2$ 
in combination  with binary segmentation, denoted as BS-$SN_2$, is implemented and included as well. The critical values for BS-DM and BS-$SN_2$ are obtained from their asymptotic distributions respectively.

The simulation results are shown in Table \ref{tab:wbs1}, where we present the ARI (adjusted rand index) and number of detected change points for two dependence levels $\rho=0.3, 0.6$. Note that ARI $\in [0,1]$
measures  the accuracy of change point estimation and  larger ARI corresponds to more accurate estimation.
We summarize the main findings as follows. (a) WBS-$SN_2$ is the best method in general as it can accommodate both monotonic and non-monotoic changes, and appears quite robust to temporal dependence.  For Cases 1 and 2, we see that BS-$SN_2$ does not work for non-monotone changes,  due to the use of binary segmentation procedure. (b) BS-DM tends to have more false discoveries comparing to the other methods. This is expected, as method DM is primarily proposed for i.i.d data sequence and exhibit serious oversize when there is temporal dependence in Section~\ref{sec:simu_CP}. (c) When we increase $\rho=0.3$ to $\rho=0.6$, the performance of WBS-$SN_2$ appears quite stable for both  distributional time series and  covariance matrix time series.

\begin{table}[H] \centering 
	\caption{Simulation results for multiple change point for sequential Gaussian distributions and covariance matrices based on 200 Monte Carlo repetitions.} 
	\label{tab:wbs1} \resizebox{\textwidth}{!}{
		\begin{tabular}{@{\extracolsep{5pt}} ccccccccccc} 
			\\[-1.8ex]\hline 
			\hline \\[-1.8ex] 
			\multirow{2}{*}{Model} &\multirow{2}{*}{Case} & \multirow{2}{*}{$\rho$} & \multirow{2}{*}{Method} & \multicolumn{6}{c}{\# of change points detected} & \multirow{2}{*}{ARI}  \\  \cline{5-9}
			&&  &  & $0$ & $1$ & $2$ & $3$ & $4$ & $\geq 5$ &   \\ 
			\hline \\[-1.8ex] 
			\multirow{18}{*}{$\begin{array}{c}
					\text{Gaussian}   \\
					\text{Distribution}
				\end{array}$}&\multirow{6}{*}{1} & \multirow{3}{*}{0.3} & $\text{WBS-}SN_2$ & $0$ & $0$ & $0$ & $178$ & $22$ & $0$ & $0.971$ \\ 
			&&  & $\text{BS-}SN_2$ & $200$ & $0$ & $0$ & $0$ & $0$ & $0$ & $0$ \\ 
			& &  & BS-DM & $0$ & $0$ & $0$ & $23$ & $15$ & $162$ & $0.836$ \\  \cline{3-11}
			& & \multirow{3}{*}{0.6} & $\text{WBS-}SN_2$ & $0$ & $0$ & $7$ & $116$ & $54$ & $23$ & $0.907$ \\ 
			& & & $\text{BS-}SN_2$ & $200$ & $0$ & $0$ & $0$ & $0$ & $0$ & $0$ \\ 
			& &  & BS-DM & $0$ & $0$ & $0$ & $2$ & $0$ & $198$ & $0.516$ \\  \cline{2-11}
			&\multirow{6}{*}{2} & \multirow{3}{*}{0.3} & $\text{WBS-}SN_2$ & $0$ & $0$ & $0$ & $169$ & $28$ & $3$ & $0.980$ \\ 
			& &  & $\text{BS-}SN_2$ & $200$ & $0$ & $0$ & $0$ & $0$ & $0$ & $0$ \\ 
			& &  & BS-DM & $0$ & $0$ & $0$ & $20$ & $14$ & $166$ & $0.834$ \\ \cline{3-11}
			& & \multirow{3}{*}{0.6} & $\text{WBS-}SN_2$ & $0$ & $0$ & $0$ & $101$ & $60$ & $39$ & $0.943$ \\  
			& &  & $\text{BS-}SN_2$ & $200$ & $0$ & $0$ & $0$ & $0$ & $0$ & $0$ \\ 
			& &  & BS-DM & $0$ & $0$ & $0$ & $1$ & $1$ & $198$ & $0.526$ \\  \cline{2-11}
			&\multirow{6}{*}{3} & \multirow{3}{*}{0.3} & $\text{WBS-}SN_2$ & $0$ & $0$ & $0$ & $172$ & $27$ & $1$ & $0.981$ \\ 
			& &  & $\text{BS-}SN_2$ & $0$ & $0$ & $59$ & $82$ & $18$ & $41$ & $0.808$ \\ 
			& &  & BS-DM & $0$ & $0$ & $0$ & $39$ & $25$ & $136$ & $0.893$ \\  \cline{3-11}
			& & \multirow{3}{*}{0.6} & $\text{WBS-}SN_2$ & $0$ & $0$ & $0$ & $112$ & $65$ & $23$ & $0.955$ \\ 
			& &  & $\text{BS-}SN_2$ & $0$ & $0$ & $45$ & $64$ & $41$ & $50$ & $0.828$ \\ 
			& & & BS-DM & $0$ & $0$ & $0$ & $2$ & $4$ & $194$ & $0.666$ \\ 
			\hline \\[-1.8ex] 
			\multirow{18}{*}{ $ \begin{array}{c}
					\text{Covariance}   \\
					\text{Matrix}
				\end{array} $ } &\multirow{6}{*}{1} & \multirow{3}{*}{0.3}& $\text{WBS-}SN_2$ & $0$ & $0$ & $0$ & $200$ & $0$ & $0$ & $0.991$ \\ 
			& &  & $\text{BS-}SN_2$ & $200$ & $0$ & $0$ & $0$ & $0$ & $0$ & $0$ \\ 
			& &  & BS-DM & $0$ & $0$ & $0$ & $10$ & $9$ & $181$ & $0.807$ \\ \cline{3-11}
			& & \multirow{3}{*}{0.6} & $\text{WBS-}SN_2$ & $0$ & $0$ & $8$ & $192$ & $0$ & $0$ & $0.974$ \\ 
			& &  & $\text{BS-}SN_2$ & $199$ & $1$ & $0$ & $0$ & $0$ & $0$ & $0.002$ \\ 
			& &  & BS-DM & $0$ & $0$ & $0$ & $0$ & $0$ & $200$ & $0.392$ \\  \cline{2-11}
			&\multirow{6}{*}{2} & \multirow{3}{*}{0.3} & $\text{WBS-}SN_2$ & $0$ & $0$ & $21$ & $178$ & $1$ & $0$ & $0.954$ \\ 
			& &  & $\text{BS-}SN_2$ & $200$ & $0$ & $0$ & $0$ & $0$ & $0$ & $0$ \\ 
			& &  & BS-DM & $6$ & $3$ & $1$ & $4$ & $12$ & $174$ & $0.722$ \\  \cline{3-11}
			& & \multirow{3}{*}{0.6} & $\text{WBS-}SN_2$ & $0$ & $20$ & $82$ & $97$ & $1$ & $0$ & $0.800$ \\ 
			& &  & $\text{BS-}SN_2$ & $200$ & $0$ & $0$ & $0$ & $0$ & $0$ & $0$ \\ 
			& &  & BS-DM & $5$ & $0$ & $0$ & $0$ & $0$ & $195$ & $0.510$ \\  \cline{2-11}
			&\multirow{6}{*}{3} & \multirow{3}{*}{0.3}  & $\text{WBS-}SN_2$ & $0$ & $0$ & $131$ & $69$ & $0$ & $0$ & $0.806$ \\ 
			& &  & $\text{BS-}SN_2$ & $24$ & $0$ & $124$ & $3$ & $46$ & $3$ & $0.468$ \\ 
			& &  & BS-DM & $0$ & $0$ & $1$ & $5$ & $10$ & $184$ & $0.781$ \\  \cline{3-11}
			& & \multirow{3}{*}{0.6}  & $\text{WBS-}SN_2$ & $0$ & $2$ & $181$ & $16$ & $1$ & $0$ & $0.728$ \\ 
			& &  & $\text{BS-}SN_2$ & $93$ & $1$ & $71$ & $18$ & $17$ & $0$ & $0.310$ \\ 
			& &  & BS-DM & $0$ & $0$ & $0$ & $0$ & $0$ & $200$ & $0.437$ \\ 
			\hline \\[-1.8ex] 
	\end{tabular} }
\end{table}

\section{Applications}\label{sec:app}

In this section, we present two real data illustrations, one for two sample testing and the other  for change-point detection. Both datasets are in the form of non-Euclidean time series and neither seems to be analyzed before by using  techniques that take into account unknown temporal dependence. 

\subsection{Two sample tests}

\textbf{Mortality data.} Here we are interested in comparing the longevity of people living in different countries of Europe. From the Human Mortality Database (\url{https://www.mortality.org/Home/Index}), we can obtain a time series that consists of yearly age-at-death distributions for each country.  We shall focus on distributions for female from year 1960 to 2015 and there are 26 countries included in the analysis after exclusion of countries with missing data. Pair-wise two sample tests between the included countries are performed using our statistic $D_2$ to understand the similarity of age-at-death distributions between different countries. The resulting $p$-value matrix is plotted in Figure \ref{fig:mort} (left). 

To better present the testing results and gain more insights, we define the dissimilarity between two given countries by subtracting each $p$-value from 1. Treating these dissimilarities as ``distances", we apply multidimensional scaling to ``project" each country onto two dimensional plane for visualization.  See Figure \ref{fig:mort} (right) for the plot of ``projected" countries. It appears that several west European countries, including UK, Belgium, Luxembourg, Ireland, and Austria, and Denmark, form a cluster; whereas several central and eastern European countries, including Poland, Latvia, Russian, Bulgaria, Lithuania and Czechia share similar distributions. We suspect the similarity in Mortality distribution is much related to the similarity in their economic development and healthcare system, less dependent on the geographical locations.



\begin{figure}[h]
	\centering
	\begin{subfigure}{0.49\textwidth}
		\centering
		\includegraphics[width=1\textwidth]{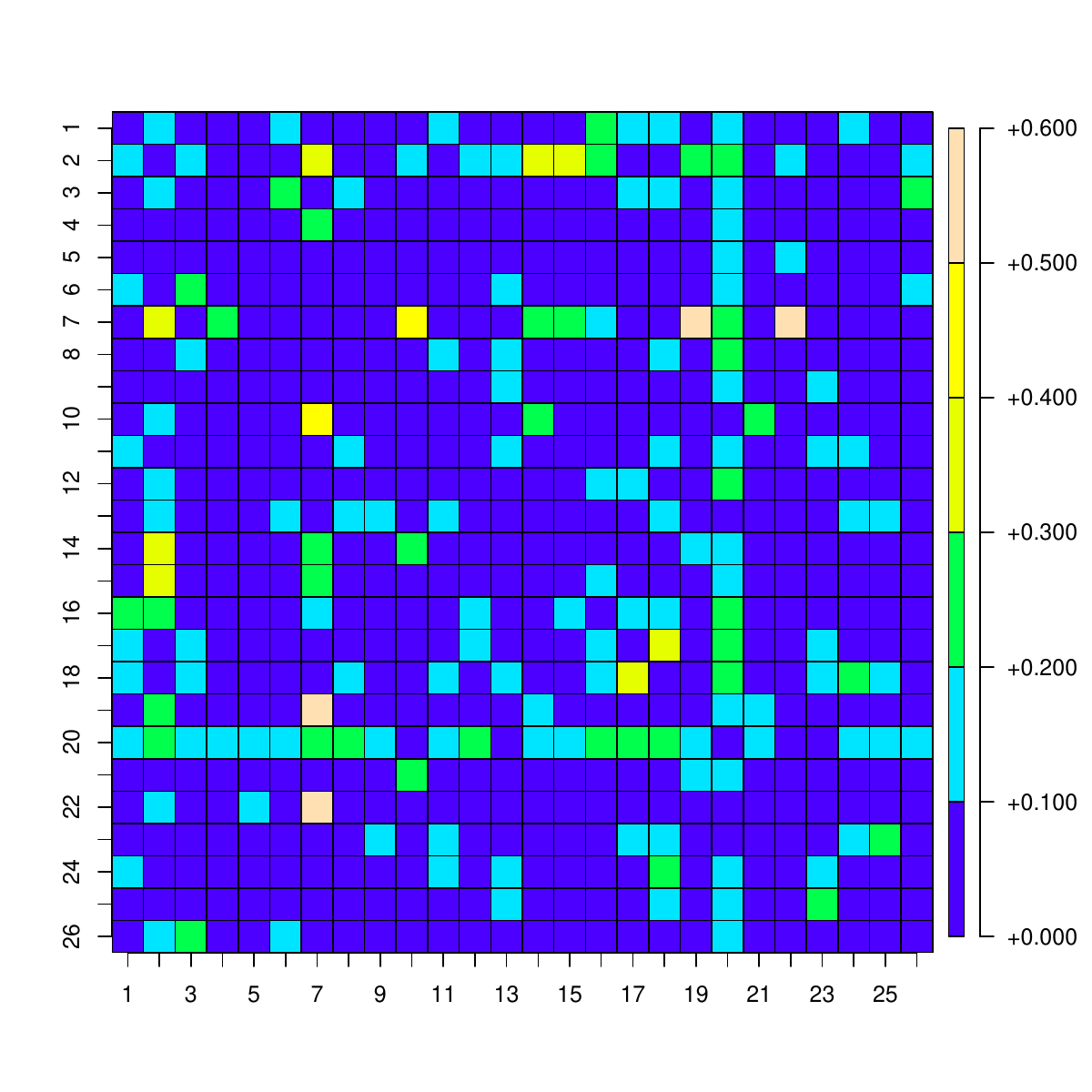}
	\end{subfigure}
	\begin{subfigure}{0.49\textwidth}
		\centering
		\includegraphics[width=1\textwidth]{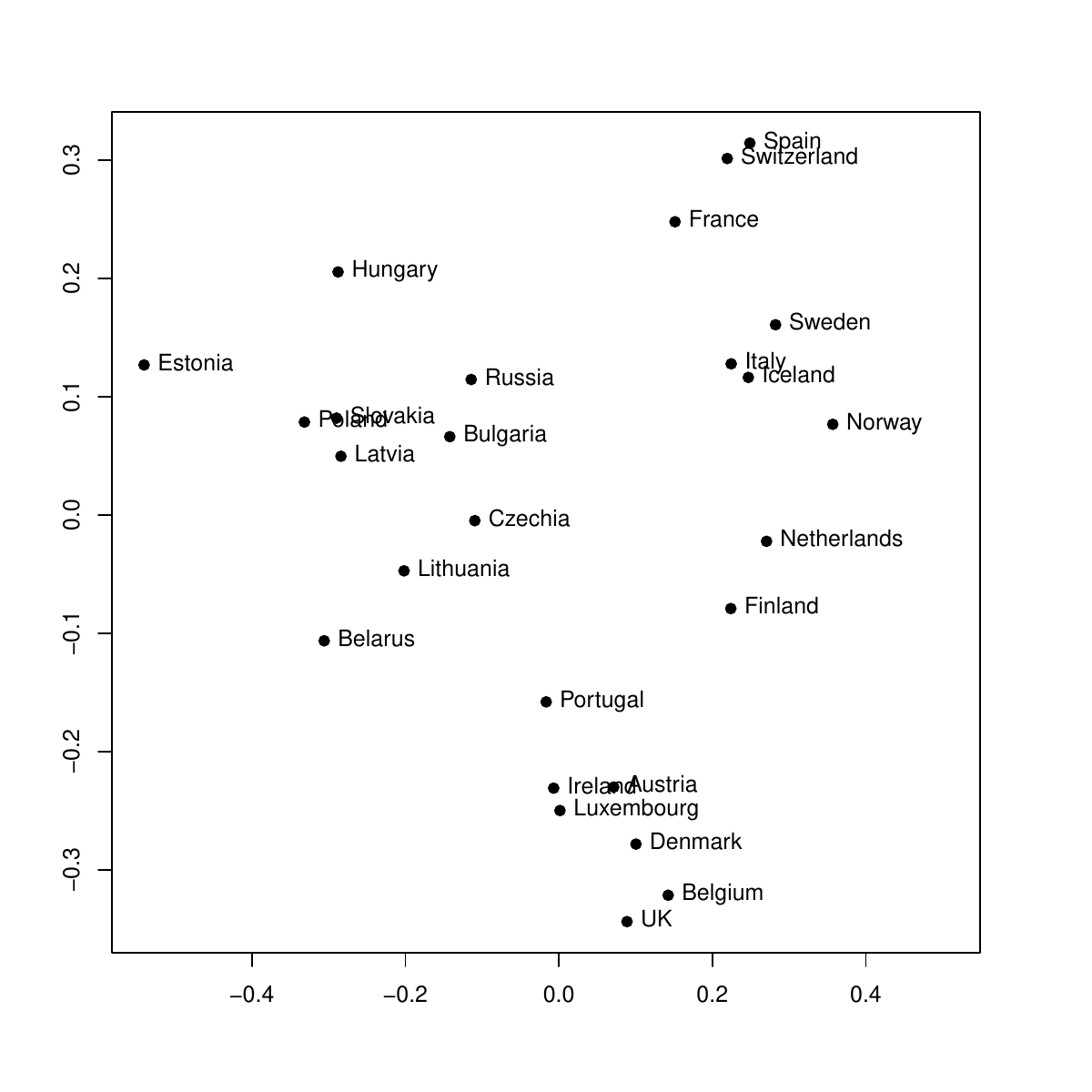}
	\end{subfigure}
	\caption{Application to Mortality data. The left figure is the plot of $p$-value matrix, where the 26 countries, Austria, Belarus, Belgium, Bulgaria, Czechia, Denmark, Estonia, Finland, France, Hungary, Iceland, Ireland, Italy, Latvia, Lithuania, Luxembourg, Netherlands, Norway, Poland, Portugal, Russia, Slovakia, Spain, Sweden, Switzerland, UK, are labeled by $1, 2, \dots, 26$ respectively. The right figure is the plot of points from multidimensional scaling with dissimilarity between any countries defined as subtracting the corresponding $p$-value from 1.}
	\label{fig:mort}
\end{figure}

\subsection{Change point detection}


\noindent \textbf{Cryptocurrency data.} Detecting change points in the Pearson correlation matrices for a set of interested cryptocurrencies can uncover structural breaks in the correlation of these cryptocurrencies and can play an important role in the investors' investment decisions. Here, we construct the daily Pearson correlation matrices from minute prices of Bitcoin, Doge coin, Cardano, Monero and Chainlink  for year 2021. The cryptocurrency data can be downloaded at \url{https://www.cryptodatadownload.com/analytics/correlation-heatmap/}. See Figure \ref{fig:crypto} for the plot of time series of pairwise correlations. Three methods, namely, our $SN_2$ test combined with WBS (WBS-$SN_2$), $SN_2$ test combined with binary segmentation (BS-$SN_2$), and   DM test of \cite{dubey2019frechet} in conjunction with  binary segmentation  (BS-DM),  are  applied to detect potential change points for this time series, 

Method WBS-$SN_2$ detects an abrupt change on day 2021-05-17 and method BS-$SN_2$ detects a change point on day 2021-04-29. By comparison, more than 10 change points are detected by BS-DM and we suspect that many of them are false discoveries (see Section \ref{sec:simu_WBS} for simulation evidence of BS-DM's tendency of over-detection). The change point in  mid-May 2021 is well expected and  corresponds to a major crush in crypto market that wiped out 1 trillion dollars. The major causes of this crush are  the withdrawal of Tesla's commitment to accept Bitcoin as payment and warnings regarding cryptocurrency sent by Chinese central bank to the financial institutes and business in China. Since this major crush, the market has been dominated by negative sentiments and fear for a recession. We refer the following CNN article for some discussions about this crush \url{https://www.cnn.com/2021/05/22/investing/crypto-crash-bitcoin-regulation/index.html}.




\begin{figure}[h]
	\centering
	\begin{tikzpicture}
		\matrix (m) [row sep = 0em, column sep = 0em]{ 
			\node (p11) {Bitcoin}; & \node (p12) {\includegraphics[scale=0.12]{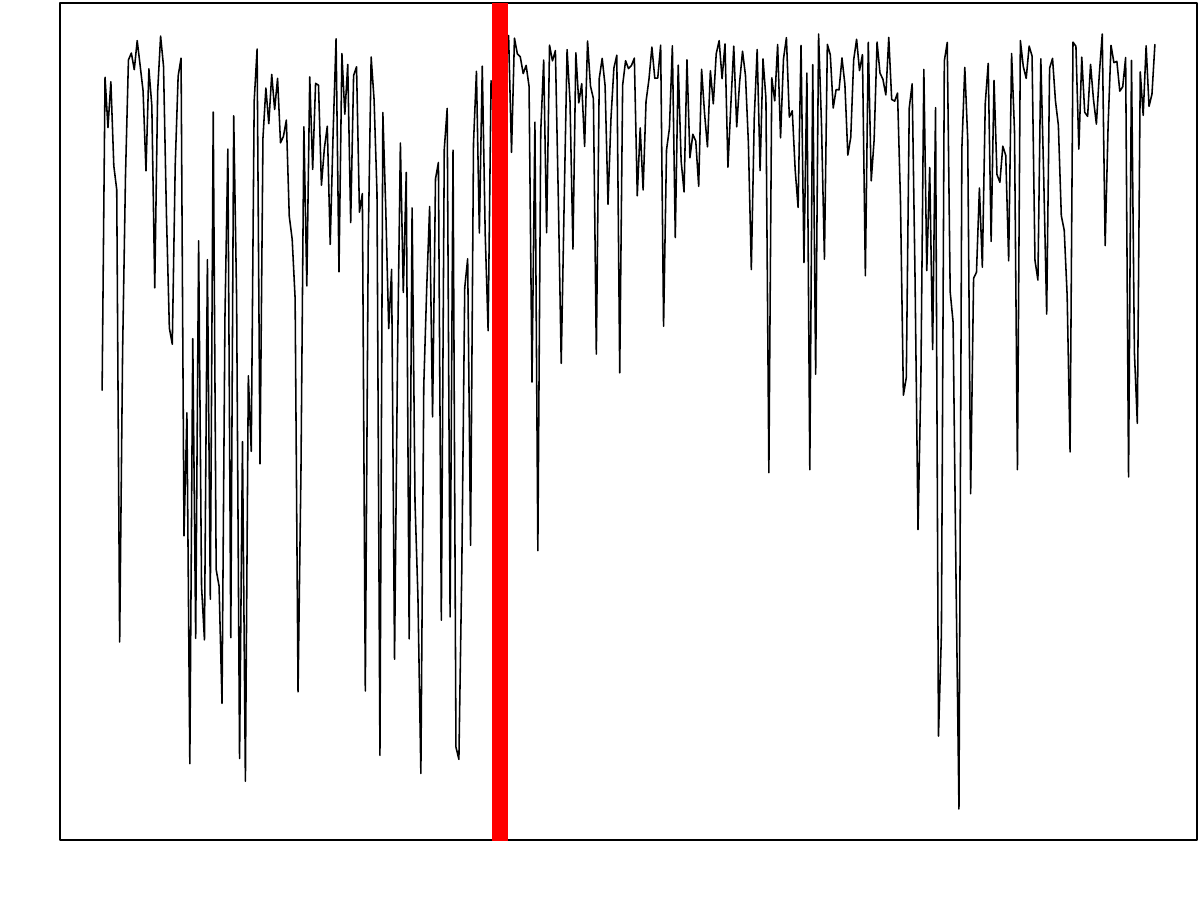}}; &
			\node (p13) {\includegraphics[scale=0.12]{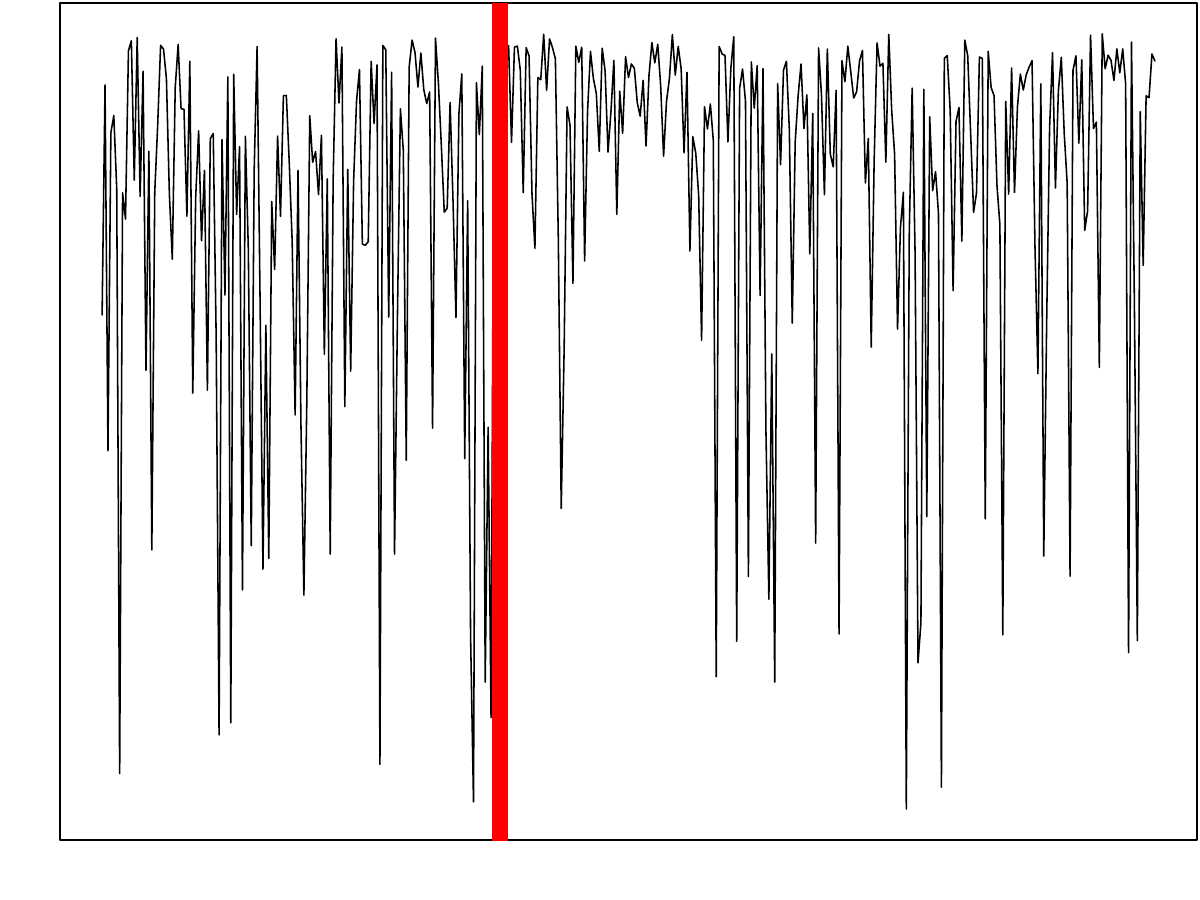}};
			& \node (p14) {\includegraphics[scale=0.12]{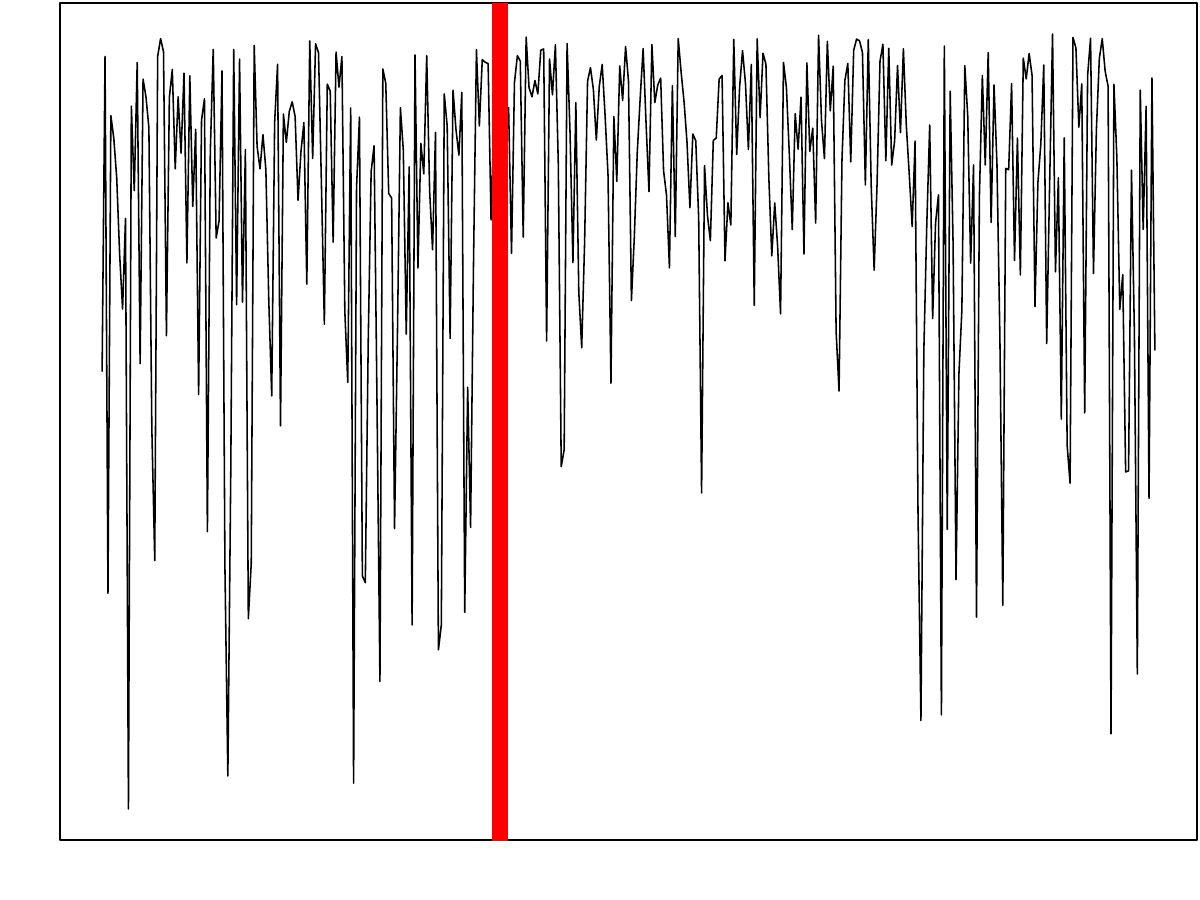}};  & \node (p15) {\includegraphics[scale=0.12]{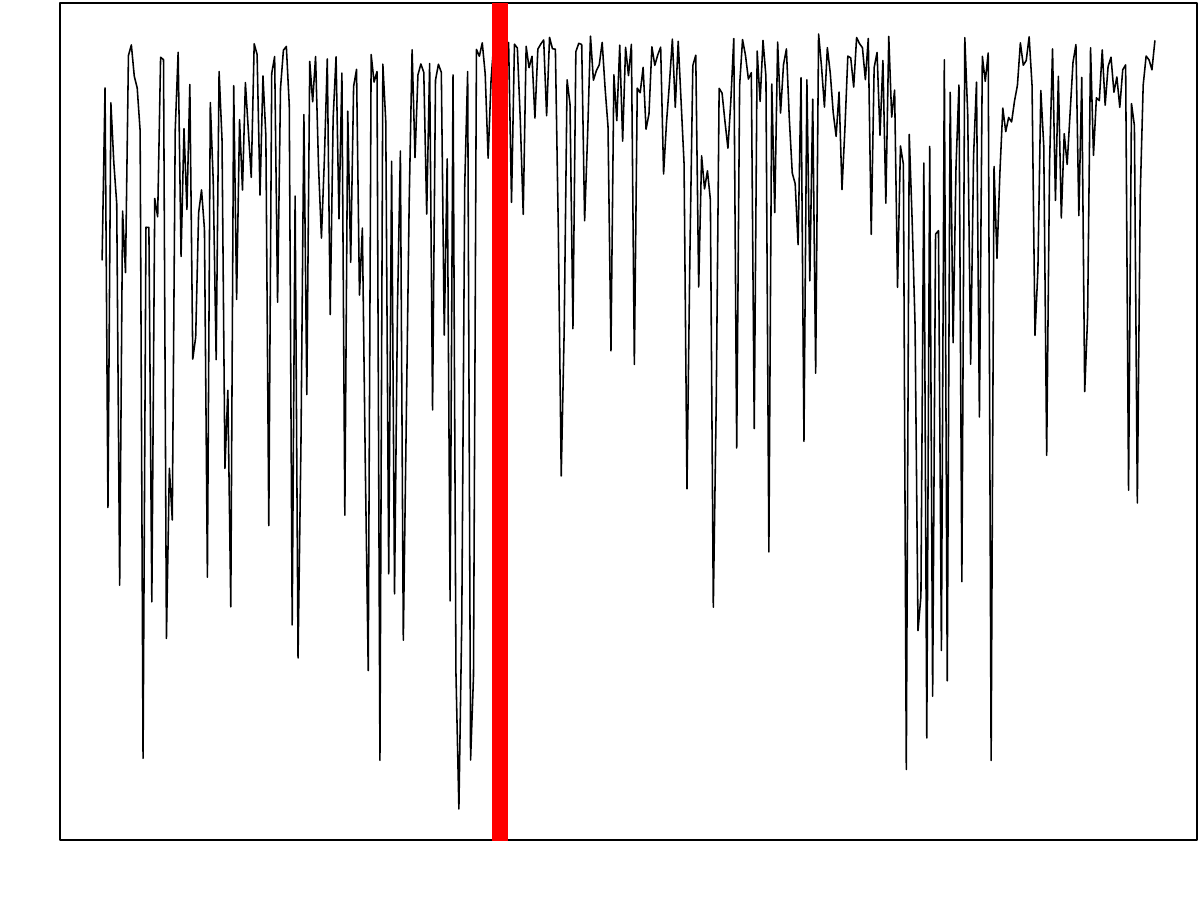}}; 
			\\ 
			& \node (p22) {Dogecoin}; &
			\node (p13) {\includegraphics[scale=0.12]{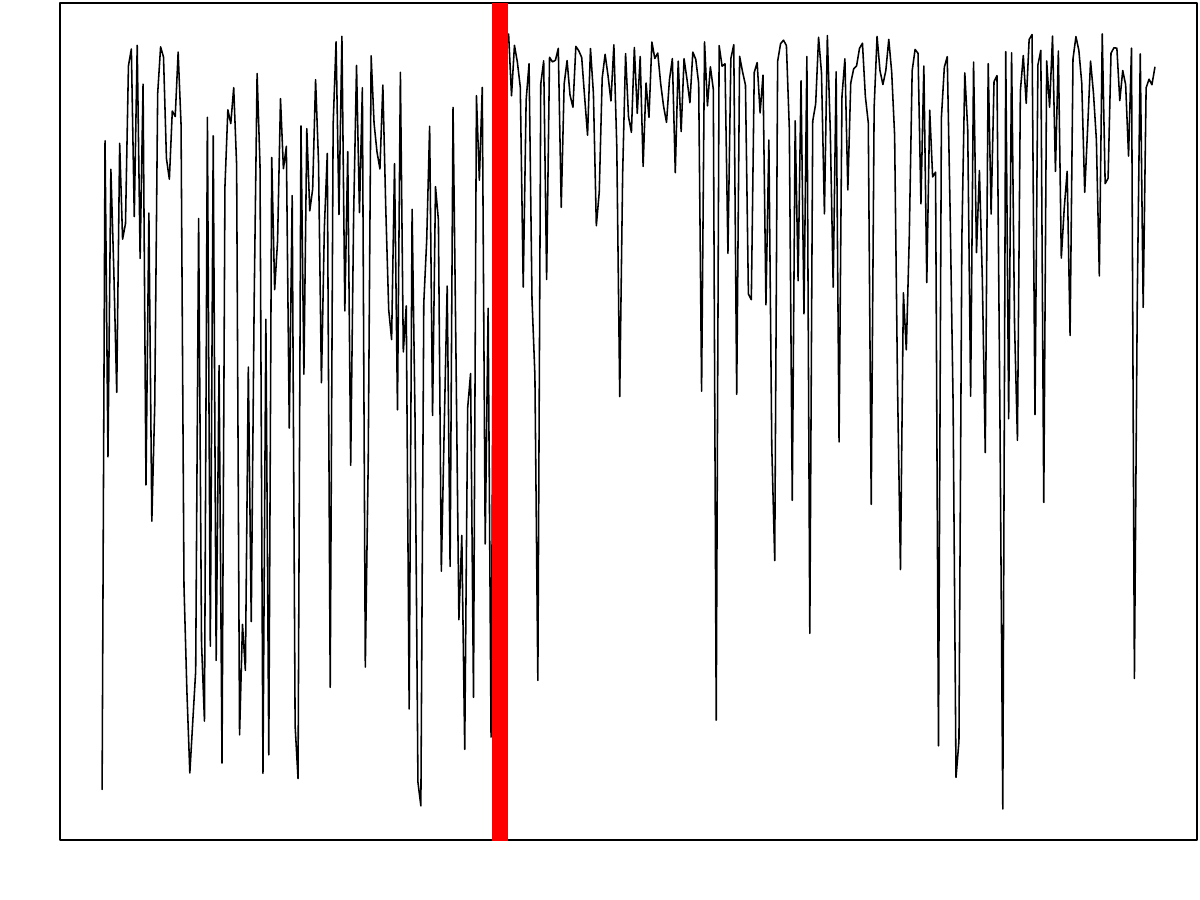}};
			& \node (p14) {\includegraphics[scale=0.12]{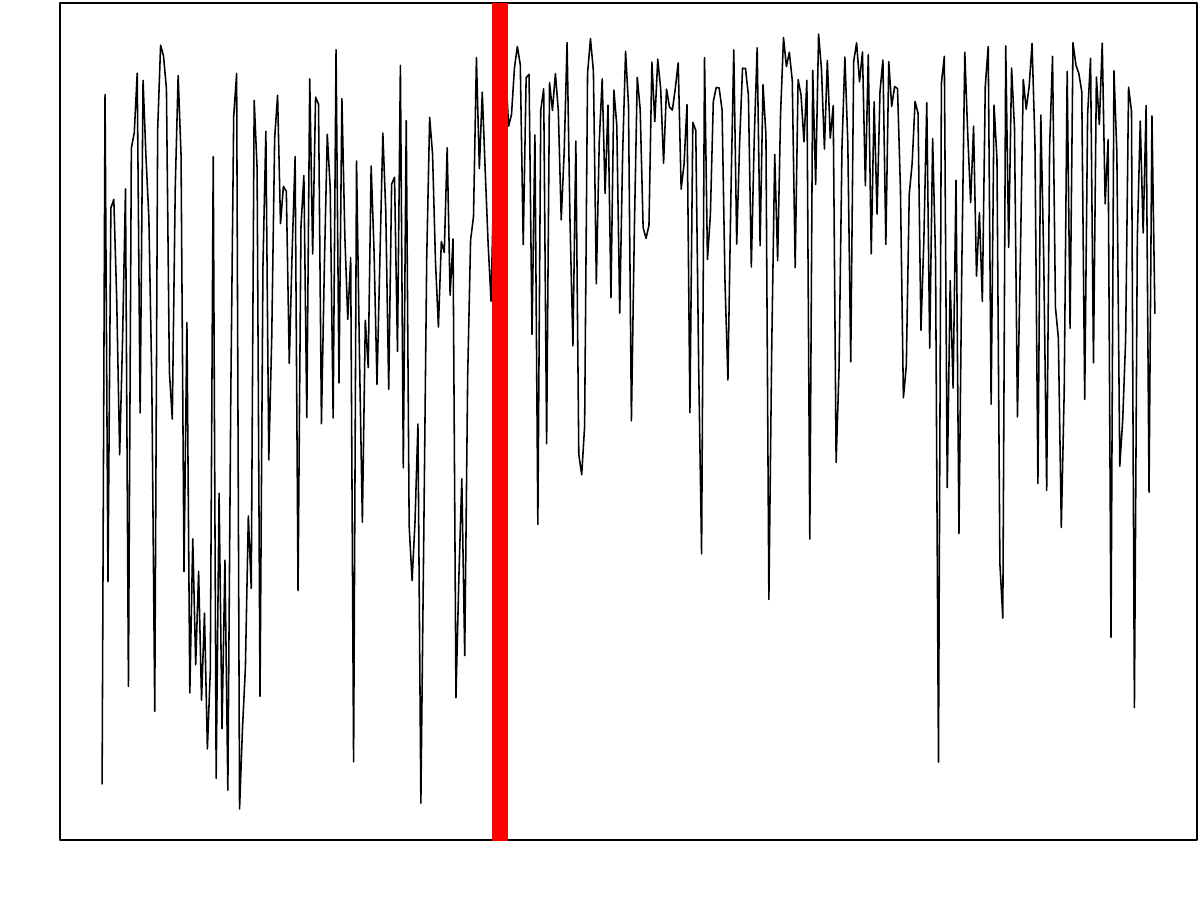}};  & \node (p15) {\includegraphics[scale=0.12]{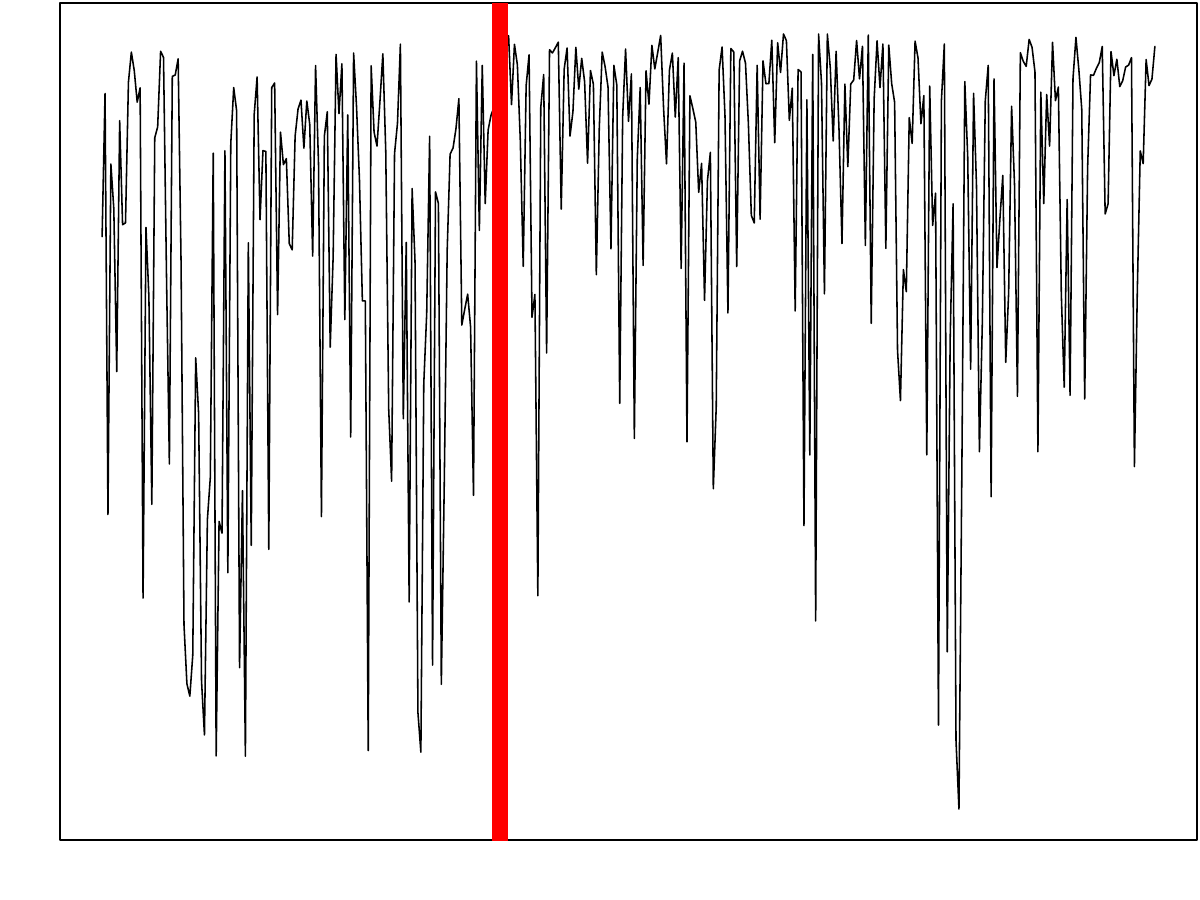}}; 
			\\ 
			&  &\node (p22) {Cardano};
			& \node (p14) {\includegraphics[scale=0.12]{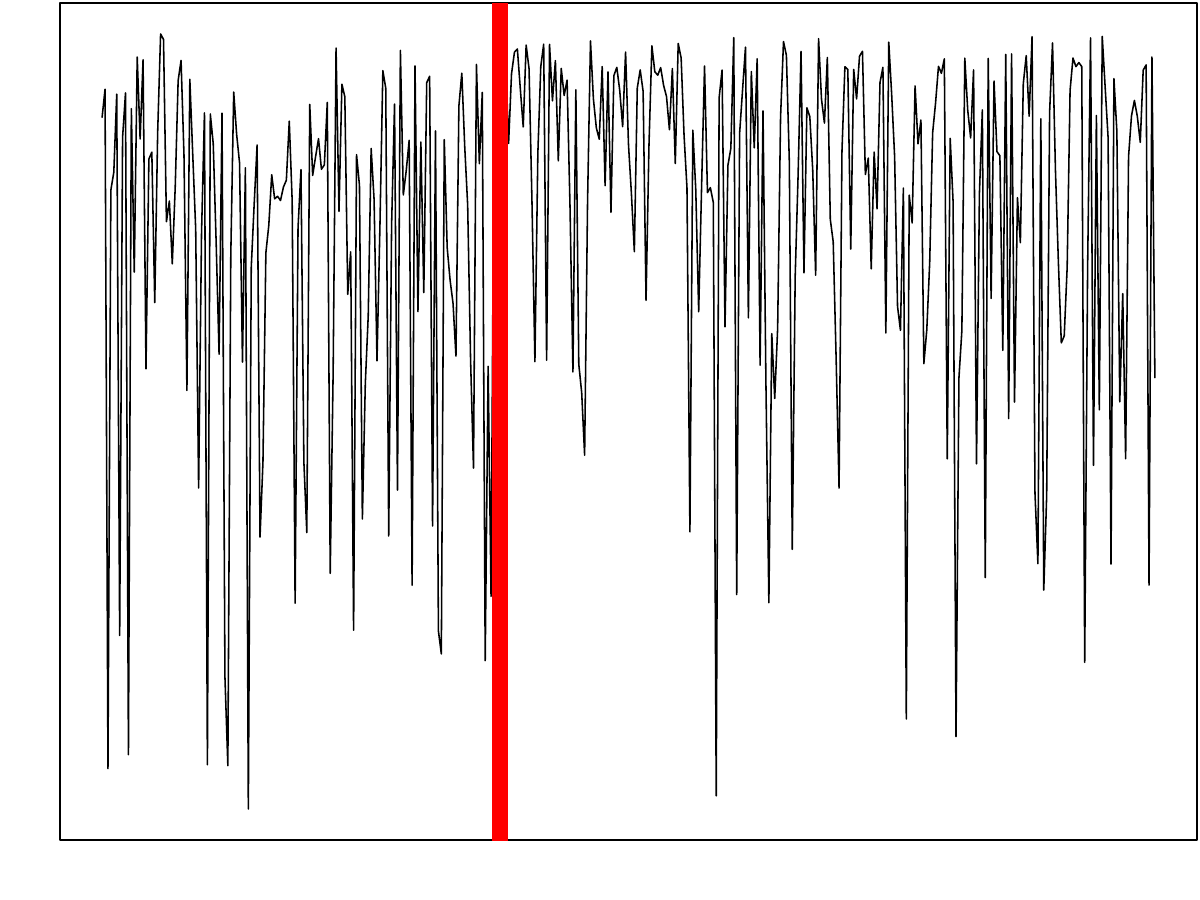}};  & \node (p15) {\includegraphics[scale=0.12]{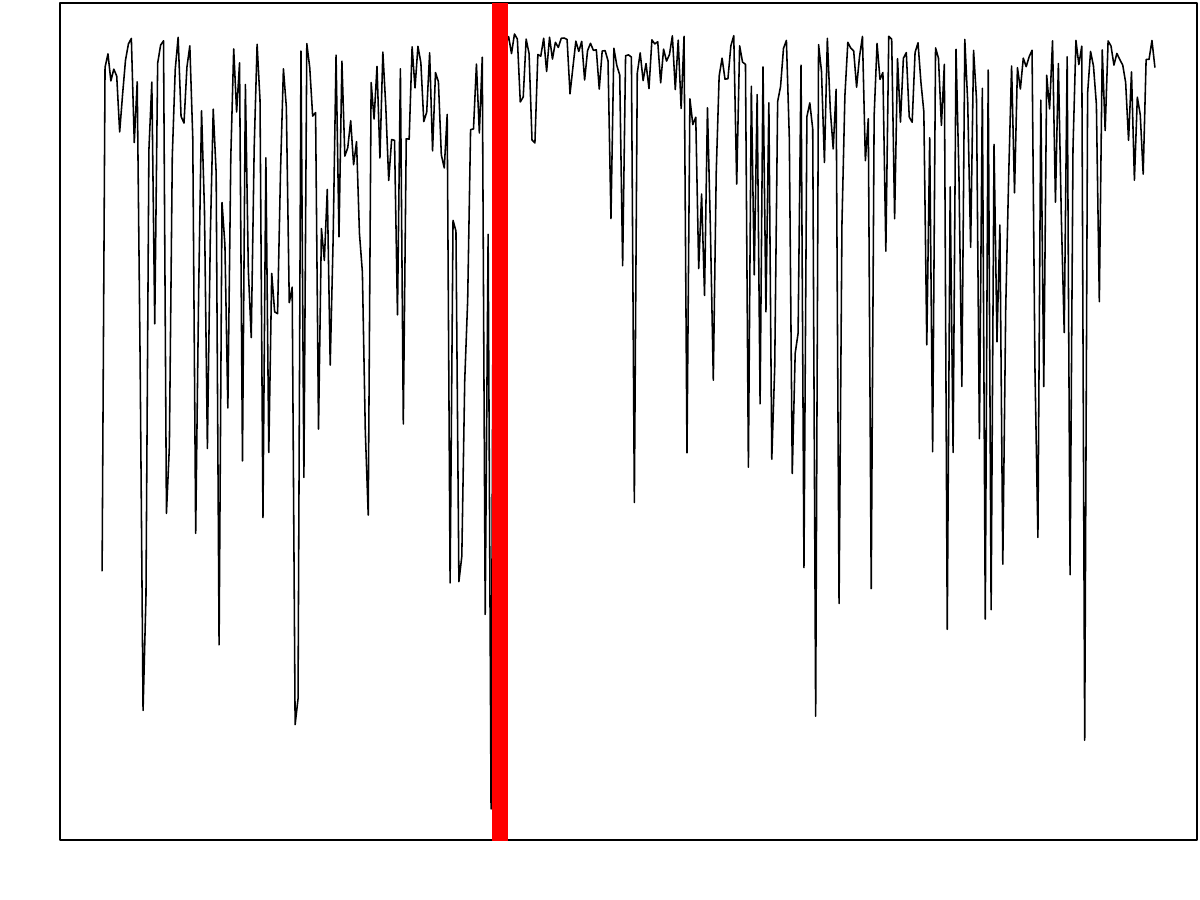}}; 
			\\ 
			&  &
			&\node (p22) {Monero};   & \node (p15) {\includegraphics[scale=0.12]{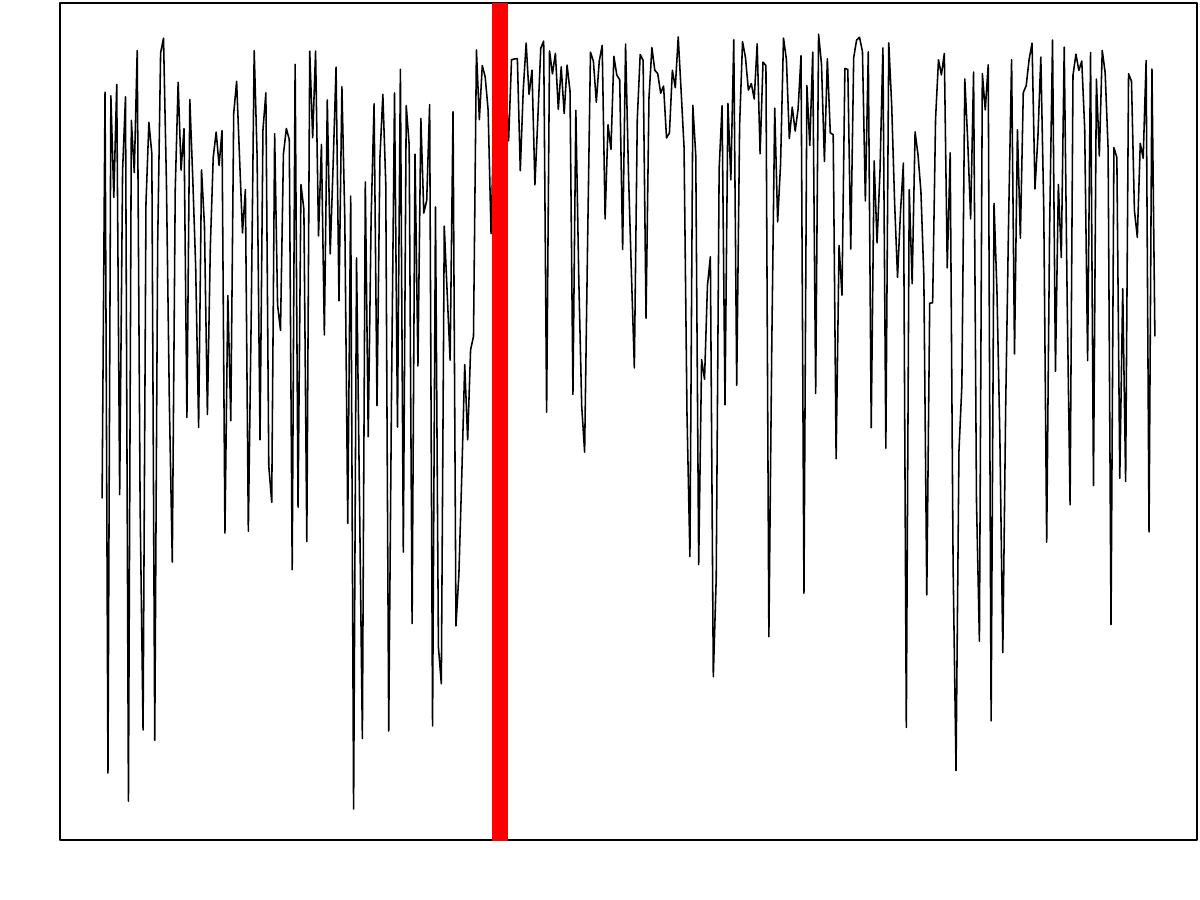}}; 
			\\  
			&  &
			&   &\node (p22) {Chainlink};
			\\  
		};
	\end{tikzpicture}
	\caption{Plot of time series of pairwise correlations. Red vertical line indicates detected change point by WBS-$SN_2$.}
	\label{fig:crypto}
\end{figure}
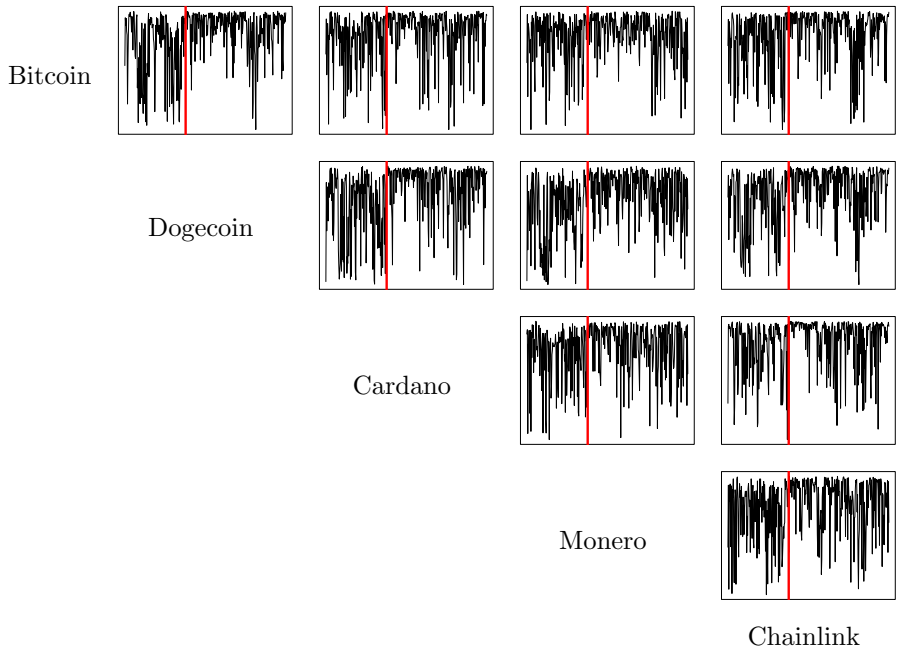

\section{Conclusion}\label{sec:con}

Motivated by increasing availability of non-Euclidean time series data, this paper considers two-sample testing and change-point detection for temporally dependent random objects. Our inferential framework builds upon the nascent SN technique, which has been mainly developed for  conventional Euclidean time series or functional time series in Hilbert space,  and the extension of SN to the time series of objects residing in metric spaces is the first in the literature. The proposed tests are robust to weak temporal dependence, enjoy effortless tuning and  are  broadly applicable to many non-Euclidean data types with easily verified technical conditions. On the theory front, we derive the asymptotic distributions of our two sample and change-point tests  under both null and local alternatives of order $O(n^{-1/2})$. Furthermore, for change-point problem, the consistency of the change-point estimator is established under mild conditions. 
Both simulation and real data illustrations demonstrate the robustness of our test with respect to temporal dependence and the effectiveness in testing and estimation problems. 


To conclude, we mention several  interesting but unsolved  problems for analyzing non-Euclidean time series. 
For example, although powerful against Fr\'echet mean differences/changes, the testing statistics developed in this paper rely on the asymptotic behaviors of  Fr\'echet (sub)sample variances.  It is  imperative to construct formal tests that can target directly at  Fr\'echet mean differences/changes. For the change-point detection problem in non-Euclidean data, the existing literature, including this paper, only derives the consistency of the change-point estimator. It would be very useful to derive explicit convergence rate and the asymptotic distribution of the change-point estimator, which is needed for confidence interval construction.  Also it would be interesting to study how to detect structural changes when the underlying distributions of  random objects change smoothly. We leave these topics for future investigation.

\appendix

\section{Technical Proofs}\label{sec:proof}
\subsection{Auxiliary Lemmas}\label{sec:lemma}
We first introduce some notations. We denote $o_{up}(\cdot)$ as the uniform $o_p(\cdot)$ w.r.t. the partial sum index $(a,b)\in\mathcal{I}_{\eta}$.  Let $M_{n}(\omega,[a,b])=n^{-1}\sum_{t=\lfloor na\rfloor+1}^{\lfloor nb\rfloor}f_{\omega}(Y_t)$, where $f_{\omega}(Y)=d^2(Y,\omega)-d^2(Y,\mu)$, then it is clear that 
$$
\hat{\mu}_{[a,b]}=\arg\min_{\omega\in\Omega}M_n(\omega,[a,b]).
$$
Let  $\tilde{V}_{[a,b]}=\frac{1}{\lfloor n b\rfloor-\lfloor n a\rfloor} \sum_{t=\lfloor n a\rfloor+1}^{\lfloor n b\rfloor} d^{2}\left(Y_{t}, \mu\right)$.

The following three main lemmas are verified under Assumption \ref{ass_diff}-\ref{ass_expand}, and they are used repeatedly throughout the proof for main theorems. 
\begin{lem}\label{lem_drate}
	$\sup_{(a,b)\in \mathcal{I}_{\eta}}\sqrt{n}d(\hat{\mu}_{[a,b]},\mu)=O_p(1)$.
\end{lem}
\begin{proof}
	(1). We first show the uniform convergence, i.e. 
	$
	\sup_{(a,b)\in\mathcal{I}_{\eta}}d(\hat{\mu}_{[a,b]},\mu)=o_{up}(1).
	$
	
	For any $\epsilon>0$, define \begin{equation}\label{psi}
\psi(\epsilon):=\inf_{d(\omega,\mu)>\epsilon} \mathbb{E}f_{\omega}(Y),
	\end{equation}
	and we know by that $\psi(\epsilon)>0$ by the uniqueness of $\mu$ in Assumption \ref{ass_diff}.
	
	Hence, let $M(\omega,[a,b])=(b-a)\mathbb{E}f_{\omega}(Y)$, we have
	\begin{align*}
		&P\Bigg(\sup_{(a,b)\in\mathcal{I}_{\eta}}d(\hat{\mu}_{[a,b]},\mu)>\epsilon\Bigg)
  \\=&{P\Bigg(\bigcup_{(a,b)\in\mathcal{I}_{\eta}}\left\{d(\hat{\mu}_{[a,b]},\mu)>\epsilon\right\}\Bigg) }
		\\\leq & P\Bigg(\bigcup_{(a,b)\in\mathcal{I}_{\eta}}\left\{M(\hat{\mu}_{[a,b]},[a,b])-\inf_{d(\omega,\mu)> \epsilon}M(\omega,[a,b])\geq 0\right\}\Bigg)\\
		\leq & P\Bigg(\bigcup_{(a,b)\in\mathcal{I}_{\eta}}\left\{M(\hat{\mu}_{[a,b]},[a,b])\geq \eta\psi(\epsilon)/2\right\}\Bigg)
		\\ \leq &
		P\Bigg(\bigcup_{(a,b)\in\mathcal{I}_{\eta}}\Bigg\{M(\hat{\mu}_{[a,b]},[a,b])-M_n(\hat{\mu}_{[a,b]},[a,b]) \\
		 & \hspace{4cm} +M_n(\mu,[a,b])-M(\mu,[a,b])\geq \eta\psi(\epsilon)/2\Bigg\}\Bigg)
		\\ \leq & P\Bigg(\sup_{(a,b)\in\mathcal{I}_{\eta}}\sup_{\omega\in\Omega}\left|M_n(\omega,[a,b])-M(\omega,[a,b])\right|\geq \eta\psi(\epsilon)/4\Bigg)
	\end{align*}
{ where the first inequality holds because the event $\{d(\hat{\mu}_{[a,b]},\mu)>\epsilon\}$ implies that  $\hat{\mu}_{[a,b]}\in \{\omega\in\Omega:d(\omega,\mu)> \epsilon \}$, and thus $M(\hat{\mu}_{[a,b]},[a,b])\geq \inf_{d(\omega,\mu)>\epsilon}M(\omega,[a,b])$}; the second  inequality holds by $b-a\geq \eta$ (hence $({\lfloor nb\rfloor-\lfloor na\rfloor})/{n}>\eta/2$ for large $n$) 
 { and the definition of \eqref{psi} such that $\inf_{d(\omega,\mu)>\epsilon}M(\omega,[a,b])=(b-a)\psi(\epsilon)>\eta \psi(\epsilon)/2$};  and the third holds by that $M(\mu,[a,b])=0$ and $M_n(\mu,[a,b])\geq M_n(\hat{\mu}_{[a,b]},[a,b])$.
	
	Note $M_n(\omega,[a,b])-M(\omega,[a,b])=M_n(\omega,[0,b])-M(\omega,[0,b])-M_n(\omega,[0,a])+M(\omega,[0,a])$. Therefore, it suffices to show the  weak convergence of the process $\{M_n(\omega,[0,u])-M(\omega,[0,u]), u\in[0,1],\omega\in \Omega\}$ to zero.  Note the pointwise convergence holds easily by the boundedness of $f_{\omega}$ and Assumption \ref{ass_LLN}, so we only need to show the stochastic equicontinuity, i.e. 
	\begin{align*}
	&\limsup_{n\to\infty}P\Bigg(\sup_{|u-v|<\delta_1,d(\omega_1,\omega_2)<\delta_2}|M_n(\omega_1,[0,u])-M(\omega_1,[0,u]) \\
	& \hspace{4cm}-M_n(\omega_2,[0,v])+M(\omega_2,[0,v])|>\epsilon\Bigg)\to 0
	\end{align*}
	as $\max(\delta_1,\delta_2)\to 0$. 
	
	Then, by triangle inequality,  we have 
	\begin{flalign*}
		&|M_n(\omega_1,[0,u])-M(\omega_1,[0,u])-M_n(\omega_2,[0,v])+M(\omega_2,[0,v])|
		\\\leq &|M_n(\omega_1,[0,u])-M_n(\omega_1,[0,v])|+|M_n(\omega_1,[0,v])-M_n(\omega_2,[0,v])|\\&+|M(\omega_1,[0,u])-M(\omega_1,[0,v])|+|M(\omega_1,[0,v])-M(\omega_2,[0,v])|
		\\:= &\sum_{i=1}^4 R_{n,i}.
	\end{flalign*}
	Without loss of generality, we assume $v>u$, and by the boundedness of the metric, we have for some $K>0$,
	\begin{flalign*}
		R_{n,1}\leq n^{-1}\sum_{t=\lfloor nu\rfloor+1}^{\lfloor nv\rfloor}d^2(Y_t,\omega_1)\leq K|u-v|\leq K\delta_1.
	\end{flalign*}
	Similarly, $R_{n,3}\leq K$.  Furthermore, we can see that 
	$
	R_{n,2},R_{n,4}\leq 2\mathrm{diam}(\Omega)d(\omega_1,\omega_2)\leq K\delta_2.
	$
	Hence, the result follows by letting $\delta_1$ and $\delta_2$ sufficiently small.
	
	Thus, the uniform convergence holds.
	
	(2). We then derive the convergence rate based on Assumption \ref{ass_expand}.
	
	By the consistency, we have for any $\delta>0$, $P(\sup_{(a,b)\in\mathcal{I}_{\eta}}d(\hat{\mu}_{[a,b]},\mu)\leq \delta)\to 1$. Hence, on the event that  $\sup_{(a,b)\in\mathcal{I}_{\eta}}d(\hat{\mu}_{a,b},\mu)\leq \delta$,  { and  note that $M_n(\mu,[a,b])=n^{-1}\sum_{t=\lfloor na\rfloor+1}^{\lfloor nb\rfloor}[d^2(Y_t,\mu)-d^2(Y_t,\mu)]=0$,} we have 
	\begin{flalign*}
		 0=&M_n(\mu,[a,b]) \\
		\geq& M_n(\hat{\mu}_{[a,b]},[a,b])\\
		=&  K_d\frac{\lfloor nb\rfloor-\lfloor na\rfloor}{n}d^2(\hat{\mu}_{[a,b]},\mu)  + n^{-1}\sum_{t=\lfloor na\rfloor+1}^{\lfloor nb\rfloor}\left[g(Y_t,\hat{\mu}_{[a,b]},\mu)+R(Y_t,\hat{\mu}_{[a,b]},\mu)\right]
		\\\geq & \frac{K_d \eta}{2}d^2(\hat{\mu}_{[a,b]},\mu)  \\
		&\hspace{0.5cm}+d(\hat{\mu}_{[a,b]},\mu)\left[ \frac{n^{-1}\sum_{t=\lfloor na\rfloor+1}^{\lfloor nb\rfloor}g(Y_t,\hat{\mu}_{[a,b]},\mu)}{d(\hat{\mu}_{[a,b]},\mu)}+o_{up}(n^{-1/2}+d(\hat{\mu}_{[a,b]},\mu))\right],
	\end{flalign*}
	where the last inequality holds by Assumption \ref{ass_expand} and the fact $({\lfloor nb\rfloor-\lfloor na\rfloor})/{n}>\eta/2$ for large $n$.
	
	Note the above analysis holds uniformly for $(a,b)\in\mathcal{I}_{\eta}$, this implies that 
	\begin{flalign*}
		&\sup_{(a,b)\in\mathcal{I}_{\eta}}\left[\frac{K_d \eta}{2}d(\hat{\mu}_{[a,b]},\mu)-o_{up}(d(\hat{\mu}_{[a,b]},\mu))\right]\\\leq&  n^{-1/2} \sup_{(a,b)\in\mathcal{I}_{\eta}}\left| \frac{n^{-1/2}\sum_{t=\lfloor na\rfloor+1}^{\lfloor nb\rfloor}g(Y_t,\hat{\mu}_{[a,b]},\mu)}{d(\hat{\mu}_{[a,b]},\mu)}\right|+o_{up}(n^{-1/2})=O_p(n^{-1/2}),
	\end{flalign*}
	and hence $\sup_{(a,b)\in\mathcal{I}_{\eta}}d(\hat{\mu}_{[a,b]},\mu)=O_p(n^{-1/2})$.
\end{proof}

\begin{lem}\label{lem_Vhat}
	$\sup_{(a,b)\in \mathcal{I}_{\eta}}\sqrt{n}|\hat{V}_{[a,b]}-\tilde{V}_{[a,b]}|=o_p(1)$.
\end{lem}

\begin{proof}
	By Lemma \ref{lem_drate}, and Assumption \ref{ass_expand}, we have \begin{flalign*}
		&\sup_{(a,b)\in\mathcal{I}_{\eta}}\sqrt{n}M_n(\hat{\mu}_{[a,b]},[a,b])
		\\\leq& K_d\sup_{(a,b)\in\mathcal{I}_{\eta}}d(\hat{\mu}_{[a,b]},\mu) \sup_{(a,b)\in\mathcal{I}_{\eta}}\Bigg|\sqrt{n}d(\hat{\mu}_{[a,b]},\mu) \\
		& \hspace{1cm} + \frac{n^{-1/2}\sum_{t=\lfloor na\rfloor+1}^{\lfloor nb\rfloor}g(Y_t,\hat{\mu}_{[a,b]},\mu)}{d(\hat{\mu}_{[a,b]},\mu)}+o_{up}(1+\sqrt{n}d(\hat{\mu}_{[a,b]},\mu))\Bigg|
		\\=&O_p(n^{-1/2}).
	\end{flalign*}
	Hence, we have that 
	$$
	\sup_{(a,b)\in \mathcal{I}_{\eta}}\sqrt{n}|\hat{V}_{[a,b]}-\tilde{V}_{[a,b]}|\leq \eta^{-1}\sup_{(a,b)\in \mathcal{I}_{\eta}}\sqrt{n}M_n(\hat{\mu}_{[a,b]},[a,b]),
	$$
	the result follows. 
\end{proof}

\begin{lem}\label{lem_VhatC}
	Let $
	\hat{V}^C_{[a,b]}(\tilde{\omega})=\frac{1}{\lfloor nb\rfloor-\lfloor na \rfloor}\sum_{t=\lfloor na\rfloor+1}^{\lfloor nb\rfloor}d^2(Y_i,\tilde{\omega}),
	$ where $\tilde{\omega}\in\Omega$ is a random object such that 
	$$
	\sqrt{n}\sup_{(a,b)\in\mathcal{I}_{\eta}}d(\tilde{\omega},\hat{\mu}_{[a,b]})=O_p(1).
	$$
	
	Then, $$\sqrt{n}\sup_{(a,b)\in\mathcal{I}_{\eta}}|\hat{V}^C_{[a,b]}(\tilde{\omega})-\tilde{V}_{[a,b]}|=o_p(1).$$
\end{lem}
\begin{proof}
	By triangle inequality and Lemma \ref{lem_Vhat},
	\begin{flalign*}
		& \sqrt{n}\sup_{(a,b)\in\mathcal{I}_{\eta}}|\hat{V}^C_{[a,b]}(\tilde{\omega})-\tilde{V}_{[a,b]}| \\ 
		= &\sup_{(a,b)\in\mathcal{I}_{\eta}}\left|\frac{\sqrt{n}}{\lfloor n b\rfloor-\lfloor n a\rfloor} \sum_{t=\lfloor n a\rfloor+1}^{\lfloor n b\rfloor} d^{2}\left(Y_{t}, \tilde{\omega}\right)-d^2(Y_i,\mu)\right|
		\\\leq &(\eta/2)^{-1}\sup_{(a,b)\in\mathcal{I}_{\eta}}\sqrt{n}M_n(\tilde{\omega},[a,b]).
	\end{flalign*}
	Note by triangle inequality for the metric, $d(\tilde{\omega},\mu)\leq d(\hat{\mu}_{[a,b]},\mu)+d(\tilde{\omega},\hat{\mu}_{[a,b]})=O_p(n^{-1/2})$, and  we know that $d(\tilde{\omega},\mu)<\delta$ with probability tending to 1, and on this event, by Assumption \ref{ass_expand},
	\begin{flalign*}
		\sqrt{n}M_n(\tilde{\omega},[a,b])
		\leq & K_dd^2(\tilde{\omega},\mu)  +n^{-1}\left| \sum_{t=\lfloor na\rfloor+1}^{\lfloor nb\rfloor}g(Y_t,\tilde{\omega},\mu)\right|\\ &\hspace{1cm}+n^{-1}\left|\sum_{t=\lfloor na\rfloor+1}^{\lfloor nb\rfloor}R(Y_t,\tilde{\omega},\mu)\right|.
	\end{flalign*}
	Similar to the proof of Lemma \ref{lem_Vhat},  we get the result.
\end{proof}

\subsection{Proof of Theorems in Section \ref{sec:two}}
\label{sec:proof3}
Let $\tilde{V}^{(1)}_{r}=\frac{1}{\lfloor rn_1\rfloor}\sum_{t=1}^{\lfloor rn_1\rfloor}d^2(Y_t^{(1)},\mu^{(1)})$, and  $\tilde{V}^{(2)}_{r}=\frac{1}{\lfloor rn_2\rfloor}\sum_{t=1}^{\lfloor rn_2\rfloor}d^2(Y_t^{(2)},\mu^{(2)})$.

For each $r\in[\eta,1]$, we consider the  decomposition,
\begin{flalign}\label{T_decomp}
	\begin{split}
		\sqrt{n}T_n(r)=&{\sqrt{n}}r(\hat{V}^{(1)}_{r}-\hat{V}^{(2)}_{r})\\
		=&{\sqrt{n}}r(\hat{V}^{(1)}_{r}-\tilde{V}^{(1)}_{r}+\tilde{V}^{(1)}_{r}-V^{(1)})  \\ 
		& -{\sqrt{n}}r(\hat{V}^{(2)}_{r}-\tilde{V}^{(2)}_{r}+\tilde{V}^{(2)}_{r}-V^{(2)}) \\
		&+{\sqrt{n}}r(V^{(1)}-V^{(2)})
		\\:=&R_{n,1}(r)+R_{n,2}(r)+R_{n,3}(r).
	\end{split}
\end{flalign}
and 
\begin{flalign}\label{TC_decomp}
	\begin{split}
		\sqrt{n}T_n^C(r)=&
		\sqrt{n}r(\hat{V}^{C,(1)}_{r}-\tilde{V}^{(1)}_{r})-\sqrt{n}r(\hat{V}^{(1)}_{r}-\tilde{V}^{(1)}_{r}) \\
		& +\sqrt{n}r(\hat{V}^{C,(2)}_{r} -\tilde{V}^{(2)}_{r})-\sqrt{n}r(\hat{V}^{(2)}_{r}-\tilde{V}^{(2)}_{r})\\
		:=&R^C_{n,1}(r)+R^C_{n,2}(r)+R^C_{n,3}(r)+R^C_{n,4}(r).
	\end{split}
\end{flalign}

By Lemma \ref{lem_Vhat},  
\begin{equation}\label{tilde_diff}
	\sup_{r\in[\eta,1]}\sqrt{n}r(\hat{V}^{(1)}_r-\tilde{V}^{(1)}_r)=o_p(1),\quad \sup_{r\in[\eta,1]}\sqrt{n}r(\hat{V}^{(2)}_r-\tilde{V}^{(2)}_r)=o_p(1),
\end{equation}
i.e. 
\begin{equation}\label{R2+4C}
	\left\{R^C_{n,2}(r)+R^C_{n,4}(r)\right\}_{r\in[\eta,1]}\Rightarrow 0.
\end{equation}
Furthermore, by Assumption \ref{ass_FCLT},
$$
\sqrt{n}r(\hat{V}^{(1)}_r-V^{(1)})\Rightarrow \gamma_1^{-1}\sigma_1B^{(1)}(\gamma_1r),\quad \sqrt{n}r(\hat{V}^{(2)}_r-V^{(2)})\Rightarrow \gamma_2^{-1}\sigma_2B^{(2)}(\gamma_2r).
$$
This implies that 
\begin{equation}\label{R1+2}
	\left\{R_{n,1}(r)+R_{n,2}(r)\right\}_{r\in[\eta,1]}\Rightarrow \{\xi_{\gamma_1,\gamma_2}(r;\sigma_1,\sigma_2)\}_{r\in[\eta,1]}.
\end{equation}

\subsection*{Proof of Theorem \ref{thmtwo_0}}
Under $\mathbb{H}_0$, $R_{n,3}(r)\equiv 0$, and $\mu^{(1)}=\mu^{(2)}=\mu$. Hence, by \eqref{T_decomp} and \eqref{R1+2}, we obtain that 
$$
\{\sqrt{n}T_n(r)\}_{r\in[\eta,1]} \Rightarrow \{\xi_{\gamma_1,\gamma_2}(r;\sigma_1,\sigma_2)\}_{r\in[\eta,1]}.
$$
Next, by Lemma \ref{lem_drate}, can obtain that
$$
\sqrt{n}\sup_{r\in[\eta,1]}d(\hat{\mu}^{(1)}_{r},\mu)=o_p(1),\quad \sqrt{n}\sup_{r\in[\eta,1]}d(\hat{\mu}^{(2)}_{r},\mu)=o_p(1).
$$
Hence, by Lemma \ref{lem_VhatC},  we have 
$$\left\{R^C_{n,1}(r)+R^C_{n,3}(r)\right\}_{r\in[\eta,1]}\Rightarrow 0.$$ Together with \eqref{R2+4C}, we have 
$$
\{\sqrt{n}T^C_n(r)\}_{r\in[\eta,1]} \Rightarrow 0.
$$
Hence, by continuous mapping theorem, for both $i=1,2$,
$$
D_{n,i}\to_d \frac{\xi^2_{\gamma_1,\gamma_2}(1;\sigma_1,\sigma_2)}{\int_{\eta}^{1}\left[\xi_{\gamma_1,\gamma_2}(r;\sigma_1,\sigma_2)-r\xi_{\gamma_1,\gamma_2}(1;\sigma_1,\sigma_2)\right]^2dr}.
$$

\qed

\subsection*{Proof of Theorem \ref{thmtwo_alter}}

In view of \eqref{T_decomp} and \eqref{R1+2}, $$\left\{\sqrt{n}T_n(r)\right\}_{r\in[\eta,1]}\Rightarrow \left\{ \xi_{\gamma_1,\gamma_2}(r;\sigma_1,\sigma_2)+rn^{-\kappa_V+1/2}\Delta_V\right\}_{r\in[\eta,1]}.$$
Hence 
\begin{itemize}
	\item For $\kappa_V \in(1/2,\infty)$,  $\left\{\sqrt{n}T_n(r)\right\}\Rightarrow \left\{ \xi_{\gamma_1,\gamma_2}(r;\sigma_1,\sigma_2)\right\}_{r\in[\eta,1]}.$
	\item For $\kappa_V=1/2$, $\left\{\sqrt{n}T_n(r)\right\}_{r\in[\eta,1]}\Rightarrow \left\{ \xi_{\gamma_1,\gamma_2}(r;\sigma_1,\sigma_2)+r\Delta_V\right\}_{r\in[\eta,1]}.$
	\item For $\kappa_V\in(0,1/2)$, $\sqrt{n}T_n(1)\to_p\infty$, and $\left\{\sqrt{n}T_n(r)-\sqrt{n}rT_n(1)\right\}_{r\in[\eta,1]}\Rightarrow\left\{ \xi_{\gamma_1,\gamma_2}(r;\sigma_1,\sigma_2)-r\xi_{\gamma_1,\gamma_2}(1;\sigma_1,\sigma_2)\right\}_{r\in[\eta,1]}$. 
\end{itemize}
Next, we focus on $\sqrt{n}T_n^C(r)$. 
When $\kappa_M\in (0,\infty)$, it holds that $d(\mu^{(1)},\mu^{(2)})=O(n^{-\kappa_M/2})=o(1)$, and by triangle inequality,  for any $r\in[\eta,1]$,
\begin{equation}\label{tri_ineq}
	|d(\mu^{(1)},\mu^{(2)})-d(\hat{\mu}^{(2)}_{r},\mu^{(2)})|\leq d(\hat{\mu}^{(2)}_{r},\mu^{(1)})\leq |d(\mu^{(1)},\mu^{(2)})+d(\hat{\mu}^{(2)}_{r},\mu^{(2)})|.
\end{equation}
By Lemma \ref{lem_drate}, we have $\sup_{r\in[\eta,1]}d(\hat{\mu}^{(2)}_{r},\mu^{(2)})=O_p(n^{-1/2})$.
This and \eqref{tri_ineq} imply that 
\begin{itemize} 
	\item when $\kappa_M\in(1/2,\infty)$, $d^2(\hat{\mu}^{(2)}_{r},\mu^{(1)})=o_{up}(n^{-1/2})$;
	\item when $\kappa_M\in(0,1/2]$, $d^2(\hat{\mu}^{(2)}_{r},\mu^{(1)})=d^2(\mu^{(1)},\mu^{(2)})+o_{up}(n^{-1/2})=n^{-\kappa_M}\Delta_M+o_{up}(n^{-1/2})$.
\end{itemize}
Similarly,
\begin{itemize} 
	\item when $\kappa_M\in(1/2,\infty)$, $d^2(\hat{\mu}^{(1)}_{r},\mu^{(2)})=o_{up}(n^{-1/2})$;
	\item when $\kappa_M\in(0,1/2]$, $d^2(\hat{\mu}^{(1)}_{r},\mu^{(2)})=n^{-\kappa_M}\Delta_M+o_{up}(n^{-1/2})$.
\end{itemize}
Furthermore, by  Assumption \ref{ass_expand}, equations \eqref{TC_decomp} and \eqref{R2+4C}, we obtain 
\begin{flalign}\label{TnC}
	\begin{split}
		&\sqrt{n}T_n^C(r)=R_{n,1}^C(r)+R_{n,3}^C(r)+o_{up}(1)
		\\=&\sqrt{n}K_d rd^2(\hat{\mu}^{(2)}_{r},\mu^{(1)})+ rd(\hat{\mu}^{(2)}_{r},\mu^{(1)})\left[\frac{n^{-1/2}\sum_{t=1}^{\lfloor \gamma_1nr\rfloor}g(Y_t^{(1)},\hat{\mu}^{(2)}_{r},\mu^{(1)})}{d(\hat{\mu}^{(2)}_{r},\mu^{(1)})}\right] \\
		& \hspace{1cm} +o_{up}(d(\hat{\mu}^{(2)}_{r},\mu^{(1)})+\sqrt{n}d^2(\hat{\mu}^{(2)}_{r},\mu^{(1)}))\\
		&+\sqrt{n}K_dr d^2(\hat{\mu}^{(1)}_{r},\mu^{(2)})+ rd(\hat{\mu}^{(1)}_{r},\mu^{(2)})\left[\frac{n^{-1/2}\sum_{t=1}^{\lfloor \gamma_2nr\rfloor}g(Y_t^{(2)},\hat{\mu}^{(1)}_{r},\mu^{(2)})}{d(\hat{\mu}^{(1)}_{r},\mu^{(2)})}\right]\\
		& \hspace{1cm} +o_{up}(d(\hat{\mu}^{(1)}_{r},\mu^{(2)})+\sqrt{n}d^2(\hat{\mu}^{(1)}_{r},\mu^{(2)}))
		\\&+o_{up}(1).
	\end{split}
\end{flalign}
\begin{itemize}
	\item For $\kappa_M \in(1/2,\infty)$, $
	d^2(\hat{\mu}^{(2)}_{r},\mu^{(1)})=o_{up}(n^{-1/2}),$ and $
	d^2(\hat{\mu}^{(1)}_{r},\mu^{(2)})=o_{up}(n^{-1/2}).$
	Hence, $\{\sqrt{n}T_n^C(r)\}_{r\in[\eta,1]}\Rightarrow 0$.
	\item For $\kappa_M=1/2$, we note that 
	$
	d^2(\hat{\mu}^{(2)}_{r},\mu^{(1)})=n^{-1/2}\Delta_M+o_{up}(1),$ and $
	d^2(\hat{\mu}^{(1)}_{n},\mu^{(2)})=n^{-1/2}\Delta_M+o_{up}(1).
	$ 
	Hence, $\{\sqrt{n}T_n^C(r)\}_{r\in[\eta,1]}\Rightarrow \{2rK_d\Delta_M\}_{r\in[\eta,1]}$, and $\{\sqrt{n}[T_n^C(r)-rT_n^C(1)]\}_{r\in[\eta,1]}\Rightarrow 0$.
	\item For  $\kappa_M\in(0,1/2)$, we multiply $n^{2\kappa_M-1}$ on both denominator and numerator of $D_{n,2}$, and obtain 
	\begin{equation}\label{Dn2_equiv}
		D_{n,2}=\frac{n^{2\kappa_M}\left\{\left[T_n(1)\right]^2+\left[T_n^C(1)\right]^2\right\}}{n^{-1}\sum_{k=\lfloor n\eta\rfloor}^n  n^{2\kappa_M}\left\{\left[T_n(\frac{k}{n})-\frac{k}{n}T_n(1)\right]^2+\left[T_n^C(\frac{k}{n})-\frac{k}{n}T^C_n(1)\right]^2\right\}}. 
	\end{equation}
	Note that $n^{\kappa_M-1/2}\to0$, as $n\to\infty$, we obtain that 
	\begin{equation}\label{Tnkappa}
		\left\{n^{\kappa_M}[T_n(r)-rT_n(1)\right]\}_{r\in[\eta,1]}\Rightarrow 0.
	\end{equation}
	Furthermore, in view of \eqref{TnC}, we obtain 
	\begin{flalign*}
		n^{\kappa_M}T^C_n(r)=
		n^{\kappa_M}r(K_d+o_{up}(1))\left[d^2(\hat{\mu}^{(2)}_r,\mu^{(1)})+d^2(\hat{\mu}^{(1)}_r,\mu^{(2)})\right]+o_{up}(1),
	\end{flalign*}
	By arguments below \eqref{tri_ineq}, we know that 
	$
	n^{\kappa_M}d^2(\hat{\mu}^{(2)}_r,\mu^{(1)})=\Delta_M+o_{up}(n^{\kappa_M-1/2})=\Delta_M+o_{up}(1).
	$
	And similarly, $
	n^{\kappa_M}d^2(\hat{\mu}^{(1)}_r,\mu^{(2)})=\Delta_M+o_{up}(1).$ We thus obtain that
	\begin{equation}\label{TnCkappa}
		\left\{n^{\kappa_M} T^C_n(r)-rT^C_n(1)\right\}_{r\in[\eta,1]}\Rightarrow 0,
	\end{equation}
	and 
	\begin{equation}\label{TnCkappa_num}
		n^{\kappa_M}T_n^C(1)\to_p 2K_d\Delta_M.
	\end{equation}
	Therefore,  \eqref{Tnkappa} and \eqref{TnCkappa} implies that the  denominator of \eqref{Dn2_equiv} converges to 0 in probability, while \eqref{TnCkappa_num} implies the numerator of \eqref{Dn2_equiv} is larger than a positive constant in probability, we thus obtain $D_{n,2}\to_p\infty$.
\end{itemize} 
Summarizing the cases of $\kappa_V$ and $\kappa_M$, and by continuous mapping theorem, we get the result.
\qed

\subsection*{Proof of \Cref{cor_xi}}
When $\gamma_1=\gamma_2=1/2$, it can be shown that $$\xi_{\gamma_1,\gamma_2}(r;\sigma_1,\sigma_2)=2\sigma_1B^{(1)}(r/2)-2\sigma_2B^{(1)}(r/2)=_d \sqrt{2\sigma_1^2+2\sigma_2^2-4\rho\sigma_1\sigma_2}B(r);$$ 
and when $\rho=0$.
$$\xi_{\gamma_1,\gamma_2}(r;\sigma_1,\sigma_2)=_d \sqrt{\frac{\sigma_1^2}{\gamma_1}+\frac{\sigma_2^2}{\gamma_2}}B(r).$$ 
The result follows by the continuous mapping theorem.

\qed 

\subsection{Proof of Theorems in Section \ref{sec:cpt}} \label{sec:proof4}
With a bit abuse of notation, we define $\mathcal{I}_{\eta}=\{(a,b): 0\leq a<b\leq 1, b-a\geq \eta_2 \}$  and $\mathcal{J}_{\eta}=\{(r;a,b): 0\leq a<r<b\leq 1, b-r\geq \eta_2, r-a\geq \eta_2 \}$.
\subsection*{Proof of Theorem \ref{thm_cpt_H0}}
For $(r;a,b)\in \mathcal{J}_{\eta}$, we note that 
\begin{flalign*}
	&\sqrt{n}T_n(r;a,b)\\=& \sqrt{n}\left\{\frac{(r-a)(b-r)}{(b-a)}\left(\hat{V}_{[a, r]}-\tilde{V}_{[a,r]}+\tilde{V}_{[a,r]}-V\right)\right\} \\ & -\sqrt{n}\left\{\frac{(r-a)(b-r)}{(b-a)}\left(\hat{V}_{[r, b]}-\tilde{V}_{[r, b]}+\tilde{V}_{[r, b]}-V\right)\right\}.
\end{flalign*}
By Lemma \ref{lem_Vhat} we know that $\sup_{(a,r)\in\mathcal{I}_{\eta}}\sqrt{n}|\hat{V}_{[a, r]}-\tilde{V}_{[a,r]}|=o_p(1)$, $\sup_{(r,b)\in\mathcal{I}_{\eta}}\sqrt{n}|\hat{V}_{[r,b]}-\tilde{V}_{[r,b]}|=o_p(1)$, and by Assumption \ref{ass_FCLT}, \begin{flalign*}
	&\left\{\sqrt{n}(r-a)(\tilde{V}_{[a,r]}-V)\right\}_{(a,r)\in\mathcal{I}_{\eta}}\Rightarrow \left\{\sigma [B(r)-B(a)]\right\}_{(a,r)\in\mathcal{I}_{\eta}},\\ &\left\{\sqrt{n}(b-r)(\tilde{V}_{[r,b]}-V)\right\}_{(r,b)\in\mathcal{I}_{\eta}}\Rightarrow \left\{\sigma [B(b)-B(r)]\right\}_{(r,b)\in\mathcal{I}_{\eta}}.
\end{flalign*}
Hence, \begin{flalign*}
    &\left\{\sqrt{n}T_n(r;a,b)\right\}_{(r;a,b)\in \mathcal{J}_{\eta}}\\ \Rightarrow &\sigma\left\{ \frac{(b-r)}{(b-a)}[B(r)-B(a)]-\frac{(r-a)}{(b-a)}[B(b)-B(r)]\right\}_{(r;a,b)\in \mathcal{J}_{\eta}}.
\end{flalign*}

Furthermore, we note that 
\begin{flalign*}
	&\sqrt{n}T_n^C(r;a,b)\\=& \frac{(b-r)}{(b-a)}n^{-1/2}\Bigg\{\sum_{i=\lfloor n a\rfloor+1}^{\lfloor n r\rfloor} \left[d^{2}\left(Y_{i}, \hat{\mu}_{[r, b]}\right)-d^{2}\left(Y_{i}, \mu\right)\right] \\ 
	& \hspace{1cm} - \sum_{i=\lfloor n a\rfloor+1}^{\lfloor n r\rfloor} \left[d^{2}\left(Y_{i}, \hat{\mu}_{[a,r]}\right)-d^{2}\left(Y_{i}, \mu\right)\right]\Bigg\}
	\\&+ \frac{(r-a)}{(b-a)}n^{-1/2}\sum_{i=\lfloor n r\rfloor+1}^{\lfloor n b\rfloor} \Bigg\{\left[d^{2}\left(Y_{i}, \hat{\mu}_{[a, r]}\right)-d^{2}\left(Y_{i}, \mu\right)\right]\\
	& \hspace{1cm} - \sum_{i=\lfloor n r\rfloor+1}^{\lfloor n b\rfloor} \left[d^{2}\left(Y_{i}, \hat{\mu}_{[r, b]}\right)-d^{2}\left(Y_{i}, \mu\right)\right]\Bigg\}+o_{up}(1)
\end{flalign*}
where $o_{up}(1)$ is the rounding error due to $[n(r-a)]^{-1}-[\lfloor nr\rfloor-\lfloor na\rfloor]^{-1}$ and $[n(b-r)]^{-1}-[\lfloor nb\rfloor-\lfloor nr\rfloor]^{-1}$.  Note by Lemma \ref{lem_drate}, we know that $\sup_{(a,r)\in\mathcal{I}_{\eta}}d(\hat{\mu}_{[a, r]},\mu)=O_p(n^{-1/2})$ and $\sup_{(r,b)\in\mathcal{I}_{\eta}}d(\hat{\mu}_{[r, b]},\mu)=O_p(n^{-1/2})$, hence by Lemma \ref{lem_Vhat} and \ref{lem_VhatC}, we obtain $$\sup_{(r;a,b)\in \mathcal{J}_{\eta}}|\sqrt{n}T_n^C(r;a,b)|=o_p(1).$$
The result follows by continuous mapping theorem.

\subsection*{Proof of Theorem \ref{thm_cpt_Ha}}
Note for any $k=\lfloor n\eta_1\rfloor,\cdots,n-\lfloor n\eta_1\rfloor$, and $i=1,2$, 
$$\max_{\lfloor n\eta_1\rfloor\leq k\leq n-\lfloor n\eta_1\rfloor}D_{n,i}(k)\geq D_{n,i}(\lfloor n\tau\rfloor).$$

We  focus on $k^*=\lfloor n\tau\rfloor$. In this case, the left and right part of the self-normalizer are both from stationary segments, hence  by similar arguments as in $\mathbb{H}_0$,
\begin{align}\label{SN1}
\begin{split}
	& \{ \sqrt{n}T_{n}\left(r ; 0, \tau\right)\}_{r\in[\eta_2,\tau-\eta_2]}\Rightarrow \{\sigma_1\mathcal{G}_1(r; 0, \tau) \}_{r\in[\eta_2,\tau-\eta_2]},
	\\
	& \{ \sqrt{n}T^C_{n}\left(r ; 0, \tau\right)\}_{r\in[\eta_2,\tau-\eta_2]}\Rightarrow 0; 
\end{split}
\end{align}
and 
\begin{align}\label{SN2}
\begin{split}
	& \{ \sqrt{n}T_{n}\left(r ; \tau, 1\right)\}_{r\in[\tau+\eta_2,1-\eta_2]}\Rightarrow \{\sigma_2\mathcal{G}_2(r;\tau, 1) \}_{r\in[\tau+\eta_2,1-\eta_2]},\\
	& \{ \sqrt{n}T^C_{n}\left(r ; \tau, 1\right)\}_{r\in[\eta_2,\tau-\eta_2]}\Rightarrow 0,
\end{split}
\end{align}
where $\mathcal{G}_i(r;a,b)=\frac{(b-r)}{(b-a)}[B^{(i)}(r)-B^{(i)}(a)]-\frac{(r-a)}{(b-a)}[B^{(i)}(b)-B^{(i)}(r)]$ for $i=1,2.$

Hence, we only need to consider the numerator, where
\begin{align}\label{decomp}
\begin{split}
	& \sqrt{n}T_n(\tau;0,1)=\sqrt{n}{\tau(1-\tau)}\left(\hat{V}_{[0, \tau]}-\hat{V}_{[\tau, 1]}\right),\\
	&\sqrt{n}T_n^C(\tau;0,1)=\sqrt{n}{\tau(1-\tau)}\left(\hat{V}_{[\tau; 0, 1]}^{C}-\hat{V}_{[0,\tau]}-\hat{V}_{[\tau, 1]}\right).
	\end{split}
\end{align}
For $\sqrt{n}T_n(\tau;0,1)$, we have 
\begin{flalign*}
	\sqrt{n}T_n(\tau;0,1)=& \sqrt{n}\left\{{\tau(1-\tau)}\left(\hat{V}_{[0, \tau]}-\tilde{V}_{[0,\tau]}+\tilde{V}_{[0,\tau]}-V^{(1)}\right)\right\}\\&-\sqrt{n}\left\{{\tau(1-\tau)}\left(\hat{V}_{[\tau, 1]}-\tilde{V}_{[\tau, 1]}+\tilde{V}_{[\tau, 1]}-V^{(2)}\right)\right\}\\
	&+\sqrt{n}\tau(1-\tau)(V^{(1)}-V^{(2)})
	\\=&T_{11}+T_{12}+T_{13}.
\end{flalign*}
By Lemma \ref{lem_Vhat}, we know that $
\sqrt{n}(\hat{V}_{[0,\tau]}-\tilde{V}_{[0,\tau]})=o_p(1),
$
and by Assumption \ref{ass_FCLT}, we have $
\sqrt{n}\tau(\tilde{V}_{[0,\tau]}-V^{(1)})\to_d \sigma_1B^{(1)}(\tau).
$
This implies that $$T_{11}\to_d (1-\tau)\sigma_1B^{(1)}(\tau).$$ Similarly, we can obtain $$T_{12}\to_d -\tau\sigma_2[B^{(2)}(1)-B^{(2)}(\tau)].$$
Hence, using the fact that $\sqrt{n}(V^{(1)}-V^{(2)})=\Delta_V$,  we obtain
\begin{equation}\label{T1}
	\sqrt{n}T_n(\tau;0,1)\to_d(1-\tau)\sigma_1B^{(1)}(\tau)-\tau\sigma_2[B^{(2)}(1)-B^{(2)}(\tau)]+\tau(1-\tau)\Delta_V.
\end{equation}
For $\sqrt{n}T_n^C(\tau;0,1)$ we have 
\begin{flalign*}
	&\sqrt{n}T_n^C(\tau;0,1)\\=&(1-\tau)n^{-1/2}\Bigg\{ \sum_{i=1}^{\lfloor n \tau\rfloor} \left[d^{2}\left(Y_{i}, \hat{\mu}_{[\tau,1]}\right)-  d^{2}\left(Y_{i}, \mu^{(1)}\right)\right] \\
	& \hspace{3cm}- \sum_{i=1}^{\lfloor n \tau\rfloor}\left[ d^{2}\left(Y_{i}, \hat{\mu}_{[0,\tau]}\right)- d^{2}\left(Y_{i}, \mu^{(1)}\right)\right]\Bigg\}
	\\&+\tau n^{-1/2} \Bigg\{\sum_{i=\lfloor n \tau\rfloor+1}^{n} \left[d^{2}\left(Y_{i}, \hat{\mu}_{[0,\tau]}\right)-d^{2}\left(Y_{i}, \mu^{(2)}\right)\right] \\ 
	& \hspace{3cm} -\sum_{i=\lfloor n \tau\rfloor+1}^{n} \left[d^{2}\left(Y_{i}, \hat{\mu}_{[\tau,1]}\right)-d^{2}\left(Y_{i}, \mu^{(2)}\right)\right]\Bigg\}+o_p(1)
	\\:=& T_{21}+T_{22}+T_{23}+T_{24}+o_p(1),
\end{flalign*}
where $o_p(1)$ is the rounding error due to $(n\tau)^{-1}-\lfloor n\tau\rfloor^{-1}$ and $[n(1-\tau)]^{-1}-(n-\lfloor n\tau\rfloor)^{-1}$. 

Note by Lemma \ref{lem_drate}, we have $d(\hat{\mu}_{[0,\tau]},\mu^{(1)})=O_p(n^{-1/2}),$ and by  triangle inequality, we know that 
$$d(\hat{\mu}_{[\tau,1]},\mu^{(1)})\leq d(\hat{\mu}_{[\tau,1]},\mu^{(2)})+d(\mu^{(1)},\mu^{(2)})=O_p(n^{-1/4}).$$
Then, by Assumption \ref{ass_expand}, we know 
\begin{flalign*}
	T_{21}
	=&\sqrt{n}(1-\tau)\tau K_d d^2(\hat{\mu}_{[\tau,1]},\mu^{(1)})\\
	&+(1-\tau) d(\hat{\mu}_{[\tau,1]},\mu^{(1)})\left[\frac{n^{-1/2}\sum_{i=1}^{\lfloor n\tau\rfloor}g(Y_i,\hat{\mu}_{[\tau,1]},\mu^{(1)})}{d(\hat{\mu}_{[\tau,1]},\mu^{(1)})}\right]\\&+o_{p}(d(\hat{\mu}_{[\tau,1]},\mu^{(1)})+\sqrt{n}d^2(\hat{\mu}_{[\tau,1]},\mu^{(1)}))
	\\=&\sqrt{n}(1-\tau)\tau K_dd^2(\hat{\mu}_{[\tau,1]},\mu^{(1)})+O_p(n^{-1/4})+o_p(1).
\end{flalign*}
Now, by triangle inequality, we know 
\begin{align*}
   &\sqrt{n}[d(\hat{\mu}_{[\tau,1]},\mu^{(2)})-d(\mu^{(1)},\mu^{(2)})]^2\leq \sqrt{n}d^2(\hat{\mu}_{[\tau,1]},\mu^{(1)}) \\ 
   & \hspace{4cm}\leq \sqrt{n}[d(\hat{\mu}_{[\tau,1]},\mu^{(2)})+d(\mu^{(1)},\mu^{(2)})]^2, 
\end{align*}
and note $d(\hat{\mu}_{[\tau,1]},\mu^{(2)})=O_p(n^{-1/2})$ by Lemma \ref{lem_drate}, we obtain $\sqrt{n}d^2(\hat{\mu}_{[\tau,1]},\mu^{(1)})=\Delta_M+o_p(1)$, and 
$$T_{21}=(1-\tau)\tau K_d\Delta_M+o_p(1).$$
By Lemma \ref{lem_Vhat},  $T_{22}=o_p(1)$. Hence $T_{21}+T_{22}=(1-\tau)\tau K_d\Delta_M+o_p(1)$.
Similarly, we obtain that $T_{23}+T_{24}=(1-\tau)\tau K_d\Delta_M+o_p(1)$. Therefore, 
\begin{flalign}\label{T2}
	\sqrt{n}T_n^C(\tau;0,1)=2\tau(1-\tau)K_d\Delta_M+o_p(1).
\end{flalign}

Hence, combining results of \eqref{SN1}--\eqref{T2}, we have \begin{flalign*}
	D_{n,1}(\lfloor n\tau\rfloor)& \to_d\frac{[(1-\tau)\sigma_1B^{(1)}(\tau)-\tau\sigma_2[B^{(2)}(1)-B^{(2)}(\tau)]+\tau(1-\tau)\Delta_V]^2}{[\int_{\eta_2}^{r-\eta_2} \sigma_1^2\mathcal{G}_1^{2}(u ; 0, r) d u+\int_{r+\eta_2}^{1-\eta_2} \sigma_2^2\mathcal{G}_2^{2}(u ; r, 1) d u]} \\
	& := \mathcal{S}_{\eta,1}(\tau;\Delta),
\end{flalign*}
and,
\begin{align*}
    & D_{n,2}(\lfloor n\tau\rfloor) \\
    \to_d & \frac{[(1-\tau)\sigma_1B^{(1)}(\tau)-\tau\sigma_2[B^{(2)}(1)-B^{(2)}(\tau)]+\tau(1-\tau)\Delta_V]^2+4[\tau(1-\tau)\Delta_M]^2}{[\int_{\eta_2}^{r-\eta_2} \sigma_1^2\mathcal{G}_1^{2}(u ; 0, r) d u+\int_{r+\eta_2}^{1-\eta_2} \sigma_2^2\mathcal{G}_2^{2}(u ; r, 1) d u]}\\
    := & \mathcal{S}_{\eta,2}(\tau;\Delta).
\end{align*}

Therefore, we know that for the $1-\alpha$ quantile of $\mathcal{S}_{\eta}$, denoted by $Q_{1-\alpha}(\mathcal{S}_{\eta})$, for $i=1,2$,
\begin{align*}
    & \lim_{n\to\infty} P\left(\max_{\lfloor n\eta_1\rfloor\leq k\leq n-\lfloor n\eta_1\rfloor}D_{n,i}(k)\geq Q_{1-\alpha}(\mathcal{S}_{\eta})\right) \\
    \geq & \lim_{n\to\infty} P\left(D_{n,i}(\lfloor n\tau\rfloor)\geq Q_{1-\alpha}(\mathcal{S}_{\eta})\right) \\ 
    = & P\left(\mathcal{S}_{\eta,i}(\tau;\Delta)\geq Q_{1-\alpha}(\mathcal{S}_{\eta})\right).
\end{align*}
In particular, 
\begin{flalign*}
	& \lim_{|\Delta_V|\to\infty}P\Big(\mathcal{S}_{\eta,1}(\tau;\Delta)\geq Q_{1-\alpha}(\mathcal{S}_{\eta})\Big)=1,\\
	& \lim_{\max\{|\Delta_V|,\Delta_M\}\to\infty}P\Big(\mathcal{S}_{\eta,2}(\tau;\Delta)\geq Q_{1-\alpha}(\mathcal{S}_{\eta})\Big)=1.
\end{flalign*}

\subsection*{Proof of Theorem \ref{thm_cpt_con}}
Define the pointwise limit of $\hat{\mu}_{[a,b]}$ under $\mathbb{H}_a$ as 
$$
\mu_{[a,b]}=
\begin{cases}
	\mu^{(1)}, & b\leq \tau\\
	\arg\min_{\omega\in\Omega}\left\{(\tau-a)\mathbb{E}d^2(Y_t^{(1)},\omega)+(b-\tau)\mathbb{E}d^2(Y_t^{(2)},\omega)\right\}, &a<\tau<b\\
	\mu^{(2)},& \tau\leq a
\end{cases}
$$
Define the Fr\'echet variance and pooled contaminated variance under $\mathbb{H}_a$ as 
$$
V_{[a,b]}=\begin{cases}
	V^{(1)}& b\leq \tau\\
	\frac{\tau-a}{b-a}\mathbb{E}(d^2(Y_t^{(1)},\mu_{[a,b]}))+\frac{b-\tau}{b-a}\mathbb{E}(d^2(Y_t^{(2)},\mu_{[a,b]})), &a<\tau<b\\
	V^{(2)},& \tau\leq a,
\end{cases}
$$
and 
\begin{flalign*}
&V^C_{[r;a,b]}=\\&\begin{cases}
	V^{(1)}& b\leq \tau\\
	\frac{\tau-a}{r-a}\mathbb{E}(d^2(Y_t^{(1)},\mu_{[r,b]}))+\frac{r-\tau}{r-a}\mathbb{E}(d^2(Y_t^{(2)},\mu_{[r,b]}))+\mathbb{E}(d^2(Y_t^{(2)},\mu_{[a,r]})), &a<\tau\leq r\\
	\mathbb{E}(d^2(Y_t^{(1)},\mu_{[r,b]}))+\frac{\tau-r}{b-r}\mathbb{E}(d^2(Y_t^{(1)},\mu_{[a,r]}))+\frac{b-\tau}{b-r}\mathbb{E}(d^2(Y_t^{(2)},\mu_{[a,r]})), &r<\tau<b\\
	V^{(2)},& \tau\leq a.
\end{cases}    
\end{flalign*}
We want to show that 
\begin{flalign*}
	& \left\{T_n(r;a,b)\right\}_{(r;a,b)\in\mathcal{J}_{\eta}}
	\Rightarrow \left\{T(r;a,b)\right\}_{(r;a,b)\in\mathcal{J}_{\eta}},\\
	& \left\{T^C_n(r;a,b)\right\}_{(r;a,b)\in\mathcal{J}_{\eta}}
	\Rightarrow \left\{T^C(r;a,b)\right\}_{(r;a,b)\in\mathcal{J}_{\eta}},
\end{flalign*}
where 
\begin{align*}
&T(r;a,b)=\frac{(r-a)(b-r)}{b-a}\left({V}_{[a, r]}-{V}_{[r, b]}\right), \\ 
&T^C(r;a,b)=\frac{(r-a)(b-r)}{b-a}\left({V}_{[r ; a, b]}^{C}-{V}_{[a, r]}-{V}_{[r, b]}\right).
\end{align*}

To achieve this, we  need to show (1). $\sup_{(a,b)\in\mathcal{I}_{\eta}}d(\hat{\mu}_{[a,b]},\mu_{[a,b]})=o_p(1)$; (2). $\sup_{(a,b)\in\mathcal{I}_{\eta}}|\hat{V}_{[a,b]}-V_{[a,b]}|=o_p(1)$; and (3). $\sup_{(r;a,b)\in\mathcal{J}_{\eta}}|\hat{V}^C_{[r;a,b]}-V^C_{[r;a,b]}|=o_p(1)$.

(1). The cases when $b\leq \tau$ and $a\geq \tau$ follow by Lemma \ref{lem_drate}. 
For the case when $\tau \in(a,b)$, recall
\begin{flalign*}
	\hat{\mu}_{[a, b]}=&\arg \min _{\omega \in \Omega}\frac{1}{\lfloor nb\rfloor-\lfloor na\rfloor} \sum_{t=\lfloor n a\rfloor+1}^{\lfloor n b\rfloor} d^{2}\left(Y_{t}, \omega\right)
	\\=&\arg \min _{\omega \in \Omega}\Bigg\{\frac{n}{\lfloor nb\rfloor-\lfloor na\rfloor}\frac{1}{n} \sum_{t=\lfloor n a\rfloor+1}^{\lfloor n \tau \rfloor} d^{2}\left(Y_{t}^{(1)}, \omega\right) \\ & \hspace{2cm} +\frac{n}{\lfloor nb\rfloor-\lfloor na\rfloor} \frac{1}{n} \sum_{t=\lfloor n \tau\rfloor+1}^{\lfloor n b \rfloor} d^{2}\left(Y_{t}^{(2)}, \omega\right)\Bigg\}.
\end{flalign*}
By the proof of (1) in Lemma \ref{lem_drate}, for $i=1,2$,  we have $$
\left\{\frac{1}{n} \sum_{t=1}^{\lfloor n u \rfloor} d^{2}\left(Y_{t}^{(i)}, \omega\right)-u\mathbb{E}d^2\left(Y_t^{(i)},\omega\right)\right\}_{\omega\in\Omega,u\in[0,1]}\Rightarrow 0,
$$
which implies that 
\begin{flalign}\label{mapping}
	\begin{split}
		&\Bigg\{\frac{n}{\lfloor nb\rfloor-\lfloor na\rfloor}\frac{1}{n} \sum_{t=\lfloor n a\rfloor+1}^{\lfloor n \tau \rfloor} d^{2}\left(Y_{t}^{(1)}, \omega\right) \\ & \hspace{1cm} +\frac{n}{\lfloor nb\rfloor-\lfloor na\rfloor} \frac{1}{n} \sum_{t=\lfloor n \tau\rfloor+1}^{\lfloor n b \rfloor} d^{2}\left(Y_{t}^{(2)}, \omega\right)\Bigg\}_{\omega\in\Omega,(a,b)\in\mathcal{I}_{\eta}}\\
		& \hspace{2cm} \Rightarrow \left\{\frac{\tau-a}{b-a}\mathbb{E}(d^2(Y_t^{(1)},\omega)+\frac{b-\tau}{b-a}\mathbb{E}(d^2(Y_t^{(2)},\omega))\right\}_{\omega\in\Omega,(a,b)\in\mathcal{I}_{\eta}}.
	\end{split}
\end{flalign}
By Assumption \ref{ass_unique}, and the argmax continuous mapping theorem (Theorem 3.2.2 in \cite{vaart1996weak}), the result follows.

(2).  The cases when $b\leq \tau$ and $a\geq \tau$ follows by Lemma \ref{lem_Vhat}.   For the case when $\tau \in(a,b)$, we have  for some constant $K>0$
\begin{flalign*}
	& \sup_{(a,b)\in\mathcal{I}_{\eta}}|\hat{V}_{[a,b]}-V_{[a,b]}| \\
	\leq& \sup_{(a,b)\in\mathcal{I}_{\eta}}\left(\frac{1}{\lfloor nb\rfloor-\lfloor na\rfloor} \sum_{t=\lfloor n a\rfloor+1}^{\lfloor n b\rfloor} \left|d^{2}\left(Y_{t}, \hat{\mu}_{[a,b]}\right)-d^{2}\left(Y_{t}, \mu_{[a,b]}\right)\right|\right)\\&+\sup_{(a,b)\in\mathcal{I}_{\eta}}\left|\frac{1}{\lfloor nb\rfloor-\lfloor na\rfloor} \sum_{t=\lfloor n a\rfloor+1}^{\lfloor n b\rfloor} d^{2}\left(Y_{t}, \mu_{[a,b]}\right)-V_{[a,b]}\right|
	\\\leq& \sup_{(a,b)\in\mathcal{I}_{\eta}}\left(\frac{1}{\lfloor nb\rfloor-\lfloor na\rfloor} \sum_{t=\lfloor n a\rfloor+1}^{\lfloor n b\rfloor} K\left|d\left(Y_{t}, \hat{\mu}_{[a,b]}\right)-d\left(Y_{t}, \mu_{[a,b]}\right)\right|\right)+o_p(1)
	\\\leq& \sup_{(a,b)\in\mathcal{I}_{\eta}}Kd(\hat{\mu}_{[a,b]},\mu_{[a,b]})+o_p(1)=o_p(1)
\end{flalign*}
where the second inequality holds by the boundedness of the metric and \eqref{mapping}, and the third inequality holds by the triangle inequality of the metric.

(3). The proof is similar to (2).


By continuous mapping theorem, we obtain that for $i=1,2$,
$$\left\{D_{n,i}(\lfloor nr\rfloor)\right\}_{r\in[\eta_1,1-\eta_1]}\Rightarrow\left\{D_{i}(r)\right\}_{r\in[\eta_1,1-\eta_1]},$$
where 
\begin{flalign*}
	D_1(r)=& \frac{[T(r;0,1)]^2}{\int_{\eta_2}^{r-\eta_2}[T(u;0,r)]^2du+\int_{r+\eta_2}^{1-\eta_2}[T(u;r,1)]^2du},\\ D_2(r)=& \frac{[T(r;0,1)]^2+[T^C(r;0,1)]^2}{\int_{\eta_2}^{r-\eta_2}[T(u;0,r)]^2+[T^C(u;0,r)]^2du+\int_{r+\eta_2}^{1-\eta_2}[T(u;r,1)]^2+[T^C(u;r,1)]^2du}.
\end{flalign*}

In particular,  at $r=\tau$,  we  obtain $D_{i}(\tau)=\infty$. Hence, to show the consistency of $\hat{\tau}$, it suffices to show that for any small $\epsilon>0$,  if $|r-\tau|>\epsilon$,
$$
D_i(r)<\infty.
$$
By symmetry, we consider the case of  $r-\tau>\epsilon$.

For $r-\tau>\epsilon$, we note that for both $i=1,2,$
$$
\sup_{r-\tau>\epsilon}D_{i}(r)\leq\frac{\sup_{r}\left\{[T(r;0,1)]^2+[T^C(r;0,1)]^2\right\}}{\inf_{r-\tau>\epsilon}\int_{\eta_2}^{r-\eta_2}[T(u;0,r)]^2du}.
$$
By proof of Proposition 1 in \cite{dubey2020frechet}, we obtain that for some universal constant $K>0$,
\begin{flalign*}
	\sup_{r}\left\{[T(r;0,1)]^2+[T^C(r;0,1)]^2\right\}\leq K(\Delta^2_M+\Delta^2_V)<\infty.
\end{flalign*}
Therefore, it suffices to show that there exists a function $\zeta(\epsilon)>0$, such that for any $r-\tau>\epsilon$,
$$
\int_{\eta_2}^{\tau-\eta_2}[T(u;0,r)]^2du>\zeta(\epsilon).
$$
For $r>\tau$, and for any  $u\in[\eta_2,\tau-\eta_2]$, 
\begin{flalign*}
	& T(u;0,r) \\
	=&\frac{u(r-u)}{r}(V^{(1)}-V_{[u,r]})\\=&\frac{u(r-u)}{r}\left[V^{(1)}-\frac{\tau-u}{r-u}\mathbb{E}(d^2(Y_t^{(1)},\mu_{[u,r]}))-\frac{r-\tau}{r-u}\mathbb{E}(d^2(Y_t^{(2)},\mu_{[u,r]}))\right]
	\\=&\frac{u(r-u)}{r}[V^{(1)}-V(\frac{\tau-u}{r-u})].
\end{flalign*}
By Assumption \ref{ass_mix}, we can obtain that 
$$
|T(u;0,r)|>\frac{u(r-u)}{r}\varphi(\frac{\epsilon}{r-u})\geq \eta_2^2\varphi(\epsilon).
$$
Hence, we can choose $\zeta(\epsilon)=\eta_2^6\varphi^2(\epsilon)$.

\section{Examples}\label{sec:example}
As we have mentioned in the main context, since $d^2(Y_t,\omega)$ takes value in $\mathbb{R}$ for any fixed $\omega\in\Omega$,  both Assumption \ref{ass_LLN} and \ref{ass_FCLTo}  could be implied by  high-level weak temporal dependence conditions in conventional Euclidean space. Therefore, we only discuss the verification of Assumption \ref{ass_diff}, \ref{ass_unique} and \ref{ass_expand} in what follows.

\subsection{Example 1: $L_2$ metric  $d_L$ for square integrable functions defined on $[0,1]$}
Let $\Omega$ be the Hilbert space of all square integrable functions defined on $I=[0,1]$ with inner product $\langle f,g\rangle=\int_{I}f(t)g(t)dt$ for two functions $f,g\in\Omega$. Then, for the corresponding norm $\|f\|=\langle f,f\rangle^{1/2}$, $L_2$ metric is defined by $$d_L^2(f,g)=\int_{I}[f(t)-g(t)]^2dt.$$

Assumptions \ref{ass_diff} and \ref{ass_unique} follows easily by the Riesz representation theorem and convexity of $\Omega$. We only consider Assumption \ref{ass_expand}.

Note that 
\begin{flalign*}
	d_L^2(Y,\omega)-d_L^2(Y,\mu)=&\int_{0}^1 [{\omega}(t)-{\mu}(t)][{\omega}(t)+{\mu}(t)-2Y(t)]dt\\
	=&d_L^2({\omega},{\mu})+2\int_{0}^1 [{\omega}(t)-{\mu}(t)][{\mu}(t)-Y(t)]dt
	\\:=&d_L^2({\omega},{\mu})+g(Y,\omega,\mu),
\end{flalign*}
and $R(Y,\omega,\mu)\equiv 0$.
Furthermore, 
\begin{flalign*}
	& \left|n^{-1/2}\sum_{i=\lfloor n a\rfloor+1}^{\lfloor n b\rfloor}g(Y_i,\omega,\mu)\right| \\
	=&\left|2\int_{0}^1 [{\omega}(t)-{\mu}(t)]n^{-1/2}\sum_{i=\lfloor n a\rfloor+1}^{\lfloor n b\rfloor}[{Y_i}(t)-{\mu}(t)]dt\right|
	\\\leq &2d_L({\omega},{\mu}) \left\{\int_{0}^1 \left|n^{-1/2}\sum_{i=\lfloor n a\rfloor+1}^{\lfloor n b\rfloor}[{Y_i}(t)-{\mu}(t)]\right|^2dt\right\}^{1/2},
\end{flalign*}
where the inequality holds by Cauchy-Schwarz inequality. 

By the boundedness of $d_L({\omega},{\mu})$, Assumption \ref{ass_expand} then follows if
$$
\sup_{t\in[0,1]}\sup_{(a,b)\in\mathcal{I}_{\eta}}\left|n^{-1/2}\sum_{i=\lfloor n a\rfloor+1}^{\lfloor n b\rfloor}[{Y_i}(t)-{\mu}(t)]\right|=O_p(1),
$$
which holds under general weak temporal dependence for functional observations, see, e.g. \cite{berkes2013weak}.

\subsection{Example 2: 2-Wasserstein metric $d_W$ of univariate CDFs}
Let $\Omega$ be the set of univariate CDF function on $\mathbb{R}$, consider the 2-Wasserstein metric defined by 
$$
d_W^2(G_1,G_2)=\int_{0}^1 (G_1(t)-G_2(t))^2dt,
$$
where $G_1$ and $G_2$ are two inverse CDFs  or quantile functions.

The verification of Assumption \ref{ass_diff} and \ref{ass_unique} can be found in Proposition C.1 in \cite{dubey2020frechet}. Furthermore, by similar arguments as Example 1, Assumption \ref{ass_expand} holds under weak temporal dependence conditions, see \cite{berkes2013weak}.

\subsection{Example 3: Frobenius metric $d_F$ for graph Laplacians or covariance matrices}
Let $\Omega$ 
be the set of graph Laplacians or covariance matrices of a fixed dimension $r$, with uniformly bounded diagonals, and equipped with the Frobenius metric $d_F$, i.e.
$$
d_F^2(\Sigma_1,\Sigma_2)=\mathrm{tr}[(\Sigma_1-\Sigma_2)^{\top}(\Sigma_1-\Sigma_2)].
$$
for two $r\times r$ matrices $\Sigma_1$ and $\Sigma_2$.

The verification of Assumption \ref{ass_diff} and \ref{ass_unique} can be found in Proposition C.2 in \cite{dubey2020frechet}. We only consider Assumption \ref{ass_expand}.

Note that 
\begin{flalign*}
	d_F^2(Y,\omega)-d_F^2(Y,\mu)=&\mathrm{tr}({\omega}-{\mu})^{\top}({\omega}+{\mu}-2Y)
	\\=&d_F^2({\omega},{\mu})+2\mathrm{tr}({\omega}-{\mu})^{\top}({\mu}-Y)\\
	:=&d_F^2({\omega},{\mu})+g(Y,\omega,\mu),
\end{flalign*}
and $R(Y,\omega,\mu)\equiv 0$.
Furthermore, by Cauchy-Schwarz inequality,
\begin{flalign*}
	\left|n^{-1/2}\sum_{i=\lfloor n a\rfloor+1}^{\lfloor n b\rfloor}g(Y_i,\omega,\mu)\right|
	=&2\left|\mathrm{tr}[({\omega}-{\mu})^{\top}n^{-1/2}\sum_{i=\lfloor n a\rfloor+1}^{\lfloor n b\rfloor}({Y_i}-{\mu})]\right|
	\\\leq &2d_F({\omega},{\mu}) d_F\left(n^{-1/2}\sum_{i=\lfloor n a\rfloor+1}^{\lfloor n b\rfloor}[{Y_i}-{\mu}],0\right).
\end{flalign*}
By the boundedness of $d_F({\omega},{\mu})$, Assumption \ref{ass_expand} then follows if
$$
\sup_{(a,b)\in\mathcal{I}_{\eta}}\left\|n^{-1/2}\sum_{i=\lfloor n a\rfloor+1}^{\lfloor n b\rfloor}\mathrm{vec}(Y_i-\mu)\right\|=O_p(1),
$$
which holds under common weak dependence conditions in conventional Euclidean space.

\subsection{Example 4: Log-Euclidean metric $d_E$ for covariance matrices }
Let $\Omega$ be the set of all positive-definite covariance matrices of dimension $r$, with uniformly both upper and lower bounded eigenvalues, i.e. for any $\Sigma\in\Omega$, $c\leq \lambda_{min}(\Sigma)\leq \lambda_{\max}(\Sigma)\leq C$ for some constant $0<c<C<\infty$. The log-Euclidean metric is defined by $d_E^2(\Sigma_1,\Sigma_2)=d_F^2(\log_m\Sigma_1,\log_m\Sigma_2)$, where $\log_m$ is the matrix-log function. 

Note that $\log_m\Sigma$ has the same dimension as $\Sigma$, hence the verification of Assumptions \ref{ass_diff}, \ref{ass_unique} and \ref{ass_expand} follows directly from Example 3.

\section{Functional Data in Hilbert Space}\label{sec:simu_functional}
Our proposed tests and DM test are also applicable to the inference of functional data in Hilbert space, such as $L_2[0,1]$, since the norm in Hilbert space naturally corresponds to the distance metric $d$. In a sense, our methods can be regarded as fully functional \citep{aue2018detecting} since no dimension reduction procedure is required. In this section, we further compare them with SN-based testing procedure by   \cite{zhangshao2015} for comparing two sequences of temporally dependent functional data, i.e.   $\{Y_t^{(i)}\}_{t=1}^{n_i}$ $i=1,2,$ defined on $[0,1]$.  The general idea is to first apply FPCA, and then compare  score functions (for mean) or covariance operators (for covariance)  between two samples in the space spanned by leading $K$ eigenfunctions. SN technique is also invoked to account for unknown temporal dependence. 

Although the test statistic in \cite{zhangshao2015} targets at the difference in covariance  operators of $\{Y_t^{(1)}\}$ and $\{Y_t^{(2)}\},$  their test can be readily modified to testing the mean difference.    
To be specific, denote $\mu^{(i)}$ as the mean function of $Y_t^{(i)}$, $t=1,\cdots,n_i$, $i=1,2$, we are interested in testing  
$$
\mathbb{H}_0: \mu^{(1)}(x)=\mu^{(2)}(x),\quad \forall x\in[0,1].
$$
We assume the covariance operator is common for both samples, which is denoted by $C_p$.

By Mercer’s Lemma, we have 
$$
C_p=\sum_{j=1}^{\infty}\lambda_p^j\phi_p^j\otimes \phi_p^j,
$$
where $\{{\lambda}^j_p\}_{j=1}^{\infty}$ and $\{{\phi}^j_p\}_{j=1}^{\infty}$ are the eigenvalues and eigenfunctions respectively.

By the Karhunen-Lo\`eve expansion, 
$$
Y_t^{(i)}=\mu^{(i)}+\sum_{j=1}^{\infty}\eta_{t,j}^{(i)}\phi^j_p, \quad t=1,\cdots,n_i;~~i=1,2,
$$
where $\{\eta_{t,j}^{(i)}\}$ are
the principal components (scores) defined by $\eta_{t,j}^{(i)}=\int_{[0,1]}\{Y_t^{(i)}-\mu^{(i)}\}\phi^j_p(x)dx=\int_{[0,1]}\{Y_t^{(i)}-\mu_p+\mu_p-\mu^{(i)}\}\phi^j_p(x)dx$ with $\mu_p=\gamma_1\mu^{(1)}+\gamma_2\mu^{(2)}$.

Under $\mathbb{H}_0$, $\mu^{(1)}=\mu^{(2)}=\mu_p$, and $\eta_{t,j}^{(i)}$ should have mean zero. We thus build the SN based test by comparing empirical estimates of score functions. Specifically, define the empirical covariance operator based on the pooled samples as 
$$
\hat{C}_{p}= \frac{1}{n_1+n_2}(\sum_{t=1}^{n_1}\mathcal{Y}^{(1)}_t+\sum_{t=1}^{n_2}\mathcal{Y}^{(2)}_t),
$$
where $\mathcal{Y}^{(i)}_t=  Y_t^{(i)}\otimes Y_t^{(i)},$ $i=1,2.$  Denote by $\{\hat{\lambda}^j_p\}_{j=1}^{\infty}$ and $\{\hat{\phi}^j_p\}_{j=1}^{\infty}$ the corresponding eigenvalues and eigenfunctions. We define the empirical scores (projected onto the eigenfunctions of pooled covariance operator) for each functional observation as 
$$
\hat{\eta}^{(i)}_{t,j}=\int_{[0,1]}\{Y_t^{(i)}(x)-\hat{\mu}_p(x)\}\hat{\phi}^j_p(x)dx, \quad t=1,\cdots,n_i;~~i=1,2; ~~j=1,\cdots, K,
$$
where $\hat{\mu}_p=(\sum_{t=1}^{n_1}Y_t^{(1)}+\sum_{t=1}^{n_2}Y_t^{(2)})/n$ is the pooled sample mean function.

Let $\hat{\eta}^{(i,K)}_{t,(K)}=(\hat{\eta}^{(i)}_{t,1},\cdots,\hat{\eta}^{(i)}_{t,K})^{\top}$, and $\hat{\alpha}^{(K)}(r)=(\lfloor rn_1\rfloor)^{-1}\sum_{t=1}^{\lfloor rn_1\rfloor}\hat{\eta}^{(1,K)}_{t}-(\lfloor rn_2\rfloor)^{-1}\sum_{t=1}^{\lfloor rn_2\rfloor}\hat{\eta}^{(2,K)}_{t}$ as the difference of recursive subsample mean of empirical scores, we consider the test statistic as 
\begin{align*}
 &ZSM=\\&n[\hat{\alpha}^{(K)}(1)]^{\top} \Bigg\{\sum_{k=1}^n\frac{k^2}{n^2}[\hat{\alpha}^{(K)}(k/n) -\hat{\alpha}^{(K)}(1)][\hat{\alpha}^{(K)}(k/n)-\hat{\alpha}^{(K)}(1)]^{\top}\Bigg\}^{-1}[\hat{\alpha}^{(K)}(1)],   
\end{align*}
and under $\mathbb{H}_0$ with suitable conditions, it is expected that 
$$
ZSM\to_d
B_{K}(1)^{\top} \left\{\int_{0}^{1}\left(B_{K}(r)-r B_{K}(1)\right)\left(B_{K}(r)-r B_{K}(1)\right)^{\top} \mathrm{d} r\right\}^{-1} B_{K}(1),
$$
where $B_K(\cdot)$ is a $K$-dimensional vector of independent Brownian motions.

Consider the following model taken from \cite{panaretos2010second},
\begin{flalign*}
	Y_t(x)=&\sum_{j=1}^{3}\left\{\xi^{j, 1}_{t} \sqrt{2} \sin (2 \pi j x)+\xi^{j, 2}_{t} \sqrt{2} \cos (2 \pi j x)\right\}, \quad t=1,2, \ldots,n_1
\end{flalign*}
where the coefficients $\xi_{t}=\left(\xi^{1,1}_{t}, \xi^{2,1}_{t}, \xi^{3,1}_{t}, \xi^{1,2}_{t}, \xi^{2,2}_{t}, \xi^{3,2}_{t}\right)^{\prime}$ are generated from a VAR process,
\begin{flalign*}
	\xi_{t}=&\rho \xi_{t-1}+\sqrt{1-\rho^{2}} e_{t}, \quad e_{t} \overset{i.i.d.}{\sim} \mathcal{N}\left(0,\frac{1}{2} \operatorname{diag}(\mathbf{v})+\frac{1}{2} \mathbf{1}_{6}\right)\in\mathbb{R}^6
\end{flalign*}
with $v=(12, 7, 0.5, 9, 5, 0.3)^{\top}$.

To compare the size and power performance, we generate independent functional time series $\{Y_t^{(1)}\}$ and $\{Y_t^{(2)}\}$ from the above model, and modify $\{Y_t^{(2)}\}$ according to the following settings: 
\begin{itemize}
	\item \mbox{[Case 1m]} $Y_t^{(2)}(x)= Y_t(x)+20\delta_1\sin(2\pi x)$, $x\in[0,1]$;   \item \mbox{[Case 1v]} $Y_t^{(2)}(x)= Y_t(x)+20\delta_2\eta_t\sin(2\pi x)$, $x\in[0,1]$;
	\item  \mbox{[Case 2m]} $Y_t^{(2)}(x)= Y_t(x)+20\delta_1x$, $x\in[0,1]$;
	\item  \mbox{[Case 2v]} $Y_t^{(2)}(x)= Y_t(x)+20\delta_2\eta_tx$, $x\in[0,1]$;
	\item  \mbox{[Case 3m]} $Y_t^{(2)}(x)= Y_t(x)+20\delta_1\mathbf{1}(x\in[0,1]) $;
	\item \mbox{[Case 3v]} $Y_t^{(2)}(x)= Y_t(x)+20\delta_2\eta_t\mathbf{1}(x\in[0,1])$;
\end{itemize}
where $\eta_t\overset{i.i.d.}{\sim}\mathcal{N}(0,1)$ and $\delta_1,\delta_2\in[0,0.3]$.  

The size performance of all tests are evaluated by setting $\delta_1=\delta_2=0$.
As for the power performance,  Cases 1m-3m with $\delta_1\in(0,0.3]$ correspond to alternatives caused by mean differences and Cases 1v-3v  with $\delta_2\in(0,0.3]$ correspond to covariance operator differences. In particular, we note the alternative of Cases 1m and 1v depends on the signal function $f(x)=\sin(2\pi x), x\in[0,1]$, which  is in the space spanned by the eigenfunctions of $Y_t(x)$, while  for Cases 3m and 3v, the  signal function $f(x)=\mathbf{1}(x\in[0,1])$ is orthogonal to these eigenfunctions.

We denote the two-sample mean test and covariance operator test based on \cite{zhangshao2015} as ZSM and ZSV respectively. The empirical size of all tests are outlined in Table \ref{tab_functional} at nominal level $\alpha=5\%$.  From this table, we see that (a) $D_1$ has accurate size  across all model settings and $D_2$ is generally reliable for moderate dependence level, albeit oversize phenomenon for small $n$ when $\rho=0.7$; (b) DM suffers from severe size distortion when temporal dependence is exhibited even for large $n$; (c) although both ZSM and ZSV utilize SN to robustify the tests due to  temporal dependence, we find their performances depend on the user-chosen parameter $K$ a lot, and still suffer from size distortion when $n$ is small. In particular, the size distortion when $K=4$ is considerably larger than that for $K=2$ in the presence of temporal dependence.

\begin{table}[H]
	\centering 
	\caption{Size Performance $(100\%)$ at $\alpha=5\%$} 
	\label{tab_functional}
	\begin{tabular}{ccccccccc}
		\hline
		\multicolumn{9}{c}{Functional Data based on $d_L$}                                                                                                                           \\ \hline
		\multirow{2}{*}{$\rho$} & \multirow{2}{*}{$n_i$} & \multirow{2}{*}{$D_1$} & \multirow{2}{*}{$D_2$} & \multirow{2}{*}{DM} & \multicolumn{2}{c}{ZSM} & \multicolumn{2}{c}{ZSV} \\ \cline{6-9} 
		&                        &                        &                        &                     & $K=2$      & $K=4$      & $K=2$      & $K=4$      \\ \hline
		\multirow{4}{*}{-0.4}   & 50                     & 5.4                    & 5.6                    & 11.5                & 3.7        & 2.3        & 5.8        & 10.1       \\
		& 100                    & 5.3                    & 5.3                    & 9.5                 & 3.1        & 2.5        & 4.4        & 7.6        \\
		& 200                    & 6.7                    & 6.6                    & 11.5                & 3.3        & 4.2        & 5.8        & 5.7        \\
		& 400                    & 5.6                    & 5.6                    & 8.7                 & 4.4        & 4.2        & 4.2        & 7.3        \\ \hline
		\multirow{4}{*}{0}      & 50                     & 4.9                    & 5.6                    & 5.7                 & 6.3        & 6.3        & 5.3        & 5.1        \\
		& 100                    & 3.8                    & 3.8                    & 4.3                 & 5.0        & 4.4        & 4.0        & 5.0        \\
		& 200                    & 5.8                    & 6.0                    & 5.5                 & 3.8        & 5.7        & 5.4        & 5.7        \\
		& 400                    & 4.3                    & 4.6                    & 4.1                 & 5.4        & 4.7        & 4.5        & 6.1        \\ \hline
		\multirow{4}{*}{0.4}    & 50                     & 5.9                    & 8.9                    & 10.6                & 8.3        & 13.6       & 5.3        & 10.9       \\
		& 100                    & 4.9                    & 4.7                    & 9.5                 & 6.7        & 8.4        & 5.4        & 7.1        \\
		& 200                    & 5.5                    & 5.8                    & 8.9                 & 4.7        & 7.4        & 5.8        & 6.9        \\
		& 400                    & 5.3                    & 4.9                    & 9.6                 & 5.9        & 5.8        & 6.0        & 5.3        \\ \hline
		\multirow{4}{*}{0.7}    & 50                     & 7.2                    & 20.4                   & 36.9                & 17.1       & 31.4       & 7.3        & 34.8       \\
		& 100                    & 6.5                    & 12.1                   & 29.5                & 10.1       & 16.4       & 5.7        & 18.9       \\
		& 200                    & 6.5                    & 8.2                    & 29.6                & 6.8        & 11.7       & 5.9        & 10.3       \\
		& 400                    & 4.9                    & 5.8                    & 25.0                & 7.0        & 8.4        & 6.1        & 6.8        \\ \hline
	\end{tabular}
\end{table}

Figure \ref{Fig:functional} further compares their size-adjusted powers when $n_1=n_2=400$ and $\rho=0.4$. As can be seen, $D_1$ possesses trivial power against mean differences while $D_2$ is rather stable in all settings with evident advantages in Cases 2m and 3m.    In contrast, the power performances of DM, ZSM and ZSV vary among different settings. For example, when the alternative signal function is in the span of leading eigenfunctions, i.e. Cases 1m and 1v, ZSM and ZSV with $K=2$ can deliver (second) best power performances as expected, while they are dominated by other tests when the alternative  signal function is orthogonal to eigenfunctions in Cases 3m and 3v.   As for DM, it is largely dominated by $D_2$ in terms of mean differences, although it exhibits moderate advantage over $D_2$ for covariance operator differences. 

In general, whether the difference in mean/covariance operator is orthogonal to the leading eigenfunctions, or lack thereof, is unknown to the user.  
Our test $D_2$ is  robust to unknown temporal dependence, exhibits quite accurate  size and delivers comparable powers in all settings, and thus should be preferred in practice.

\begin{figure}[H]
	\centering 
	\begin{subfigure}{0.32\textwidth}
		\centering
		\includegraphics[width=1\textwidth]{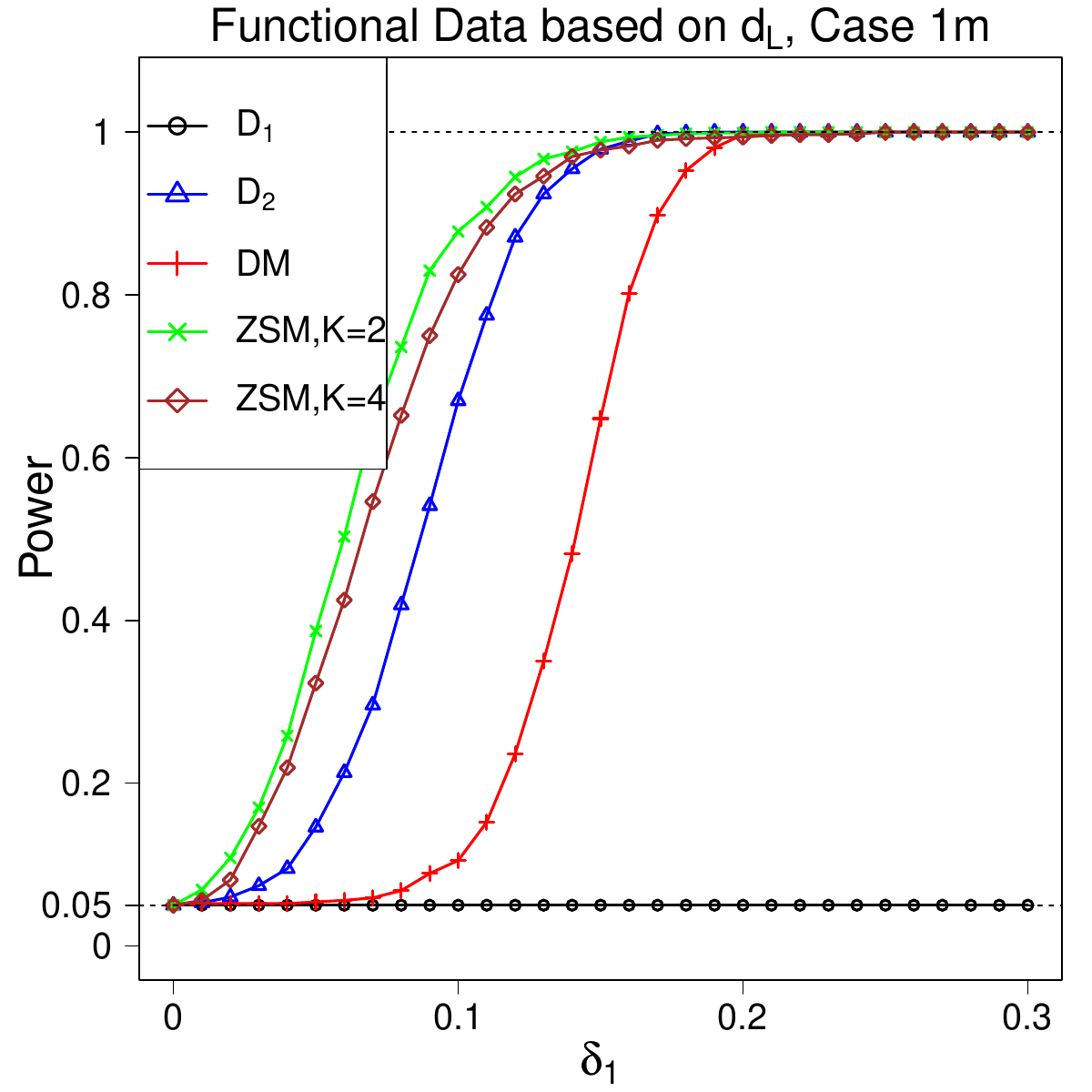}
	\end{subfigure}
	\begin{subfigure}{0.32\textwidth}
		\centering
		\includegraphics[width=1\textwidth]{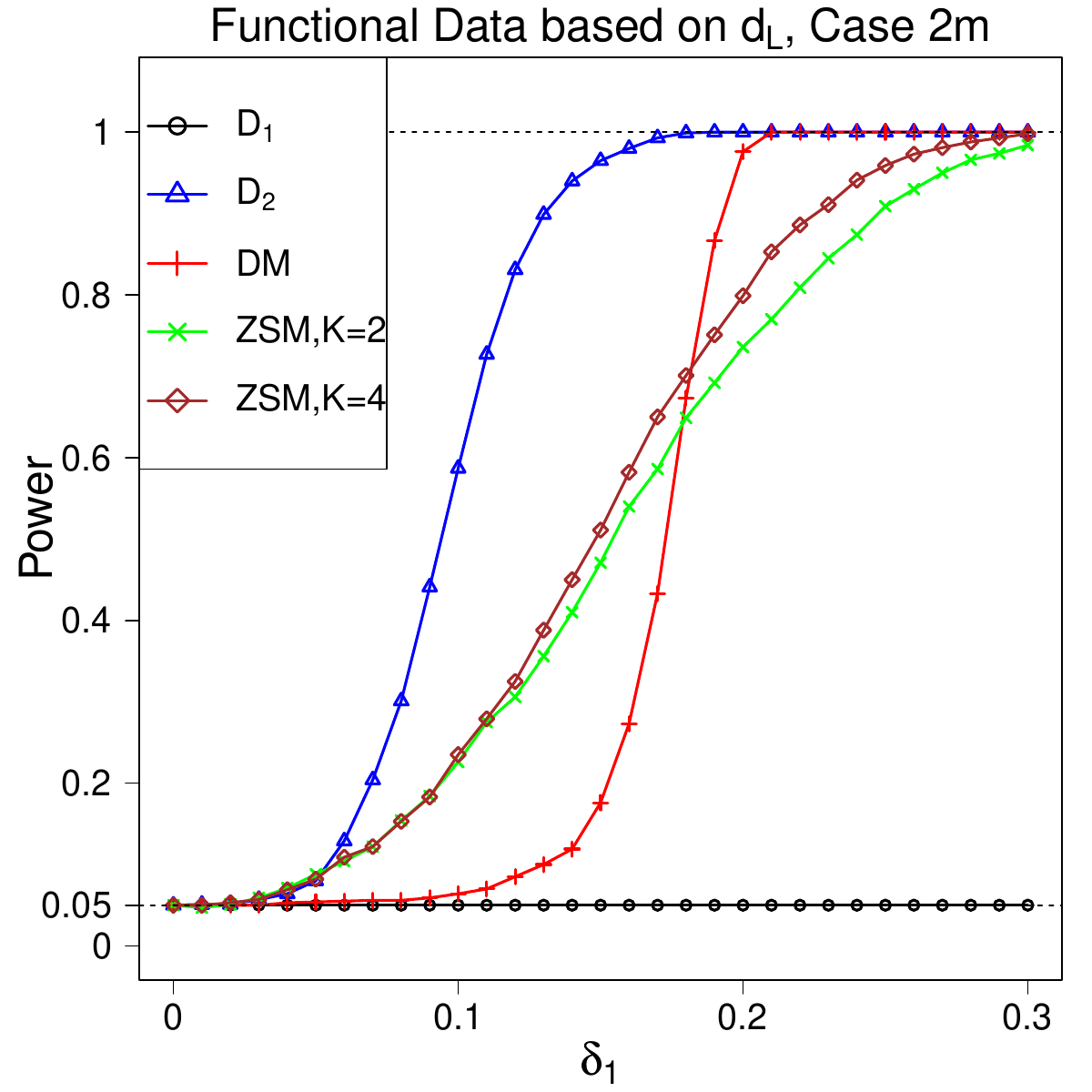}
	\end{subfigure}
	\begin{subfigure}{0.32\textwidth}
		\centering
		\includegraphics[width=1\textwidth]{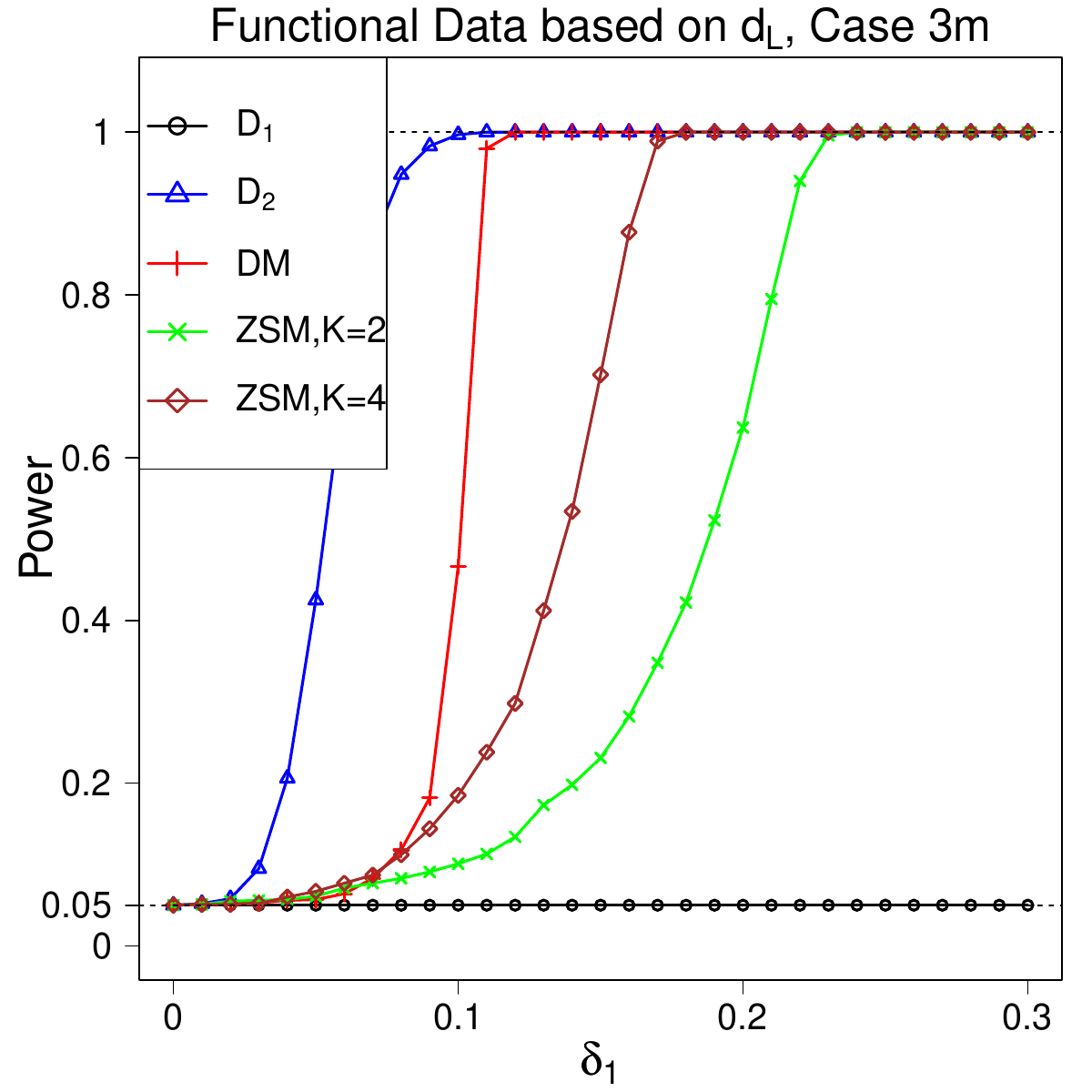}
	\end{subfigure}
	\begin{subfigure}{0.32\textwidth}
		\centering
		\includegraphics[width=1\textwidth]{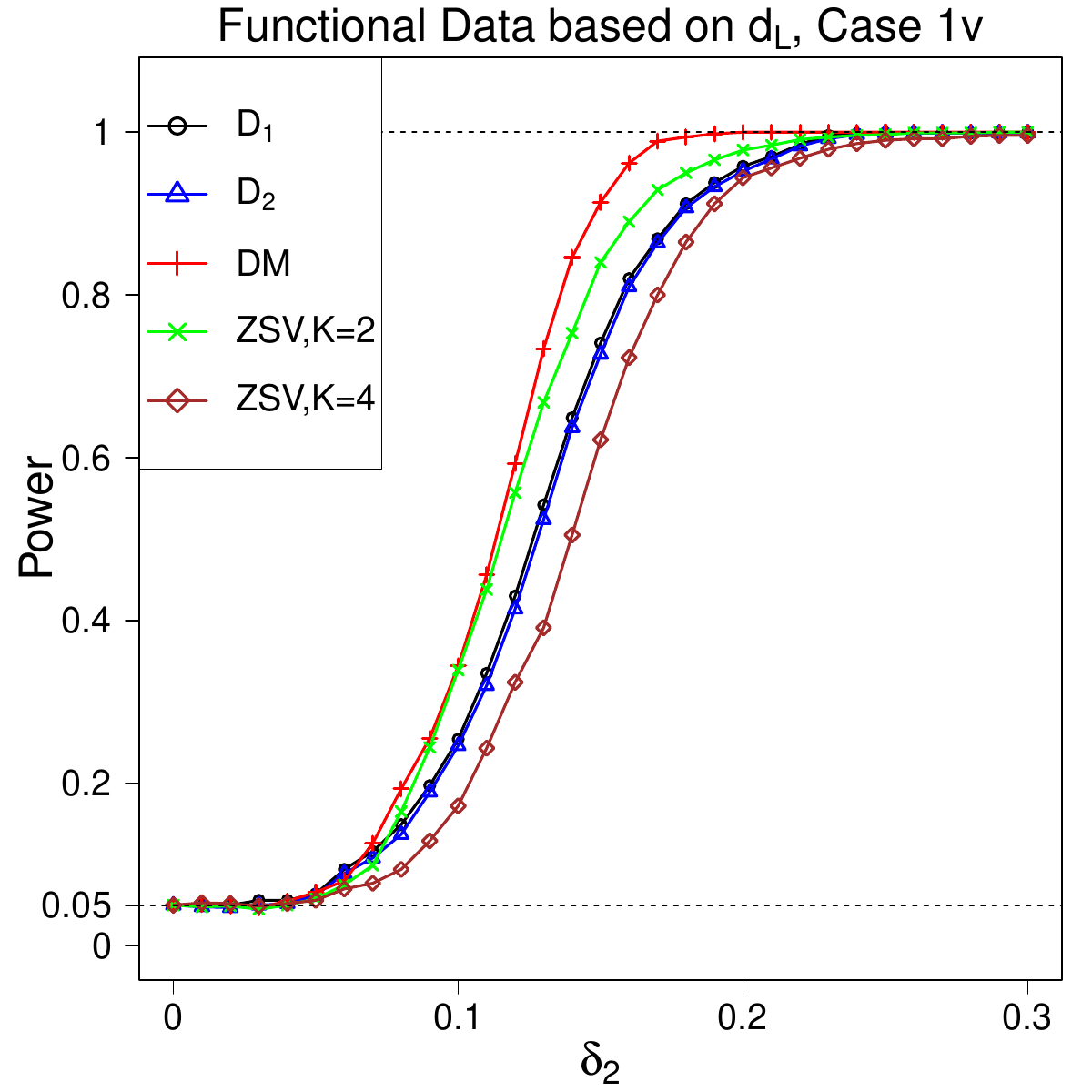}
	\end{subfigure}
	\begin{subfigure}{0.32\textwidth}
		\centering
		\includegraphics[width=1\textwidth]{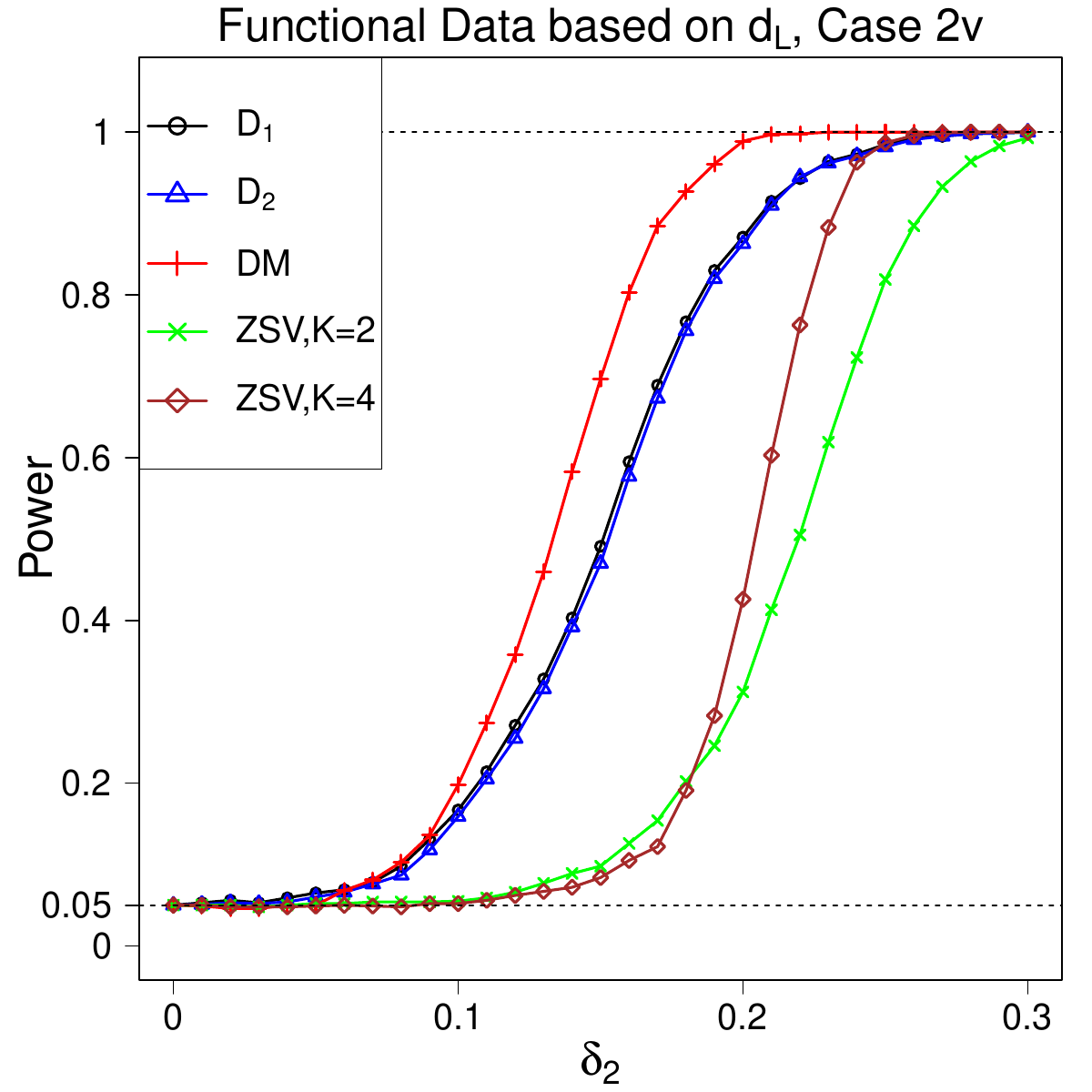}
	\end{subfigure}
	\begin{subfigure}{0.32\textwidth}
		\centering
		\includegraphics[width=1\textwidth]{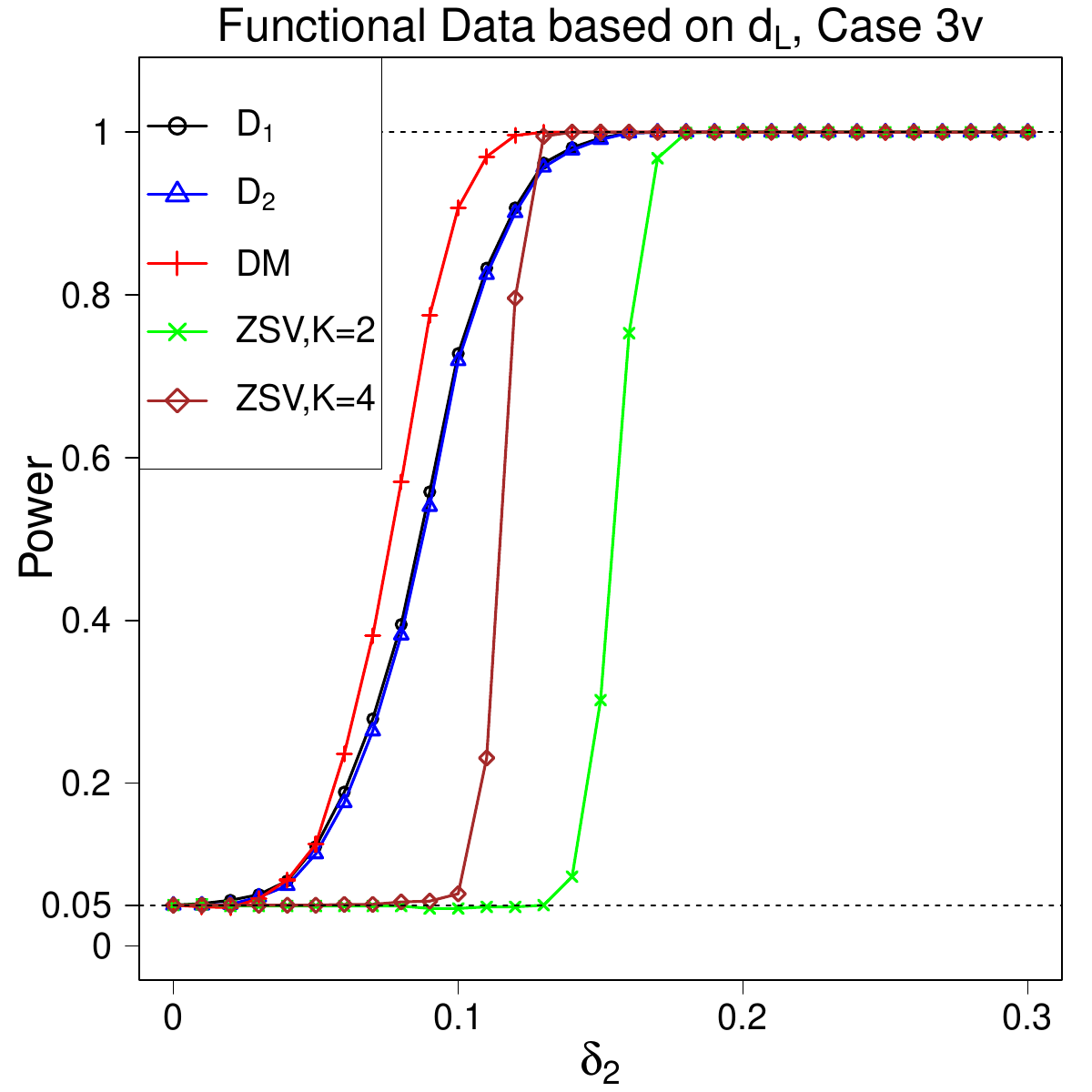}
	\end{subfigure}
	\caption{Size-Adjusted Power $(\times 100\%)$ at $\alpha=5\%$, $n_i=400$ and $\rho=0.4$.  Upper panel: mean difference; bottom panel: covariance operator difference.}
	\label{Fig:functional}
\end{figure}

\bibliographystyle{agsm}
\bibliography{reference}






\end{document}